\documentclass{article}
\usepackage{booktabs} 
\usepackage[ruled]{algorithm2e} 
\usepackage{csquotes}
\usepackage{thmtools} 
\usepackage{thm-restate}
\usepackage[shortlabels]{enumitem}
\usepackage{amsmath}
\usepackage{amsthm}
\usepackage{amssymb}
\usepackage[colorlinks=true, urlcolor=blue]{hyperref}
\usepackage{cleveref}
\usepackage{fullpage}
\usepackage[square,numbers]{natbib}
\usepackage{graphicx}
\usepackage{tikz}
\usepackage{authblk}

\SetAlFnt{\small}
\SetAlCapFnt{\small}
\SetAlCapNameFnt{\small}
\SetAlCapHSkip{0pt}
\IncMargin{-\parindent}

\DeclareMathOperator{\plurality}{\mathrm{Plurality}}

\theoremstyle{theorem}

\newtheorem{lemma}{Lemma}
\newtheorem{theorem}{Theorem}

\theoremstyle{definition}
\newtheorem{definition}{Definition}

\newtheorem{corollary}{Corollary}

\title{Replicating Electoral Success\thanks{Extended version of AAAI '25 paper. This work was performed while K.T.\ and T.N.\ were at Cornell University.}}

\author[1]{Kiran Tomlinson\thanks{kt@cs.cornell.edu}}
\author[2]{Tanvi Namjoshi}
\author[3]{Johan Ugander}
\author[4]{Jon Kleinberg}
\affil[1]{Microsoft Research}
\affil[2]{Princeton University}
\affil[3]{Stanford University}
\affil[4]{Cornell University}

\date{}

\begin{document}

\maketitle

\begin{abstract}
A core tension in the study of plurality elections is the clash between the classic Hotelling--Downs model, which predicts that two office-seeking candidates should position themselves at the median voter's policy, and the empirical observation that real-world democracies often have two major parties with divergent policies. 
Motivated in part by this tension and drawing from bounded rationality, we introduce a dynamic model of candidate positioning based on a simple behavioral heuristic: candidates imitate the policy of previous winners.
The resulting model is closely connected to evolutionary replicator dynamics and, despite its simplicity, exhibits complex behavior and contrasts considerably with existing modeling approaches. 
For uniformly-distributed voters, we prove in our model that when there are $k = 2$, $3$, or $4$ candidates per election, any symmetric candidate distribution converges over time to a concentration of candidates at the center. 
With $k \ge 5$ or more candidates per election, however, we prove that the candidate distribution does not converge to the center. For initial distributions of $k \ge 5$ candidates without any extreme candidates, we prove a stronger statement than non-convergence, showing that the density in an interval around the center goes to zero. 
As a matter of robustness, our conclusions are qualitatively unchanged (though require different analyses) if a small fraction of candidates are not winner-copiers and are instead positioned uniformly at random in each election. 
Beyond our theoretical analysis, we illustrate our results in extensive simulations; for five or more candidates, we find a tendency towards the emergence of two clusters, a mechanism suggestive of Duverger's Law, the empirical finding that plurality leads to two-party systems.
Our simulations also explore several variations of the model, including non-uniform voter distributions, other forms of noise, and replication with memory of earlier rounds of elections. 
In these simulated variants, we find the same general pattern: convergence to the center with four or fewer candidates, but not with five or more. 
Finally, we discuss the relationship between our replicator dynamics model and prior work on strategic equilibria of candidate positioning games.
\end{abstract}

\section{Introduction}

In a democracy, election outcomes determine the trajectory of public policy. A central question in the study of elections is therefore whether we can model which policies are electorally successful. However, elections are extremely complex, with layered interactions between voters, candidates, and the incentives that guide them. To understand the principles governing elections, we therefore need to pare down this complexity and begin with simple models. The literature around this topic traces its roots to \citet{hotelling1929stability} and \citet{downs1957economic}. In the Hotelling--Downs model, candidates compete for election in a one-dimensional policy space. Under the assumption that voters prefer candidates closer to them in policy space, two rational office-seeking candidates will adopt the policy of the median voter, since any other position receives strictly fewer votes. Thus, the central prediction of the Hotelling--Downs model is that we should expect candidates to espouse near-identical moderate policies; in economic contexts, this is often called the \emph{principle of minimum differentiation}~\cite{eaton1975principle,de1985principle}. However, this is not what we observe in modern democracies: countries using plurality often have two dominant parties with markedly different policies~\cite{poole1984polarization,grofman2004downs,riker1982two}. Decades of research have attempted to reconcile this observed policy divergence with the intuitive arguments of Hotelling and Downs~\cite{grofman2004downs,osborne1995spatial}, postulating additional factors like the threat of third-party entry~\cite{palfrey1984spatial} or policy- rather than office-motivated candidates~\cite{wittman1983candidate}. Subsequent research has also expanded beyond two-candidate analysis to consider $k$-candidate elections~\cite{cox1987electoral}.

  The majority of this work has continued under the traditional assumption that candidates are rational and able to make strategically optimal decisions. However, the growing literature on bounded rationality~\cite{simon1955behavioral,simon1979rational} and decision-making heuristics~\cite{tversky1974judgment}, as well as the complexity of elections, casts doubt on whether this is likely in practice. In  a notable exception to the literature on rational candidate positioning, \citet{bendor2011behavioral} argue that heuristics play a crucial role in electoral strategy:

\begin{displayquote}
Campaigns are of chess like complexity---worse, probably; instead of a fixed board, campaigns are fought out on stages that can change over time, and new players can enter the game. Hence, cognitive constraints (e.g., the inability to look far down the decision tree, to anticipate your opponent's response to your response to their response to your new ad) inevitably matter. [...] Thus, political campaigns, like military ones, are filled with trial and error. A theme is tried, goes badly (or seems to), and is dropped. The staff hurries to find a new one, which seems to work initially and then weakens. A third is tried, and then a fourth. [...] \emph{In short, there are good reasons for believing that the basic properties of experiential learning---becoming more likely to use something that has worked in the past and less likely to repeat something that has failed---hold in presidential campaigns.}~\cite[emphasis ours]{bendor2011behavioral}
\end{displayquote}

\begin{figure}[t]
 \includegraphics[width=\textwidth]{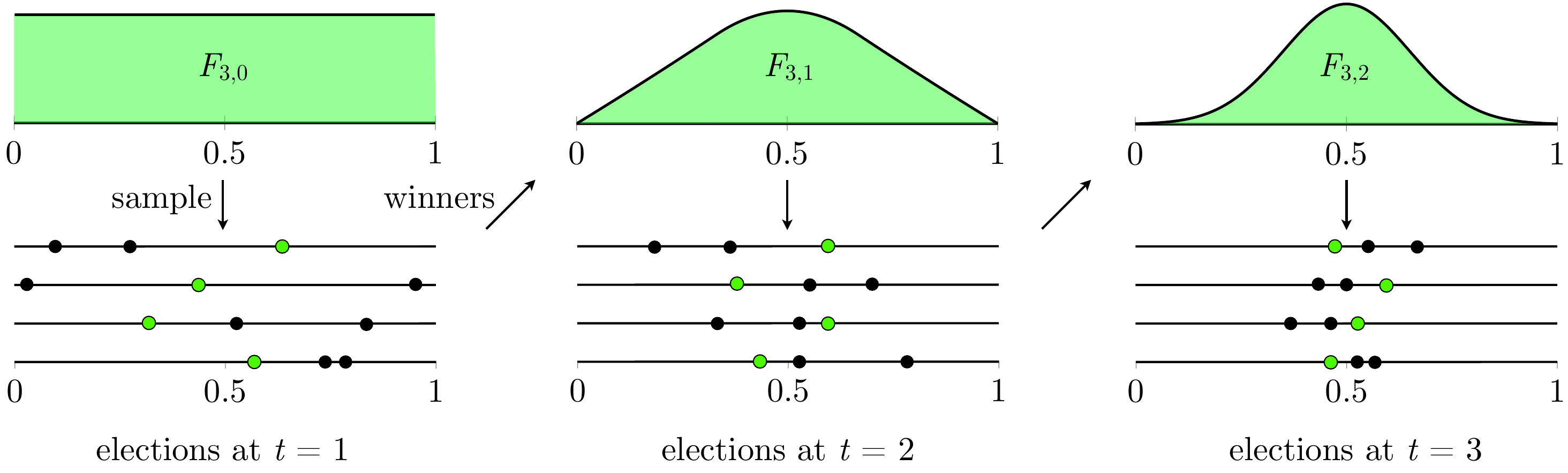}
  \caption{Replicator dynamics for candidate positioning with $k=3$ candidates per election. The top row shows the winner distributions $F_{k, t}$ for each generation $t$, starting from a uniform distribution at $t = 0$,  while the bottom row shows four example elections per generation. In each generation, candidates sample their positions from the winner distribution from the previous generation. Plurality winners (with voters uniform over $[0, 1]$) are indicated in green.}
  \label{fig:overview}
\end{figure}

\paragraph{Our model.} In this paper, we introduce a
model of candidate positioning based on the above heuristic:
candidates imitate success. We focus on plurality elections, where
each voter casts one vote and the candidate with the most votes wins.
We  assume voters have 1-Euclidean
preferences~\cite{coombs1950psychological,elkind2016preference}, where
voters and candidates occupy points in the unit interval $[0, 1]$ and
voters prefer closer candidates. To represent a large voting
population, our model uses a continuum of voters and continuous vote
shares rather than discrete counts. Diverging from prior work, we
model a large number of $k$-candidate elections that proceed in
generations rather than an individual election or election sequence.
In each generation, we assume that candidates copy the policy position
of a winner from the previous generation, a simple heuristic in line
with \citeauthor{bendor2011behavioral}'s suggestion that candidates
use strategies that worked in the past. This heuristic is also
supported by a wealth of political science research arguing that the
imitation of policies, especially electorally successful ones, is a
major feature of
politics~\cite{shipan2008mechanisms,bohmelt2016party,ezrow2021follow}.
As with voters, our model uses a continuous distribution of candidate
positions in each generation, which can be viewed as either capturing
the expected behavior of a finite number of elections or as the
infinite-election limit.

This simple assumption about
candidate behavior (sample a position from the distribution of winners
in the last election cycle) yields a mathematical model equivalent to
the well-studied \emph{replicator dynamics} from evolutionary
biology~\cite{taylor1978evolutionary,schuster1983replicator}, which
have also found widespread use in
economics~\cite{safarzynska2010evolutionary,nelson2018modern}. In the
classic replicator dynamics, $n$ strategies (or alleles) compete in a
population, increasing in prevalence at a rate proportional to their
average fitness in pairwise contests drawn from the current
population. Our model arises from taking such dynamics and moving to a
continuous strategy space with $k$-way interactions in discrete time
(i.e., $k$-candidate elections), treating the plurality win
probability as fitness; we therefore refer to it as \emph{replicator
dynamics for candidate positioning}. 

In summary, then, our model operates in a sequence of generations;
each generation involves a large number of identically distributed
elections, and the candidates in a given generation are drawn from
the distribution of winners of the previous generation's elections.
\Cref{fig:overview} provides a schematic 
visualization of the process with $k = 3$ candidates.
While our model is phrased in terms of a large population of
elections---just as the classic replicator dynamics models a
population of organisms---there is a deep connection between
replicator dynamics and reinforcement
learning~\cite{borgers1997learning,bloembergen2015evolutionary}, so
our conclusions are likely to generalize to models of individual-level
trial-and-error.

\begin{figure}[t]
  \centering
  \includegraphics[width=\textwidth]{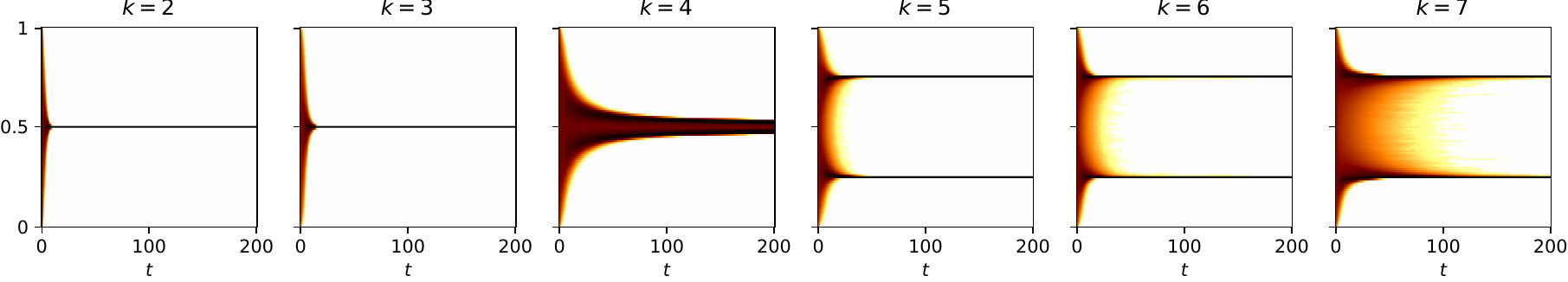}
  \caption{Replicator dynamics runs for $k = 2, \dots, 7$ and 200 generations. Each plot shows 50 runs layered on top of each other, where each run simulates 100,000 elections per generation. We also use \emph{enhanced symmetry}, a trick to keep the symmetry of the analytical model by reflecting copied points across $1/2$ (discussed further in \Cref{sec:simulations}). Darker regions indicate higher candidate density; we use a log-scaled colormap to make low-density regions visible. As our theory establishes, the candidate distribution converges to the center for $k = 2, 3, 4$, but does not for $k \ge 5$. The convergence is very fast for $k = 2$ and $3$, but much slower for $k = 4$. }
  \label{fig:main-sim}
\end{figure}

\paragraph{Our results.} Our main technical contributions characterize the long-run behavior of the replicator dynamics for different values of $k$, the number of candidates per election. We find a dramatic qualitative change in the dynamics as the number of candidates $k$ increases. For our analysis, we focus on the case in which the initial distribution of candidates is symmetric and has a continuous CDF and that voters are uniformly distributed over $[0, 1]$, but we find evidence in simulation that the same patterns hold with other symmetric voter distributions. When $k=2$, we prove that the candidate distribution converges to a point mass at $1/2$ under the replicator dynamics, just like rational candidates in the Hotelling--Downs model. However, we also prove in our model that the candidate distribution converges to the center for $k = 3$ and $4$, in stark contrast to three- and four-candidate extensions of the Hotelling--Downs argument~\cite{cox1987electoral}. Given the behavior for $k = 2, 3$, and $4$, one might be tempted to hypothesize that the replicator dynamics always cause the candidate distribution to converge to the center. Surprisingly, we prove that the pattern ends there: for any $k \ge 5$, we show that the candidate distribution does not converge to $1/2$. See \Cref{fig:main-sim} for simulations demonstrating the patterns that we characterize theoretically. These simulations reveal a tendency for candidate counts larger than $4$ to result in two distinct clusters of policies (around $1/4$ and $3/4$ with uniform voters). This is strongly reminiscent of \emph{Duverger's Law}~\cite{duverger1959political,riker1982two}, the observation that plurality elections tend towards two-party systems; it is striking that it emerges here from a model that does not include any explicit reward for clustering at points away from the center or any mechanisms like the threat of third-party candidates~\cite{bol2016endogenous}.

To strengthen this characterization of the long-run replicator dynamics, we show that our convergence results are robust to noise: even if a small fraction of candidates position themselves uniformly at random, we can still show (approximate) convergence to the center for $k=2, 3, 4$ and non-convergence for $k \ge 5$. While we are not able to theoretically derive the asymptotic distribution for $k \ge 5$ in general, we show that when the initial candidate distribution is supported only on $(1/4, 3/4)$, the candidate density in an interval around $1/2$ goes to 0. Additionally, we explore several variants of the model in simulation, including non-uniform voter distributions, noisy position-copying, memory of prior rounds of elections, and mixtures of candidate counts.\footnote{All of our simulation code and results are available at \url{https://github.com/tomlinsonk/plurality-replicator-dynamics}.} Across these variants, we observe the same general pattern: convergence to the center with up to four candidates, but not with five or more.  For candidate counts $k > 5$ we sometimes see complex and chaotic finite-sample effects in simulation. We conclude by relating our replicator dynamics model back to traditional analyses of Nash equilibria in the style of Hotelling and Downs. The close relationship between replicator dynamics fixed points  and Nash equilibria is well-known~\cite{hofbauer2003evolutionary}, but we argue that ignoring dynamics and focusing only on Nash equilibria leads to brittle conclusions. In particular, we show that different assumptions on voter behavior when candidates occupy the same points lead to dramatically different Nash equilibria than reported in prior work~\cite{cox1987electoral}; in contrast, this choice has no effect on our replicator dynamics results.

To summarize, our main finding is that a simple imitation heuristic can cause candidates to either converge to the median voter or to form two distinct parties, depending on how many candidates run in each election. Intuitively, this phenomenon is driven by the bogeyman of one-dimensional plurality elections: being flanked. If a candidate is stuck between two others, they lose votes from both the left and the right. When there are too many candidates all imitating previous moderate winners, only the leftmost or rightmost of them will avoid being flanked, making more extreme candidates more successful. However, with a small enough pool of opponents, the higher vote share a moderate can receive is worth the risk of ending up stuck between two others. This emerges naturally from our dynamics, without the need for strategic forethought. The surprising fact that falls out of our mathematical analysis is that when candidates are imitators rather than optimizers, the tipping point between the Hotelling--Downs centripetal force and the centrifugal force fueled by the problem of flanking occurs between four- and five-candidate elections.

\subsection{Related work}
Before diving into our theoretical analysis, we briefly summarize the literature in relevant areas.

\paragraph{One-shot candidate positioning games.}
Expanding on the two-candidate Hotelling--Downs foundation, subsequent research has explored higher-dimensional spaces~\cite{plott1967notion,irmen1998competition}, more than two candidates~\cite{cox1987electoral}, policy motivation~\cite{wittman1983candidate}, uncertainty about voter positions~\cite{calvert1985robustness}, and candidate valence (i.e., charisma or name recognition)~\cite{groseclose2001model,bruter2010uncertain}, among many other variations (see \citet{osborne1995spatial,kurella2017evolution} for surveys).\footnote{Hotelling framed the game in terms of two shops positioning themselves along a line (or the design of two competing goods along a single axis), while Downs applied the idea to plurality elections. The two motivations yield equivalent models, so some of the papers we cite use the language of facility location or product design rather than candidate positioning.} Some models allow a third-party candidate to enter the race after the established candidates select their positions, which can lead to non-central two-party equilibria~\cite{palfrey1984spatial,weber1992hierarchical,bol2016endogenous}. 

\paragraph{Dynamic models of candidate positioning.}
In addition to the work on one-shot games, there is also a literature on candidate positioning dynamics~\cite{duggan2017political}, although in contrast to our work, the focus of this literature has been on rational two-candidate contests. One notable early paper in this line of work studies a two-party system where the party which lost the previous election is allowed to reformulate its policy to maximize votes in the next election, which can yield predictable trajectories even in higher-dimensional policy spaces~\cite{kramer1977dynamical}. As in the one-shot literature described above, extensions of this model of two-party dynamics have added a variety of features, including policy motivation~\cite{wittman1977candidates,chappell1986policy}, forward-looking parties~\cite{rosenthal1982model,forand2014two,nunnari2017dynamic}, and---most closely related to our work---boundedly-rational candidates who are unable to exactly optimize their positions~\cite{kollman1992adaptive,kollman1998political,bendor2011behavioral}. Our paper is set apart from this prior research on electoral dynamics with bounded rationality in our replicator dynamics approach, and our success deriving analytical results for more than two candidates.
We are aware of one paper~\cite{laslier2016opportunist} combining a spatial model of elections and replicator dynamics, but the number of parties is fixed to two and the focus is instead on competition between office- and policy-motivated party members (``opportunists'' and ``militants''), where opportunists may defect to the other party.

\paragraph{Evolutionary game theory and replicator dynamics.}
Replicator dynamics~\cite{taylor1978evolutionary,schuster1983replicator,hofbauer2003evolutionary} were introduced to study the evolution of biological populations, but have since found much broader use.
Economists have used evolutionary models---including replicator dynamics---to understand investment behavior~\cite{blume1992evolution}, technological innovation~\cite{saviotti1995competition,safarzynska2011beyond}, 
and resource harvesting~\cite{noailly2003evolution}, among many other phenomena. See \citet{friedman1991evolutionary} for an introduction to evolutionary game theory from an economic perspective and \citet{nelson2018modern,safarzynska2010evolutionary} for surveys of evolutionary economics. Evolutionary models can even be justified without population-level evolution: models of individual-level learning can give rise to behavior equivalent to replicator dynamics~\cite{borgers1997learning}; see \citet{bloembergen2015evolutionary} for a survey of the connection between replicator dynamics and reinforcement learning. Evolutionary models are much less common in political science than in economics, but have been used to model the corruption of elected officials~\cite{accinelli2017controls}, coordination by voters~\cite{mebane2005partisan}, and party defection~\cite{laver2003evolution}. Extensions of the classical replicator dynamics have explored the various modifications found in our model, including multi-way interactions~\cite{gokhale2010evolutionary}, discrete time~\cite{losert1983dynamics}, and a continuous strategy space~\cite{oechssler2001evolutionary,cheung2016imitative}.

\paragraph{Elections with strategic voters.} Another line of research around strategic aspects of elections focuses instead on the strategic choices made by \emph{voters} rather than candidates~\cite{desmedt2010equilibria,thompson2013empirical,obraztsova2016trembling} (with some papers combining both voter and candidate strategy~\cite{feddersen1990rational,myerson1993theory}). Dynamics have featured prominently in the strategic voting literature~\cite{dekel2000sequential,callander2007bandwagons}---in particular, under the framework of \emph{iterative voting}~\cite{meir2010convergence,lev2012convergence,obraztsova2015convergence}, where voters are allowed to update their votes in successive rounds until they are satisfied. Intriguingly, this style of voter-dynamics analysis can also produce conclusions paralleling Duverger's Law, where two major candidates emerge, despite using a completely different approach to ours~\cite{meir2014local}. Evolutionary dynamics have also been applied to voter behavior to explain the paradox of voting (why do people vote when their probability of affecting the outcome is near zero?)~\cite{sieg1995evolutionary}.

\section{Replicator dynamics for candidate positioning}\label{sec:main-theory}

We now formally introduce our model. We consider a one-dimensional policy space represented by the unit interval $[0, 1]$. Candidates and voters reside at points in the interval. To model a large population of voters, we treat the voting population as a continuum; for our theoretical analysis, we assume voters are uniform over $[0, 1]$, but we later relax this assumption in simulation. We assume voters have 1-Euclidean preferences~\cite{elkind2016preference}---that is, they vote for the closest candidate. The \emph{vote share} of a candidate $i$ is the fraction of voters who vote for $i$. With uniform voters, the vote share of a candidate is equal to half the distance between the candidates to its left and right (a candidate adjacent to a boundary gets the entire vote share on its boundary side). Under plurality voting, the candidate with the largest vote share wins; in the case of tied maximum vote shares, the tie is broken uniformly at random. 

Our replicator dynamics model of candidate positioning supposes that elections proceed in generations $t = 1, 2, \dots$, with (infinitely) many elections per generation. We assume the number of candidates in each election is fixed at $k$ (later, we relax this assumption in simulation). The core idea of our model is that candidates in generation $t$ chose their policy positions by copying the position of a winner from the previous generation $t-1$. More formally, let $F_0$ be the initial candidate distribution and let $F_{k, t}$ denote the distribution of winner positions in generation $t$ with $k$ candidates per election. We define $F_{k, 0} = F_0$ for all $k$, although we typically write $F_0$ since the initial distribution does not depend on $k$. In generation $t$, each election consists of $k$ candidates with positions $X_{1,t}, \dots, X_{k, t}\sim F_{k, t-1}$. We use $F_{k, t}(x)$ to denote the CDF of the winner distribution in generation $t$ and $f_{k, t}(x)$ to denote the PDF. Let $\plurality(X_{1,t},\dots, X_{k, t})$ be the position of the plurality winner given candidate positions $X_{1,t},\dots, X_{k, t}$ and uniformly distributed voters.
\begin{definition}
Given an initial candidate distribution $F_0$ and a candidate count $k$, the \emph{replicator dynamics for candidate positioning} (under plurality with uniform 1-Euclidean voters) are, for all $t > 0$,
  \begin{align}
  &F_{k, t}(x) = \Pr( \plurality(X_{1,t},\dots, X_{k,t}) \le x),\label{eq:cdf-iteration}\\
  & X_{i,t}\sim F_{k, t-1},\ \forall i = 1, \dots, k.\notag
\end{align} 
Or, in terms of the PDF:
\begin{align}
  &f_{k, t}(x) = k \cdot \Pr(\plurality(x, X_{2,t}, \dots, X_{k,t}) = x) \cdot f_{k, t-1}(x).\label{eq:pdf-iteration}
\end{align}
\end{definition}

This model can be viewed through the lens of evolutionary replicator dynamics~\cite{taylor1978evolutionary,schuster1983replicator,hofbauer2003evolutionary}, although there are several differences from the classical case. In the classic replicator dynamics, there are $n$ discrete strategies, each of which increases in frequency proportionally to how well that strategy performs against the current population. This proportionality is exactly what \Cref{eq:pdf-iteration} captures: strategy $x$ increases in density proportional to its plurality win rate against the current population. 

The main question we study is how the candidate distribution evolves over time under the replicator dynamics. We focus on cases where $F_0$ is symmetric about $1/2$ and contains no point masses (i.e., the initial CDF $F_0(x)$ is continuous); we call such distributions \emph{symmetric} and \emph{atomless}. This ensures that the probability multiple candidates share the exact same point is 0, so we can ignore these cases for now. Since we assume $F_0$ is symmetric, all subsequent winner distributions are also symmetric by the symmetry of plurality with a uniform voter distribution---we lean heavily on this fact in our analysis. Some of our results require an additional assumptions on $F_0$. We say $F_0$ is \emph{positive near 1/2} if $F_0(x) < 1/2$ for all $x < 1/2$ (equivalently, $f_0(x) > 0$ in an interval around $1/2$); the symmetry of $F_0$ allows us to phrase definitions like this in terms of the left half of the unit interval, and it then applies equivalently to the right half as well.
We define $\mathcal F$ to be the set of all symmetric and atomless distributions over $[0, 1]$ and $\mathcal F ^+ \subset \mathcal F$ to be the subset of such distributions which are also positive near $1/2$.

In this section, we prove our main result piece-by-piece.

\begin{theorem}\label{thm:main}
  Let $F_0 \in \mathcal F^+$. For $k \in \{2, 3, 4\}$, the candidate distribution converges to a point mass at $1/2$ under the replicator dynamics. In contrast, for $k \ge 5$, the candidate distribution does not converge to a point mass at $1/2$.
\end{theorem}

\Cref{thm:main} follows from \Cref{thm:k-2-convergence,thm:k-3-convergence,thm:k-4-convergence,thm:large-k-no-convergence}.  Our results for $k \in \{2, 3, 4\}$ give fine-grained characterizations of the dynamics, which imply convergence to the center: for $k = 2$, we derive a closed form for the CDF at generation $t$ (\Cref{thm:k-2-convergence}), while for $k = 3$ and $4$ we derive closed-form bounds for the CDF (\Cref{thm:k-3-convergence,thm:k-4-convergence}). The negative portion of \Cref{thm:main} offers less insight into the dynamics for $k \ge 5$, only showing non-convergence to the center (\Cref{thm:large-k-no-convergence}), but in \Cref{sec:1/4-3/4} we prove a stronger result in the special case where $F_0$ has no extreme candidates. All proofs omitted from this section for the sake of readability can be found in \Cref{app:main-proofs}.

\subsection{$k = 2$}

Two-candidate plurality with symmetric voters is simple: whichever candidate is closer to $1/2$ has the larger vote share and wins. This simplicity allows us to fully characterize the dynamics with $k=2$. In particular, we derive a closed form for the CDF $F_{2, t}(x)$. 
\begin{theorem}\label{thm:k-2-convergence}
  Let $F_0 \in \mathcal F $. For all $x < 1/2$ and $t \ge 0$,
  $
    F_{2, t}(x) = \left[2 \cdot F_{0}(x)\right]^{2^t}/2.
  $
\end{theorem}
\begin{proof}
  Let $x < 1/2$. Since the candidate closer to $1/2$ wins with $k=2$, $\plurality(X_{1,t}, X_{2,t})\notin (x, 1-x)$ if and only if both $X_{1,t} \notin (x, 1-x)$ and $X_{2,t} \notin (x, 1-x)$, which occurs with probability $(2\cdot  F_{2, t-1}(x))^2$. By symmetry, we then have $F_{2, t}(x) = \Pr(\plurality(X_{1,t}, X_{2,t}) \le x ) = (2\cdot  F_{2, t-1}(x))^2 / 2 = 2 \cdot F_{2, t-1}(x)^2$. We can now prove the claim by induction on $t$. For the base case $t=0$, $(2\cdot F_{0}(x))^{2^0}/2 = F_{0}(x)$. For the inductive case $t \ge 1$, applying the inductive hypothesis yields:
  \begin{align*}
   F_{2, t}(x) &= 2 \cdot F_{2, t-1}(x) ^2 = 2 \left[(2\cdot F_{0}(x))^{2^{t-1}}/2\right]^2= \left[2\cdot F_{0}(x)\right]^{2^{t}}/2.\qedhere
    \end{align*} 
  \end{proof}
This result shows that for $k = 2$, the CDF at any point $x<1/2$ with $F_{2, 0}(x) < 1/2$ rapidly goes to 0---that is (apart from degenerate initial distributions) the candidate distribution converges toward a point mass at $1/2$. 
\begin{corollary}
  Let $F_0  \in \mathcal F ^+$. For all $x < 1/2$, $\lim_{t \rightarrow \infty} F_{2,t}(x) = 0$.
\end{corollary}
Note that by symmetry, it follows from such a statement that for all $x>1/2$, $\lim_{t \rightarrow \infty} F_{2,t}(x)=1$.

\subsection{$k = 3$}

For $k = 3$, plurality becomes more complex: the winner need not be the closest to $1/2$ (for instance, consider candidates at positioned at $1/3, 1/2,$ and $2/3$). Nonetheless, we can still show that the candidate distribution converges to the center. To do so, we find an upper bound on $F_{3, t}(x)$ which goes to 0. The idea behind the proof is to enumerate cases where a candidate in an inner interval $(x, 1-x)$ wins and add up the probability of these cases, as a function of $F_{k, t-1}(x)$. For example, if there are two candidates in $[0, x)$ and one in $(1/2, 1-x)$, then the candidate in $(1/2, 1-x)$ gets vote share greater than $1/2$ and wins; this case occurs with probability $\binom{3}{2} F_{3, t-1} (x)^2 \cdot (1/2 - F_{3, t-1} (x))$. By symmetry, we can then transform this lower bound on the probability the winner is inside $(x, 1-x)$ into an upper bound on $F_{3, t}(x)$, the probability that the winner is in $[0, x]$. 

\begin{restatable}{theorem}{kthree}
\label{thm:k-3-convergence}
  Let $F_0  \in \mathcal F$. For all $x < 1/2$ and $t > 0$,
  \begin{equation}
    F_{3, t}(x) \le  3/4 \cdot F_{3, t-1}(x) + F_{3, t-1}(x)^3.\label{eq:k-3-map}
  \end{equation} 
This can be written as a looser closed form
    \begin{equation}
    F_{3, t}(x) \le F_{0}(x) \cdot \left[ 3/4 + F_{0}(x)^2\right]^t.\label{eq:k-3-closed-form}
  \end{equation} 
\end{restatable}

This result reveals that the candidate distribution for $k=3$ also converges rapidly to the center. In particular, \eqref{eq:k-3-closed-form} shows that the CDF at any point $x$ with $F_0(x) < 1/2$ decays exponentially towards 0 in $t$. This can also be seen by analyzing the cubic iterated map suggested by the upper bound \eqref{eq:k-3-map}, which converges to a stable fixed point at 0 for all initial values in $[0, 1/2)$.

\begin{corollary}
  Let  $F_0  \in \mathcal F ^+$. For all $x < 1/2$, $\lim_{t \rightarrow \infty} F_{3,t}(x) = 0$.
\end{corollary}

\subsection{$k = 4$}
As with $k=3$, we derive an upper bound on the CDF which converges to 0. However, the bound suggests convergence is much slower for $k=4$ than for $2$ or $3$ (which we will see later confirmed in simulation). The proof follows the same case-enumeration strategy as $k = 3$, but simple cases  only show that the CDF is non-increasing in $t$ for $x \in (1/3, 1/2)$, giving us the following lemma.

\begin{restatable}{lemma}{kfourlemma}
\label{lemma:k-4-1/3-bound}
  Let  $F_0  \in \mathcal F$. For all $x \in (1/3, 1/2)$ and $t \ge 0$, $F_{4, t}(x) \le F_{4, 0}(x)$.
\end{restatable}
By using \Cref{lemma:k-4-1/3-bound}, we can strengthen the case analysis with one additional case that tips the recurrence from breaking even to shrinking exponentially towards 0. However, the base of the exponential depends very strongly on $x$, increasing rapidly towards 1 near $1/2$. 

\begin{restatable}{theorem}{kfour}
\label{thm:k-4-convergence}
Let $F_0  \in \mathcal F$. For all $x \in (1/3, 1/2)$ and $t \ge 0$,
\begin{equation}\label{eq:k-4-bound}
F_{4, t}(x) \le F_{0}(x) \cdot \left[1 -  4  (1/2 - F_{0}(x/3 + 1/3))^3 \right]^t.  
\end{equation}
\end{restatable}

Note that $x/3 + 1/3$ is the point two-thirds of the way from $x$ to $1/2$. As long as $F_{0}(x/3 + 1/3) < 1/2$, which is true for any $x<1/2$ if $F_0$ is positive near $1/2$, this result shows that the CDF left of $1/2$ decays to 0 as $t$ grows. That is, the candidate distribution converges to the center again.

\begin{corollary}
  Let $F_0  \in \mathcal F ^+$. For all $x < 1/2$, $\lim_{t\rightarrow \infty}F_{4, t}(x) = 0$.
\end{corollary}

\subsection{$k \ge 5$}
In contrast to $k=2, 3, 4$, we now show that for any larger $k$, the candidate distribution does not converge to the center. The proof is based on the following observation.

\begin{lemma}\label{lemma:1/4-3/4}
  For any $k$, if all candidates are in $(1/4, 3/4)$, then only the left- or rightmost candidate can win with uniform voters.
\end{lemma}
\begin{proof}
 Suppose all candidates are in $(1/4, 3/4)$. Any candidate between two others gets vote share less than $(1/2)/2 = 1/4$, since no two candidates are distance $1/2$ or greater apart. Meanwhile, the left- and rightmost candidates each get vote share $> 1/4$.
\end{proof}
Intuitively, if the candidate distribution starts converging to the center, then all candidates will likely be inside $(1/4, 3/4)$, at which point only the most extreme candidates can win. When $k$ is sufficiently large (i.e., $\ge 5$), the left- and rightmost candidates are likely on opposite sides and farther from $1/2$ than the average candidate. This results in a centrifugal force preventing further progress towards the center. Formally, we prove the following theorem.
\begin{restatable}{theorem}{largek}
\label{thm:large-k-no-convergence}
  Let $F_0  \in \mathcal F$. For any $k \ge 5$, there exists some $x < 1/2$ such that $\lim_{t \rightarrow \infty} F_{k, t}(x) \ne 0 $. That is, the candidate distribution does not converge to a point mass at $1/2$.
\end{restatable}
More specifically, our proof assumes for a contradiction that the distribution converges to $1/2$, so at some generation $t^*$, the probability mass left of $1/4$ must be less than some small $\alpha$. We then show that the CDF can never decrease at $F_{k, t^*}^{-1}(1/4)$ after generation $t^*$, since all candidates will likely be inside $(1/4, 3/4)$, causing only the most extreme candidates to win. This contradicts convergence to the center.

\section{Replicator dynamics with noise}\label{sec:noisy}
So far, we have assumed that all candidates copy winner positions from the previous generation. We now show that our results still hold in approximate forms if some of the candidates violate this behavior and instead position themselves uniformly at random. This demonstrates a way in which the convergence of the model is robust to alternative specifications. 

\begin{definition}
Given an initial candidate distribution $F_0$, a candidate count $k$, and a noise level $\epsilon \in (0, 1]$, the \emph{replicator dynamics for candidate positioning with $\epsilon$-uniform noise} (under plurality with uniform 1-Euclidean voters) are, for all $t > 0$, 
  \begin{align}
  &F_{k, t}^\epsilon(x) = \Pr( \plurality(X_{1,t}^\epsilon,\dots, X_{k,t}^\epsilon) \le x),\label{eq:noisy-iteration}\\
  &F_{k,0}^\epsilon = F_0\notag\\
  & X_{i,t}^\epsilon\sim \begin{cases}
    \text{Uniform}(0, 1) &\text{w.p.\ $\epsilon$},\\
    F_{k, t-1}^\epsilon &\text{w.p.\ $1-\epsilon$}.
  \end{cases}\notag
\end{align} 
\end{definition} 

As in the noiseless case, we show that the candidate distribution  converges to the center under the dynamics with $\epsilon$-uniform noise for $k = 2, 3, 4$ but do not for $k \ge 5$. However, since $\epsilon$-uniform noise introduces non-central candidates at every $t$, we need to relax the convergence requirement. The idea behind our notion of approximate convergence is that if we make the noise sufficiently small, then the distribution should get arbitrarily close to a point mass at $1/2$. That is, the CDF at any point $x < 1/2$ eventually goes below any positive threshold $c$, for sufficiently small $\epsilon>0$. 

\begin{definition}
Let $F_0 \in \mathcal F$. The candidate distribution \emph{approximately converges to the center} under the replicator dynamics with $\epsilon$-uniform noise if for all $x \in [0, 1/2)$ and $c > 0$, there exists some $\epsilon_\text{max} > 0$ such that if $\epsilon \in (0, \epsilon_\text{max}]$, then
$\limsup_{t \rightarrow \infty}F_{k, t}^\epsilon(x) < c$. 
\end{definition}

We now give the analogue of our main result with $\epsilon$-uniform noise. One additional benefit of adding noise is that we no longer need to assume $F_0$ is positive near $1/2$.

\begin{theorem}\label{thm:main-noisy}
  Let $F_0 \in \mathcal F$. For $k \in \{2, 3, 4\}$, the candidate distribution approximately converges to the center under replicator dynamics with $\epsilon$-uniform noise. In contrast, for all $k \ge 5$, the candidate distribution does not approximately converge to the center.
\end{theorem}

\Cref{thm:main-noisy} follows from \Cref{thm:k-2-noisy-convergence,thm:k-3-noisy-convergence,thm:k-4-noisy-convergence,thm:large-k-noisy-no-convergence}. See \Cref{app:noisy-proofs} for proofs omitted from this section.

\subsection{$k = 2$}
We first show that the replicator dynamics with $\epsilon$-uniform noise approximately converge to the center with two candidates. In fact, we can exactly characterize the limiting candidate distribution for $k=2$. As before with $k=2$, whichever candidate is closer to $1/2$ wins, but now these candidates can either be winner-copiers or randomly positioned. The idea behind the proof is to find an iterated map for $F_{2, t}^\epsilon(x)$ and find the stable fixed point it converges to for $x < 1/2$, which we show is smaller than $\epsilon$. 

\begin{theorem}\label{thm:k-2-noisy-convergence}
  Let $F_0  \in \mathcal F$. For any $\epsilon \in (0, 1)$ and $x  \in [0,  1/2)$ with $\epsilon$-uniform noise,
  \begin{equation}
    \lim_{t \rightarrow \infty} F_{2, t}^\epsilon(x) = \frac{1 - 4x \epsilon(1-\epsilon) -  \sqrt{1 - 8 \epsilon x (1 - \epsilon )}}{4 (1-\epsilon)^2} \le \epsilon.
  \end{equation}
\end{theorem}

\begin{proof}
  Let $x < 1/2$ and $p = F_{2, t-1}^\epsilon(x)$. Each candidate in generation $t$ is drawn from $F_{2, t-1}^\epsilon$ w.p.\ $(1-\epsilon)$ and from $\text{Uniform}(0, 1)$ w.p.\ $\epsilon$ (call such candidates \emph{uniform)}. Uniform candidates fall outside $(x, 1-x)$ w.p.\ $2x$, while non-uniform candidates fall outside $(x, 1-x)$ w.p.\ $2p$ by symmetry.  A winner in generation $t$ is not in $(x, 1-x)$ if and only if both candidates fall outside this interval, which thus occurs with probability
  \begin{align*}
    \Pr(X_{1,t}^\epsilon \notin (x, 1-x), X_{2,t}^\epsilon \notin (x, 1-x))&= \underbrace{(1-\epsilon)^2(2p)^2}_\text{neither uniform} + \underbrace{2 \epsilon (1-\epsilon)(2x) (2p)}_\text{one uniform} + \underbrace{\epsilon^2 (2x)^2}_\text{both uniform}\\
    &= 4p^2 (1-\epsilon)^2 + 8px\epsilon(1-\epsilon) + 4x^2 \epsilon^2.
  \end{align*}
  By symmetry, we then have
  \begin{align}
    F_{2, t}^\epsilon(x) &= \Pr(X_{1,t}^\epsilon \notin (x, 1-x), X_{2,t}^\epsilon \notin (x, 1-x)) /2\notag\\
    &= 2p^2 (1-\epsilon)^2 + 4px\epsilon(1-\epsilon) + 2x^2 \epsilon^2.\label{eq:k-2-noisy-map}
  \end{align}
  The claim then follows from the following technical lemma, proved in \Cref{app:noisy-proofs}.
  \begin{restatable}{lemma}{ktwonoisylemma}
\label{lemma:k-2-noisy-map}
  For all initial $p \in [0, 1/2]$, $\epsilon \in (0, 1)$, and $x \in [0, 1/2)$, the quadratic iterated map
  $
    p' = 2p^2 (1-\epsilon)^2 + 4px\epsilon(1-\epsilon) + 2x^2 \epsilon^2
  $
  converges to the fixed point 
  $
  p^* =   \frac{1 - 4x \epsilon(1-\epsilon) -  \sqrt{1 - 8 \epsilon x (1 - \epsilon )}}{4 (1-\epsilon)^2} \le \epsilon.
  $
\end{restatable}
 
  \end{proof}
  This result implies approximate convergence to the center: we can simply take $\epsilon_\text{max} < c$ and we then have $\lim_{t \rightarrow \infty}F_{2, t}^\epsilon(x) \le \epsilon_\text{max}< c$. 
  \begin{corollary}
    Let $F_0  \in \mathcal F$. For $k = 2$, the candidate distribution approximately converges to the center under replicator dynamics with $\epsilon$-uniform noise.
  \end{corollary}

\subsection{$k=3$}
We now show approximate convergence to the center for $k=3$ with $\epsilon$-uniform noise. As in the noiseless case, we cannot fully characterize the limiting distribution, but we are able to bound it for $\epsilon < 1/3$. The proof repeats the case analysis from the proof of \Cref{thm:k-3-convergence} but with $\epsilon$-uniform candidates, which yields a cubic iterated map. We then bound the attracting fixed point of this map, as we did for $k = 2$. 

\begin{restatable}{theorem}{kthreenoisy}
\label{thm:k-3-noisy-convergence}
  Let  $F_0  \in \mathcal F$. For any $\epsilon \in (0,  1/3)$ and $x \in [0, 1/2)$, $\limsup_{t \rightarrow \infty} F_{3,t}^\epsilon (x) \le 1.5\epsilon$.
\end{restatable}
As with $k=2$, this shows that for any $c > 0$, we can pick a small enough $\epsilon$ (i.e., $\epsilon < \min\{1/3, 2/3 \cdot c\}$) so that $\limsup_{t \rightarrow \infty} F_{3, t}^\epsilon(x) < c$.

\begin{corollary}
      Let $F_0  \in \mathcal F$. For $k = 3$, the candidate distribution approximately converges to the center under replicator dynamics with $\epsilon$-uniform noise.
      \end{corollary}

\subsection{$k = 4$}
As with two and three candidates, we can also show approximate convergence to the center for replicator dynamics with $\epsilon$-uniform noise and four candidates. However, for $k=4$, the bound on  $\epsilon$ required for convergence depends on the point's distance from $1/2$---just as the convergence rate did in the noiseless case. We begin with the noisy analogue of \Cref{lemma:k-4-1/3-bound}. 

\begin{restatable}{lemma}{kfourlemmanoisy}
\label{lemma:k-4-1/3-bound-noisy}
 Let $F_0  \in \mathcal F$. With $\epsilon$-uniform noise, for any $\epsilon \in (0, 1]$, $x \in (1/3, 1/2)$, and $t > 0$, 
    \begin{align*}
      F_{4, t}^\epsilon(x) \le \epsilon x +  (1 - \epsilon) F_{4, t-1}^\epsilon(x).
    \end{align*}
    Thus, $ F_{4, t}^\epsilon(x) \le \max\{x, F_{4, 0}^\epsilon(x)\}$.
\end{restatable}

We can then apply the same strategy as we did in the noiseless case and analyze the resulting iterated map as we have done for $k = 2$ and $3$.

\begin{restatable}{theorem}{kfournoisy}
\label{thm:k-4-noisy-convergence}
Let $F_0  \in \mathcal F$. For any $\epsilon \in (0, 1]$ and $x \in (1/3, 1/2)$, let $\beta = 1/2 - \epsilon(x/3 + 1/3) - (1-\epsilon)\max\{x/3 + 1/3, F_{0}(x/3 + 1/3) \}$.  
Then $\beta \in (0, 1/2]$ and $\limsup_{t\rightarrow \infty}F_{4, t}^\epsilon(x) \le \frac{1}{8\beta^3} \epsilon $.
\end{restatable}

As long as we make $\epsilon$ sufficiently small (relative to $8 \beta ^3$), the CDF at $x < 1/2$ eventually goes below any desired 
threshold $c$---although the closer $x$ is to $1/2$, the smaller $\beta$ becomes, and likewise the required $\epsilon$. 

\begin{corollary}
  Let $F_0  \in \mathcal F$. For $k = 4$, the candidate distribution approximately converges to the center under replicator dynamics with $\epsilon$-uniform noise.
\end{corollary}

\subsection{$k \ge 5$}
Now that we have seen approximate convergence to the center for $k = 2, 3, 4$, we prove that this does not happen for any higher $k$. The argument uses the same idea as in \Cref{thm:large-k-no-convergence} ($k\ge 5$ without noise): that when all candidates are in $(1/4, 3/4)$, only the left- and rightmost candidates can win. As long as we make $\epsilon$ sufficiently small, the exact same approach applies, albeit with some added care to account for randomly positioned candidates. 

\begin{restatable}{theorem}{largeknoisy}
\label{thm:large-k-noisy-no-convergence}
 Let $F_0  \in \mathcal F$. For $k \ge 5$, the candidate distribution does not approximately converge to the center under replicator dynamics with $\epsilon$-uniform noise. 
\end{restatable}

\section{Positive results for $k \ge 5$ with no extreme candidates}\label{sec:1/4-3/4}
In the previous sections, our results for $k \ge 5$ have been negative, showing the candidate distribution does not converge to the center, but without indicating what the distribution converges to instead. While simulations indicate a tendency towards a two-spike equilibrium, we have not been able to theoretically characterize the limiting distribution in general for $k\ge 5$, either with or without noise. 
However, \Cref{lemma:1/4-3/4} enables us to analyze the dynamics for $k \ge 5$ (without noise) in the special case that $F_0$ has no extreme candidates, with support only on $(1/4,3/4)$. In this setting, the dynamics are much simpler, as only the left- and rightmost candidates can win. 
The same type of argument we used before for $k \ge 5$ then provides a positive result, showing that the candidate distribution converges to one with zero mass in an interval around $1/2$. In contrast, our central convergence results for $k \in \{2, 3, 4\}$ still hold in this special case.
Proofs of results in this section can be found in \Cref{app:1/4-3/4-proofs}.
\begin{restatable}{theorem}{boundedsupp}
\label{thm:1/4-3/4}
  Suppose $F_0 \in \mathcal F$ is supported on $(1/4, 3/4)$. Let $\ell = [1-\sqrt{3/7}]/2 = 0.172\dots$. For  $k \ge 5$ and $x \in (F_{0}^{-1}(\ell), 1/2)$, $\lim_{t \rightarrow \infty} F_{k, t}(x) = 1/2.$
\end{restatable}
When $F_0$ is Uniform$(1/4, 3/4)$, note that $F_0^{-1}(0.172\dots) = 0.336\dots$, so \Cref{thm:1/4-3/4} implies that as $t \rightarrow \infty$, the candidate density in $[0.34, 0.66]$ goes to 0 for $k \ge 5$. With no extreme candidates, we can also precisely characterize the density of the candidate distribution at $1/2$ using a simple argument. Since only the left- or rightmost candidates can win, a candidate $i$ at $1/2$ only wins if the other candidates are all on the left or on the right. By symmetry, this occurs with probability $2 \cdot (1/2)^{k - 1} = (1/2)^{k - 2}$. Accounting for the $k$-fold symmetry in choosing candidate $i$ and applying an inductive argument based on \Cref{eq:pdf-iteration} then gives the following result.

\begin{restatable}{theorem}{boundedsuppdensity}
\label{thm:limited-support-density}
Suppose $F_0 \in \mathcal F$ is supported on $(1/4, 3/4)$. For any $k \ge 2$ and $t \ge 0$,
  \begin{equation}\label{eq:limited-support-density}
    f_{k, t}(1/2) = f_{0}(1/2) \cdot \left[k (1/2)^{k-2}\right]^t.
  \end{equation}
\end{restatable}

\begin{table}[t]
\centering
\caption{Base of the exponential from \Cref{thm:limited-support-density} for small $k$.}
\label{tab:limited-support-exponent-base}
  \begin{tabular}{lccccc}
  \toprule
  $k$ & 2 & 3 & 4 & 5 & 6\\
  \midrule
  $k (1/2)^{k-2}$ & 2 & $3/2$ & 1 & $5/8$ & $3/8$\\
  \bottomrule
\end{tabular}
\end{table}

With support on $(1/4, 3/4)$, the behavior of the density at $1/2$ therefore depends on whether $k (1/2)^{k-2}$ is smaller or larger than 1. This quantity is smaller than $1$ for $k \ge 5$, larger than $1$ for $k = 2, 3$ and equal to $1$ for $k = 4$ (see \Cref{tab:limited-support-exponent-base}). The larger $k$ is, the more rapidly the density at $1/2$ goes to 0. This simple argument reveals a mechanism driving the $k < 5$ vs $k \ge 5$ divide: $(1/2)^{k -2}$ is exactly the probability that a central candidate is not flanked. This probability decreases rapidly with $k$ and is counterbalanced at first by the increasing number of candidates $k$ who can be at the center---by as soon as $k \ge 5$, the exponentially low probability of being the left- or rightmost candidate at the center becomes too small.
 
\Cref{thm:limited-support-density} is particularly interesting for $k=4$, since we know the distribution converges to a point mass at $1/2$, but the density at $1/2$ stays constant at $f_0(1/2)$ when $F_0$ is supported on $(1/4, 3/4)$. These seemingly contradictory facts are both possible since the distribution converges by accumulating more and more mass in two spikes on each side of $1/2$ that approach the center arbitrarily closely. \Cref{thm:limited-support-density} thus highlights how $k=4$ is a marginal tipping point which just barely converges to the center---a phenomenon also hinted at by the marginal nature of our $k=4$ case analysis in \Cref{lemma:k-4-1/3-bound} and \Cref{thm:k-4-convergence}: the analysis in \Cref{lemma:k-4-1/3-bound} just breaks even, with the low-probability case in \Cref{thm:k-4-convergence} needed to tip the scales. As we saw in \Cref{fig:main-sim}, this manifests in simulation as slow convergence to the center for $k=4$.

\section{Simulations}\label{sec:simulations}

\begin{figure}[t]
  \centering
  \includegraphics[width=\textwidth]{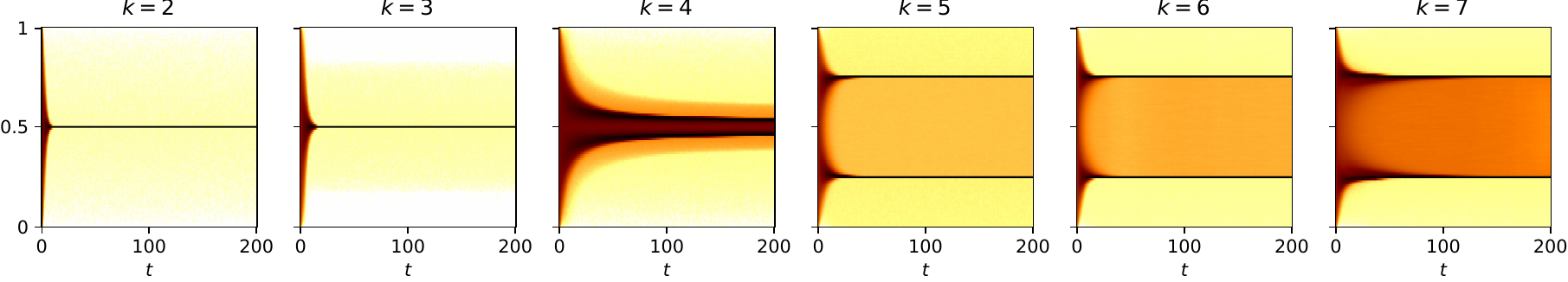}
  \caption{  Replicator dynamics runs with $0.01$-uniform noise for $k = 2, \dots, 7$ and 200 generations, using enhanced symmetry, 50 trials per plot, and 100,000 elections per generation. The behavior is qualitatively identical to the model without noise (\Cref{fig:main-sim}). }
  \label{fig:small-k-symmetry}
\end{figure}

\begin{figure}[t]
  \centering
  \includegraphics[width=\textwidth]{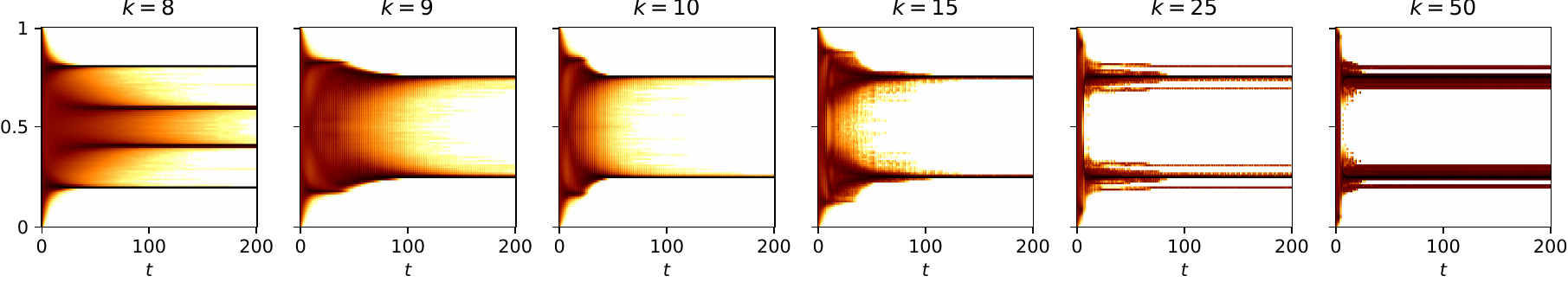}\\
  \includegraphics[width=\textwidth]{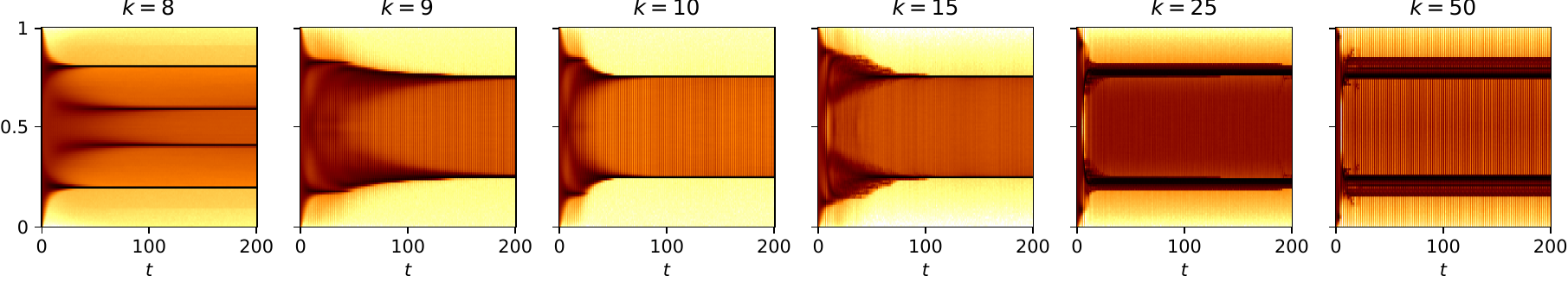}
  \caption{ Replicator dynamics runs with no noise (top row) and $0.01$-uniform noise (bottom row) for larger candidate counts $k$ and using enhanced symmetry. Other settings are identical to \Cref{fig:small-k-symmetry}, with 50 runs shown in each plot. As the theory predicts, the candidate distribution does not converge to the center; but the exact behavior varies.}
  \label{fig:large-k-symmetry}
\end{figure}

Having established our primary theoretical results, we demonstrate them in simulation.\footnote{All of our simulation code and results are available at \url{https://github.com/tomlinsonk/plurality-replicator-dynamics}.} To do so, we use Monte Carlo sampling, simulating a large number of elections per generation (100,000) and using the winners to approximate $F_{k, t}$. We initialize $F_{0}$ to be uniform. With this basic setup, we observe some effects due purely to sampling, such as oscillations due to small asymmetries in the Monte Carlo samples. In contrast, our theoretical model revolves around an evolving density which by definition is always symmetric. To preserve symmetry while maintaining the same evolving distribution, we configure our Monte Carlo sampling to use a trick we term \emph{enhanced symmetry}, mirroring each copied position across $1/2$ with probability $1/2$ in every generation.

\Cref{fig:small-k-symmetry} shows 50 aggregated simulation runs for $k = 2, \dots, 7$ using enhanced symmetry and 0.01-uniform noise (recall \Cref{fig:main-sim} for equivalent plots without noise). See \Cref{app:plots} for additional plots without enhanced symmetry and showing a single trial---these results follow the same general patterns across values of $k$. The candidate distributions evolve exactly as we would expect from our theory: rapid convergence to the center for $k = 2$ and $3$, slow convergence for $k = 4$, and non-convergence to the center for $k \ge 5$; both with and without $\epsilon$-uniform noise. Interestingly, $k = 5, 6,$ and $7$ show a tendency to converge towards two point masses at $1/4$ and $3/4$---but this phenomenon is sensitive to sampling asymmetries for $k = 6$ and $7$ (see \Cref{fig:small-k,fig:small-k-1-trial} in \Cref{app:plots}). In \Cref{sec:1/4-3/4}, we will see some theoretical justification for this two-spike behavior in a special case when $F_{0}$ has no extreme candidates. In \Cref{fig:large-k-symmetry}, we also show simulations with larger candidate counts, from $8$ to $50$. In line with \Cref{thm:large-k-no-convergence}, the distributions with large $k$ do not converge to the center. However, we see some surprising differences depending on $k$; for instance, the asymptotic distribution with $k = 8$ appears to have four clusters rather than two (mysteriously, all other large values of $k$ we have tested tend towards two clusters, at least with enhanced symmetry). In \Cref{app:plots}, we provide several additional visualizations: \Cref{fig:n-50} shows simulations with only 50 elections per generation, demonstrating that our findings still hold in a small-sample setting; \Cref{fig:large-k} shows large-$k$ simulations without enhanced symmetry; finally, \Cref{fig:1/4-3/4} demonstrates the no-extremes setting of \Cref{thm:1/4-3/4}, starting from Uniform($1/4, 3/4$).

\begin{figure}[t]
  \centering
  \includegraphics[width=0.95\textwidth]{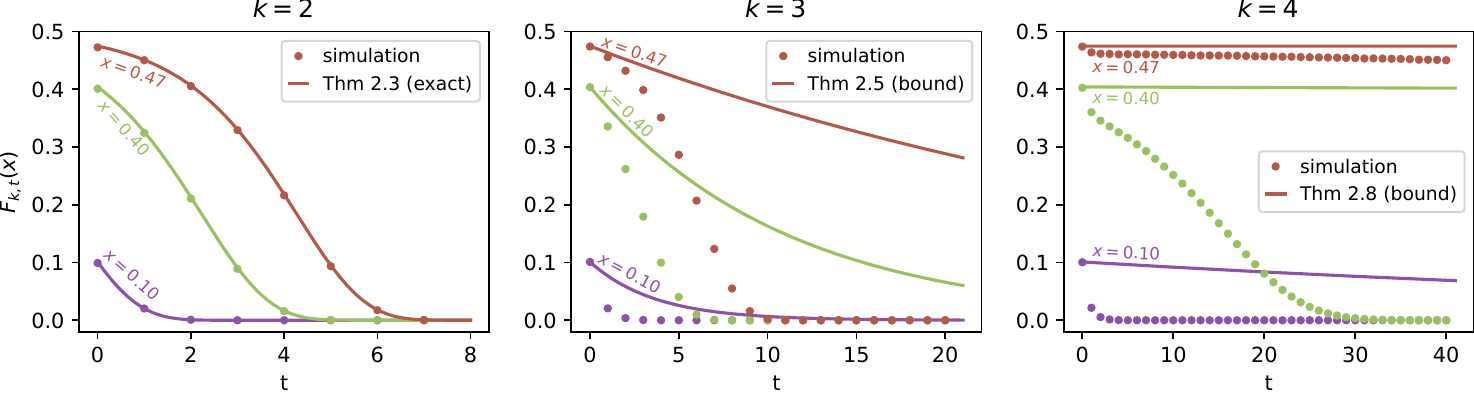}
  \caption{Simulations demonstrating our convergence results \Cref{thm:k-2-convergence,thm:k-3-convergence,thm:k-4-convergence}, showing the simulated candidate distribution CDF at various points $x$ alongside the theoretical predictions. The simulations use 50 trials with 100,000 elections per generation, no noise, and enhanced symmetry. The theorems get progressively weaker: \Cref{thm:k-2-convergence} provides an exact characterization of the two-candidate dynamics, while \Cref{thm:k-3-convergence,thm:k-4-convergence} give upper bounds that converge to 0.}
  \label{fig:pred-vs-sim}
\end{figure}

\begin{figure}[t]
  \centering
  \includegraphics[width=0.95\textwidth]{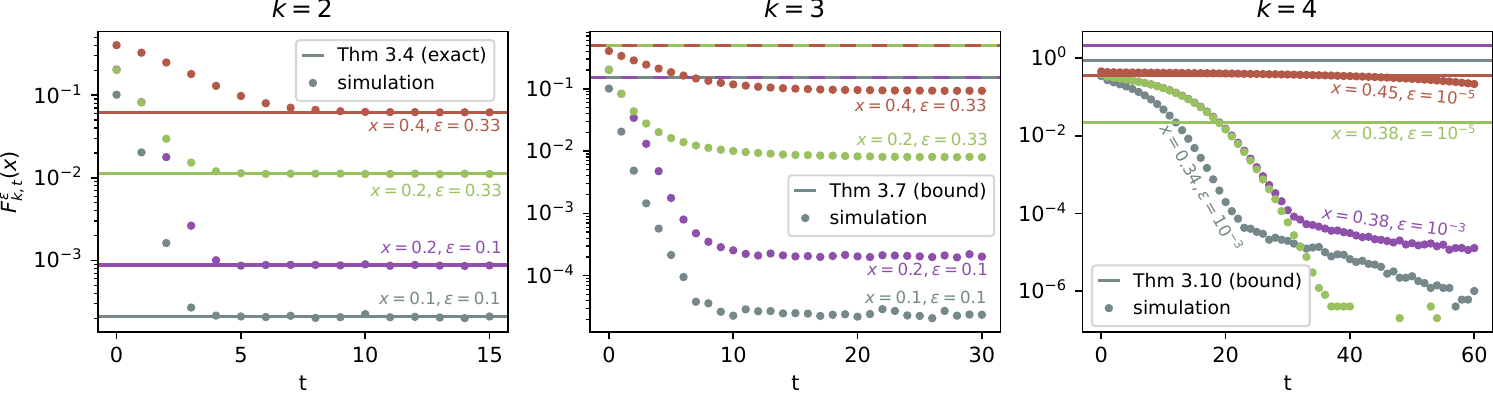}
  \caption{Simulations demonstrating \Cref{thm:k-2-noisy-convergence,thm:k-3-noisy-convergence,thm:k-4-noisy-convergence}, showing the simulated candidate distribution CDF at various points $x$ and various noise levels $\epsilon$ alongside the theoretical asymptotic bounds. The simulations use 50 trials with 100,000 elections per generation and enhanced symmetry. As in the noiseless case, the theorems get progressively weaker as $k$ increases. For $k = 3$, the asymptotic bounds depend only on $\epsilon$, while the bounds for $k = 4$ and the exact limit for $k = 2$ depend on both $\epsilon$ and $x$.}
  \label{fig:pred-vs-sim-noisy}
\end{figure}

In addition to confirming the picture painted by our theory, we also use simulations to explore how tight our bounds are---although our core focus is on characterizing the qualitative behavior of the model rather than achieving the tighest bounds on convergence rate. In \Cref{fig:pred-vs-sim}, we demonstrate the exact result from \Cref{thm:k-2-convergence} and the closed-form upper bounds on the candidate distribution CDF from \Cref{thm:k-3-convergence,thm:k-4-convergence}. The bounds on convergence rates for $k = 3$ and $4$ are indeed loose, as there are several ways that central candidates can win that are not easily captured by our case analysis. In \Cref{fig:pred-vs-sim-noisy}, we demonstrate \Cref{thm:k-2-noisy-convergence,thm:k-3-noisy-convergence,thm:k-4-noisy-convergence}. Note that these results are all asymptotic, characterizing or bounding the limit of the candidate distribution CDF as $t \rightarrow \infty$, whereas the results in \Cref{fig:pred-vs-sim} hold for finite $t$. Additionally, the results with $\epsilon$-uniform noise depend on the value of $\epsilon$, so we experiment with several different values. Again, we see that our bounds from \Cref{thm:k-3-noisy-convergence,thm:k-4-noisy-convergence} are loose, but nonetheless hold and are non-trivial. Moreover, the exact result in \Cref{thm:k-2-noisy-convergence} is nicely confirmed by simulation.

\section{Variants of the replicator dynamics}\label{sec:variants}
We now demonstrate in simulation that the qualitative picture provided by our results from \Cref{thm:main} is robust to different specifications of the model. At a high level, our model of candidate positioning consists of the following components: (1) a fixed voter distribution, (2) a subset of previous candidate positions which new candidates imitate, and (3) a rule for sampling from those previous positions. In the basic model, the voter distribution is uniform, the imitated positions are plurality winners from the previous generation, and the sampling rule chooses uniformly from those winners. Adding $\epsilon$-uniform noise modifies the sampling rule to sometimes pick uniformly random positions. In simulation, we explore natural variations of each of these modeling components: changing the voter distribution, adding plurality winners from earlier generations or runners-up to the imitation pool, and adding copying errors to the sampling rule or sampling different numbers of candidates across elections. See \Cref{app:variants} for formal definitions of the variants in this section.

\paragraph{Non-uniform voters.}

We explore our replicator dynamics with symmetric unimodal and bimodal voter distributions. In \Cref{fig:variants}, we show results with three voter distributions: the unimodal distribution $\text{Beta}(2,2)$, the bimodal, extreme voter distribution $\text{Beta}(0.5,0.5)$, and a bimodal double Weibull distribution~\cite{balakrishnan1985double} with shape $4$, location $0.5$, and scale $0.3$ (see \Cref{fig:voter-pdfs} in \Cref{app:variant-plots} for visualizations of these distributions). 
The basic pattern from \Cref{thm:main} continues to hold with these voter distributions.
However, when voters are $\text{Beta}(2,2)$-distributed, the two clusters at $k = 5$ are significantly closer to the center. 

\paragraph{Memory.}

In the basic model, candidates only copy the positions of winners in the previous generation. However, real-world candidates will likely have memory of earlier winners, so in this variant, we allow candidates to sample from winner positions in any of the last $m$ generations. In \Cref{fig:variants}, we see that adding $m = 2$ generations of memory for candidates still maintains the pattern from \Cref{thm:main}. The results for $m = 3$ are extremely similar (see \Cref{fig:3-gen-mem} in \Cref{app:plots}).

\paragraph{Perturbation noise.}

With \emph{perturbation noise}, each candidate slightly deviates from the position they copy, as if their imitation is imperfect. 
In our simulations, we add Gaussian noise with mean 0 and variance $\sigma^2$ to each copied position. \Cref{fig:variants} shows that the candidate distribution with a small amount of perturbation noise ($\sigma^2 = 0.005$) converges to the center for $k = 2, 3, 4$ but does not for $k\ge 5$. However, with sufficient noise, higher values of $k$ form a single central cluster; we see this in \Cref{fig:variants} with $k = 6$ and $\sigma^2 = 0.01$.
Additionally, for $k \geq 6$ the behavior varies significantly across runs without enhanced symmetry. We even observe phenomena such as party movement, divergence, and extinction, particularly for higher values of $k$ (see \Cref{fig:perturb-1-trial} in \Cref{app:plots}).

\paragraph{Variable candidate counts.}

In real-world elections, we might expect different numbers of candidates to run in different elections, but our model keeps the candidate count $k$ constant. In this variant, we allow elections in each generation to have a mixture of several candidate counts, where candidates copy from winner positions across all $k$ in the previous generation. We find that our results interpolate smoothly to this setting: when most elections have fewer than five candidates, we see convergence to the center, but not when most elections have $k \ge 5$ (see \Cref{fig:variants}, where we simulate an equal mixture of the listed candidate counts in each generation, with 50,000 elections per $k$). See \Cref{fig:k-mixture-heatmap} in \Cref{app:plots} for a more fine-grained experiment in which we smoothly vary the proportions of elections with $k = 3, 4, 5$.

\paragraph{Top-$h$ copying.}

 Finally, we explore a variant where candidates in generation $t$ choose a position to copy from the pool of candidates with the top-$h$ highest vote shares in generation $t-1$, rather than only winners (i.e., $h = 1$). In simulation, top-$h$ copying is the only variant which strays from the dichotomy we establish in \Cref{thm:main}---perhaps unsurprisingly, given that our result is about copying winners. For $k=3,4$, when $h=2$ the candidate distribution does not converge to the center and instead ends up as $k = 5$ usually does, with two clusters (see \Cref{fig:variants}, bottom right). For $h=3$, the candidate distribution does not even appear to converge towards point masses (see \Cref{fig:top-3} in \Cref{app:variant-plots}). These simulations suggest that our central finding (convergence to the center for $k < 5$) is a result of copying the positions of plurality winners specifically, and the dynamics under this heuristic. 
 
\begin{figure}[t!]
\centering
\includegraphics[width=0.49\textwidth]{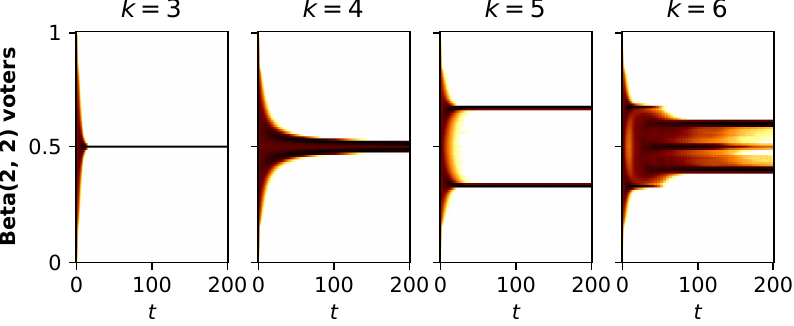}
~
\includegraphics[width=0.49\textwidth]{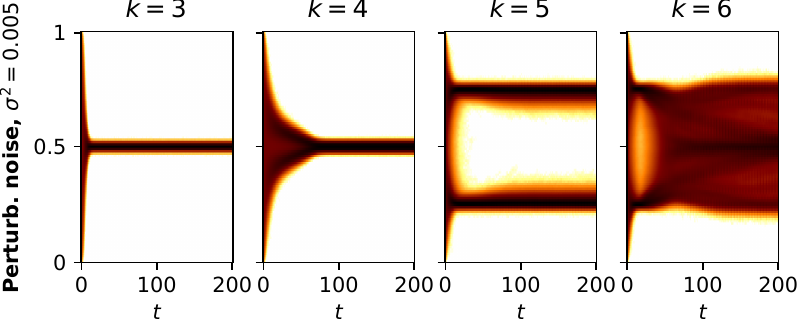}\\
\includegraphics[width=0.49\textwidth]{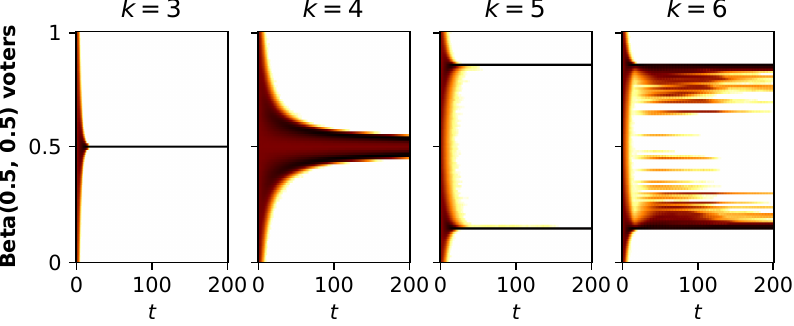}
~
\includegraphics[width=0.49\textwidth]{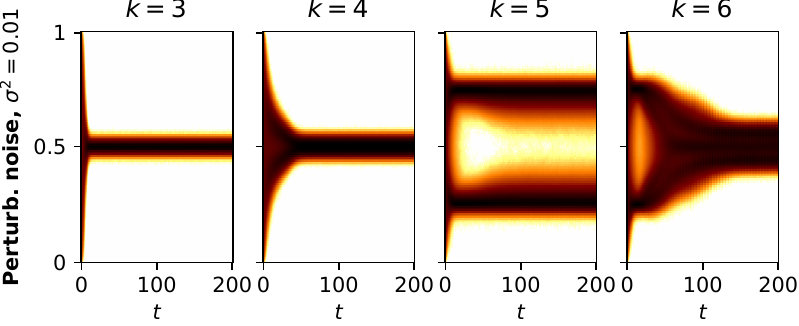}\\
\includegraphics[width=0.49\textwidth]{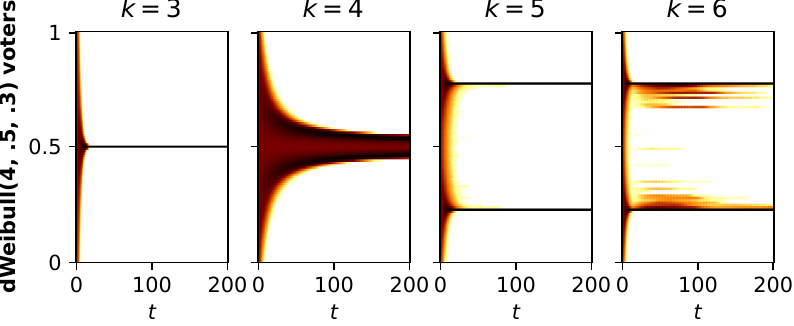}
~
\includegraphics[width=0.49\textwidth]{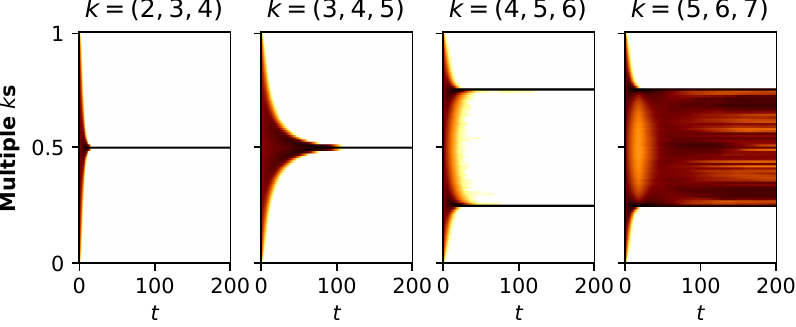}\\
\includegraphics[width=0.49\textwidth]{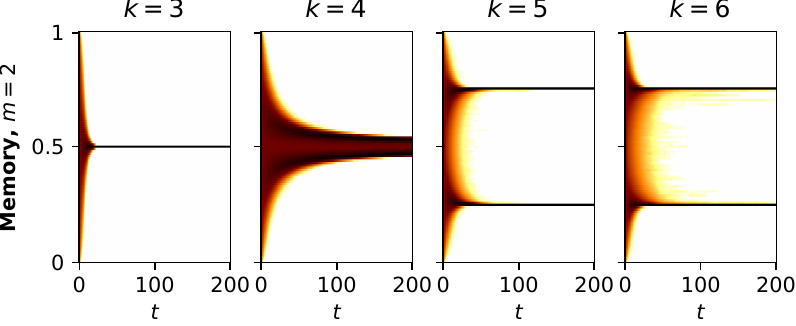}
~
\includegraphics[width=0.49\textwidth]{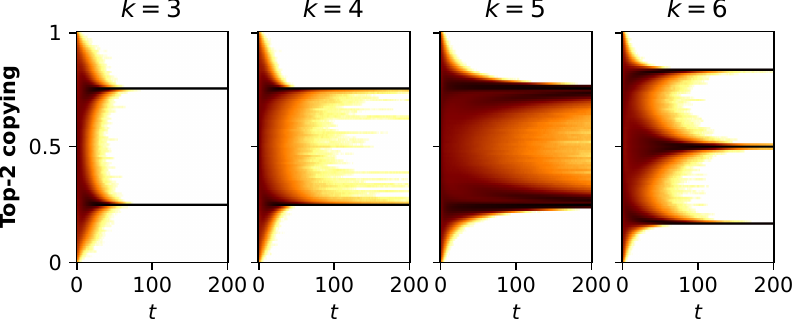}\\
\caption{Variants of the replicator dynamics. Each plot shows 50 trials with no enhanced symmetry. Left column, top to bottom: three different voter distributions and 2 generations of memory. Right column, top to bottom: perturbation noise with $\sigma^2 = 0.005$ and $0.01$, variable candidate counts, and top-2 copying. Except for top-2 copying, all of the variants converge to the center for $k < 5$. Additionally, sufficiently high perturbation noise can cause a central cluster to form for high $k$.}
  \label{fig:variants}
\end{figure}

\section{Relationship to Nash equilibria of one-shot games}\label{sec:nash}
We now take a step back and examine the relationship between our dynamics and prior research on strategic positioning. As we discussed, much of the literature on candidate positioning has focused on one-shot games rather than dynamics~\cite{osborne1995spatial,bol2016endogenous,kurella2017evolution}, as in the Hotelling--Downs model. In our 1-Euclidean setting with uniform voters, the Hotelling--Downs equilibrium has both candidates positioned at $1/2$---which as we showed, is also the attracting distribution of the replicator dynamics with $k=2$. Indeed, it is well-known that Nash equilibria of one-shot games are fixed points of the corresponding replicator dynamics~\cite{hofbauer2003evolutionary}, but replicator dynamics fixed points may not be Nash equilibria. We can see this intuitively in our setting by noting that a distribution $F$ is a (symmetric, mixed-strategy) Nash equilibrium if no strategy does better against $F$ than sampling from $F$, while $F$ is a fixed point of the replicator dynamics if no strategy \emph{drawn from $F$} does better against $F$ than sampling from $F$. For $F$ with full support, symmetric mixed-strategy Nash equilibria and replicator dynamics fixed points thus coincide~\cite{bauer2019stabilization}. However, Nash equilibria can be unstable under the dynamics---and even if they are attractors, their basins of attraction may be negligible. 

Before analyzing Nash equilibria, we first need a brief digression to address what happens when multiple candidates occupy the same point---we call these \emph{positional ties}. Since we have so far assumed that candidate distributions are atomless, our analyses of the replicator dynamics has avoided this issue: with an atomless candidate distribution, positional ties occur with probability 0. One option for handling positional ties is to suppose that candidates fail to position themselves \emph{exactly} at the same point and imagine that there is some infinitesimal jitter in their positions which determines a left--right order. Alternatively, we could suppose that candidates are in fact precisely at the same point, forcing voters to make an arbitrary choice between them. 
\begin{definition}
  Suppose multiple candidates occupy the same point. Under \emph{left--right tie-breaking}, one of these candidates (chosen u.a.r.) receives the entire left vote share allocated to that point, while a different candidate (also u.a.r.) receives the entire right vote share. 
  Under \emph{equal split tie-breaking}, all candidates at a point share the vote share allocated to that point equally. Equivalently, voters randomly choose between equidistant candidates.
\end{definition}

Armed with these positional tie-breaking rules, we now provide several results that demonstrate how a static analysis of Nash equilibria yields more fragile conclusions than analyzing the asymptotic behavior of the replicator dynamics. 
 We focus on two types of equilibria: (1) symmetric mixed-strategy Nash equilibria (SMSNEs), since these relate to fixed points of the replicator dynamics; and (2) pure-strategy Nash equilibria (PSNEs), since these are the focus of classical candidate positioning analyses.

We begin by showing there are multiple SMSNEs in the
one-shot candidate positioning game, but they are often unstable or
have tiny basins of attraction under the dynamics; that is, they are
unlikely to be relevant in practice. In contrast, as we have seen in
theory and simulation, the replicator dynamics behave in qualitatively
similar ways under a range of specifications. Then, we show that PSNEs
are very sensitive to the choice of positional tie-breaking rule: we
arrive at entirely different conclusions if we adopt
left--right versus equal split tie-breaking. In contrast, the 
positional tie-breaking rule is irrelevant to our analysis 
with atomless candidate distributions.

\subsection{Symmetric mixed-strategy Nash equilibria}
 Since SMSNEs are a subset of the replicator dynamics fixed points, we might hope to understand the dynamics by analyzing SMSNEs of the game where candidates seek to maximize their plurality win probability. However, we find that there are multiple SMSNEs and they can have trivial basins of attraction. For instance, every candidate at $1/2$ is a SMSNE and a replicator dynamics fixed point (with left--right tie-breaking\footnote{In this subsection, we adopt left--right tie-breaking since it yields equilibria that more closely align with the typical behavior of the replicator dynamics. For instance, we will see in \Cref{sec:positional-ties} that with equal split tie-breaking, all candidates at $1/2$ is only a SMSNE for $k = 2$---not $k = 3$ or $4$,  where we know the candidate distribution also converges to the center.}). But as we have seen, for $k \ge 5$  all symmetric atomless initial distributions do not converge to a point mass at $1/2$. On the other hand, if we allow initial distributions with point masses and the mass at $1/2$ is sufficiently high, the candidate distribution does indeed approach the all-at-$1/2$ SMSNE.

\begin{restatable}{theorem}{centersmsne}
\label{thm:converge-to-center-smsne}
Suppose $F_{0}$ places probability mass $p$ at $1/2$. For any $k \ge 2$, there is some $p^*_k < 1$ such that if $p > p^*_k$, the candidate distribution converges to a point mass at $1/2$ under the replicator dynamics with left--right tie-breaking. One of the fixed points of $p^k + k p^{k - 1} (1-p)$ is such a $p^*_k$.  
\end{restatable}

See \Cref{app:nash-proofs} for proofs omitted from this section. Additionally, there is another family of SMSNEs where each candidate randomly picks between the points $x$ and $1-x$ (for $x\in (1/4, 1/2)$). 

\begin{restatable}{theorem}{twospikesmsne}
\label{thm:two-spike-smsne}
  With $k \ge 4$ and left--right tie-breaking, for any $x \in (1/4, 1/2)$, the strategy where each candidate picks uniformly at random between $x$ and $1-x$ is a SMSNE. 
\end{restatable}

Just as with the all-at-1/2 equilibrium, this SMSNE is not indicative of the typical behavior of the replicator dynamics. However, we can show as before that for non-atomless distributions, the candidate distribution can converge to this type of equilibrium.

\begin{restatable}{theorem}{twospikeconvergence}
\label{thm:two-spike-convergence}
  Suppose $F_{0}$ places probability mass $p$ at $x$ and at $1-x$, for $1/4 < x < 1/2$. For any $k \ge 5$, there exists some $p^*_k < 1/2$ such that if $p > p^*_k$, the candidate distribution converges to point masses at $x$ and $1-x$ under the replicator dynamics. In particular, one of the fixed points of $(2p)^k/2 + k(1-2p)((2p)^{k-1} - 2p^{k-1}) / 2$ is such a $p^*_k$.
\end{restatable}
These results demonstrate the existence of many SMSNEs that alone do not tell us how we should expect the replicator dynamics to behave.

\subsection{Positional tie-breaking and pure-strategy Nash equilibria}\label{sec:positional-ties}
We now demonstrate how ignoring dynamics and focusing on static equilibria can yield results very sensitive to tie-breaking rules. \citet{cox1987electoral} extends the Hotelling--Downs analysis to more than two candidates, characterizing PSNEs of a one-shot candidate positioning game---crucially, with equal split tie-breaking.  With uniform voters and $k \ge 3$ candidates, \citeauthor{cox1987electoral} proves that there is no PSNE for odd $k$ and that the only PSNE for even $k$ has evenly-spaced pairs of candidates at $1/k, 3/k, \dots, (k-1)/k$.
 Clearly, this analysis makes very different predictions than our replicator dynamics. However, we show that \citeauthor{cox1987electoral}'s results depend strongly on equal split tie-breaking. 
 
 To state \citeauthor{cox1987electoral}'s result formally, we need to fully specify the candidate objective. We focus on the objective \citeauthor{cox1987electoral} calls \emph{complete plurality maximization}, where candidates seek first to maximize their vote margin against their strongest competitor, then second-strongest, etc. We extend this objective to allow stochastic positional tie-breaking, assuming candidates first maximize their win probability, then each of their expected vote margins. We can then state \citeauthor{cox1987electoral}'s result.

\begin{theorem}[Special case of Theorem 2 from \citet{cox1987electoral}]
With uniform voters, $k \ge 3$ complete plurality maximizing candidates, and equal split tie-breaking,
\begin{enumerate}
  \item if $k$ is odd, there is no PSNE,
  \item if $k$ is even, then the unique PSNE has two candidates at each of the points $1/k, 3/k, \dots, (k-1)/k$.
\end{enumerate}
\end{theorem}

If we instead use left--right tie-breaking, the picture is dramatically different. In particular, all candidates at $1/2$ is then a PSNE \emph{for all $k$}: any deviant who moves from $1/2$ loses with certainty to the center candidate who captures the opposite side of the vote. Left--right tie-breaking also introduces many additional PSNEs; we list some of them in the following theorem.

\begin{restatable}{theorem}{psnes}
\label{thm:psnes}
The following are (some\footnote{In \Cref{app:nash-proofs}, we show that for $k \le 5$, this list of PSNEs is exhaustive  (\Cref{thm:small-k-psnes}); for $k > 6$, there may be others.} of the) PSNEs with  uniform voters, complete plurality maximizing candidates, and left--right tie-breaking:
  \begin{enumerate}
    \item Any $k \ge 2$: all $k$ candidates at $1/2$.
    \item Any $k \ge 4$: for any $x \in (1/4, 1/2)$, $\lfloor k / 2 \rfloor$ candidates at $x$, $\lfloor k / 2 \rfloor$ candidates at $1-x$, and the last candidate (if $k$ is odd) at either $x$ or $1-x$. 
    \item Any $k \ge 5$: $\lfloor (k-1) / 2 \rfloor$ candidates at $1/4$, $\lfloor (k-1) / 2 \rfloor$ candidates at $3/4$, one candidate at $1/2$, and the last candidate (if $k$ is even) at either $1/4$ or $3/4$. 
    \item Even $k$: \citeauthor{cox1987electoral}'s equilibrium; two candidates at each of the points $1/k, 3/k, \dots, (k - 1)/k$.
  \end{enumerate}
\end{restatable}
 Thus, the qualitative conclusions we arrive by examining Nash equilibria are very different from \citeauthor{cox1987electoral}'s if we make another similarly reasonable assumption. \citeauthor{cox1987electoral}'s analysis tells us we should not expect candidates converging to the center for any $k > 2$, but if we use left--right tie-breaking, we find that central configurations are equilibria for all $k$. The replicator dynamics reveal when these configurations are stable: only for small $k$. These results highlight how analyzing Nash equilibria provides a brittle picture of candidate positioning, yielding results that are sensitive to tie-breaking and do not capture iterated play. Even SMSNEs, which are closely related to replicator dynamics fixed points, fail to reveal the typical behavior of the dynamics.

\section{Discussion}\label{sec:discussion}
We introduced a replicator dynamics model of one-dimensional candidate positioning in plurality elections based on simple heuristic inspired by bounded rationality. Our theoretical results show that the candidates converge to the center when there are at most four candidates per election, but diverge when there are five or more candidates per election. Simulations confirm that this pattern is robust to a large range of model variations. We contrast our results to prior work that focuses on static equilibria or lacks theoretical results for more than two candidates.

Many open questions remain in the analysis of our model. The foremost
is a theoretical characterization of the asymptotic candidate
distribution for $k \ge 5$, although this may be challenging given the
complex high-$k$ behavior we observe in simulation. An even larger
challenge is posed by expanding beyond symmetric and atomless initial
candidate distributions to  distributions which have points masses or
are asymmetric. As we saw in \Cref{thm:two-spike-convergence},
allowing atomless distributions means there are infinitely attracting
distributions for $k \ge 5$, so the task becomes one of cataloguing
all of the possible long-run candidate distributions. Theoretical
results for our model variants would be interesting, such as
characterizing which mixtures of candidate counts $k$ lead to
convergence to the center, or conditions on voter
distributions that result in central convergence for $k \le 4$.

While we explored several model variations in simulation, there are many more than can possibly be covered in a single paper. Additional variations of particular interest include policy-motivated candidates, strategic voters, probabilistic voters, and higher-dimensional preferences. Another natural direction would be to explore voting systems other than plurality, like two-round runoff, instant runoff, or Borda count; Condorcet methods are considerably less interesting under our one-dimensional replicator dynamics, since the candidate closest to the median voter always wins, but might exhibit more complex behavior in higher dimensions.

\section* {Acknowledgments}
This work was supported in part by ARO MURI, a Simons Collaboration grant, a grant from the MacArthur Foundation, a Vannevar Bush Faculty Fellowship, AFOSR grant FA9550-19-1-0183, and NSF CAREER Award \#2143176.

\bibliographystyle{plainnat}
\bibliography{references}

\clearpage

\appendix

\section*{Appendix for Replicating Electoral Success}
\vspace*{1em}

\tableofcontents
\clearpage
\section{Additional Plots}\label{app:plots}

\begin{figure}[h]
  \centering
  \includegraphics[width=\textwidth]{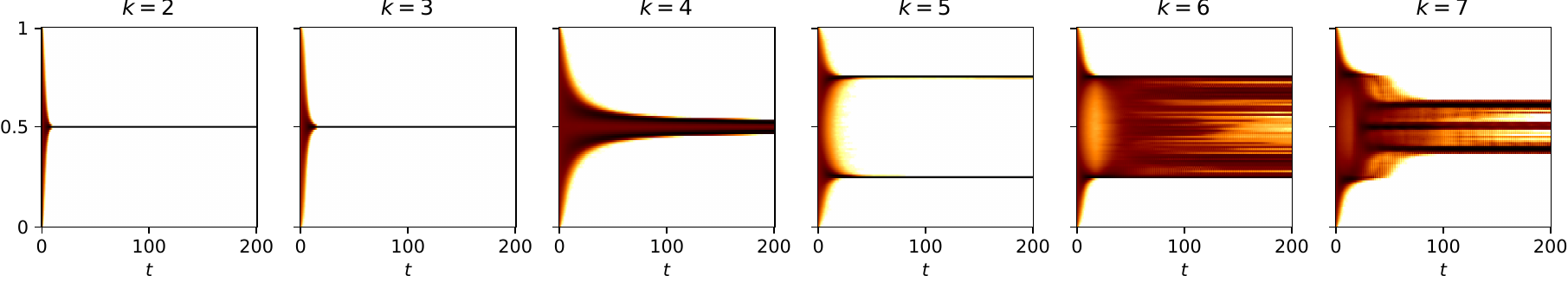}\\
  \includegraphics[width=\textwidth]{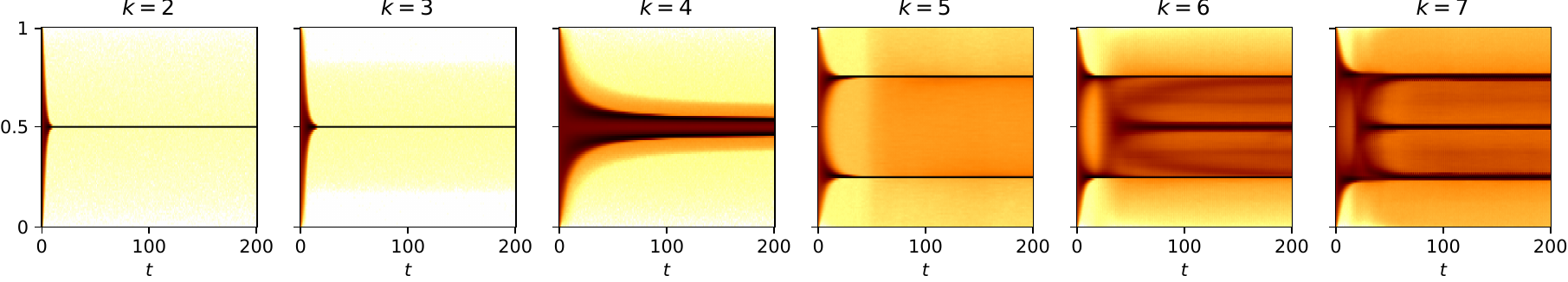}
  \caption{Replicator dynamics runs just as in \Cref{fig:small-k-symmetry}, but without enhanced symmetry. For $k > 6$, the behavior of the Monte Carlo trials becomes inconsistent without enhanced symmetry, particularly without $\epsilon$-uniform noise. See \Cref{fig:small-k-symmetry} for more details.}
  \label{fig:small-k}
\end{figure}

\begin{figure}[h]
  \centering
  \includegraphics[width=\textwidth]{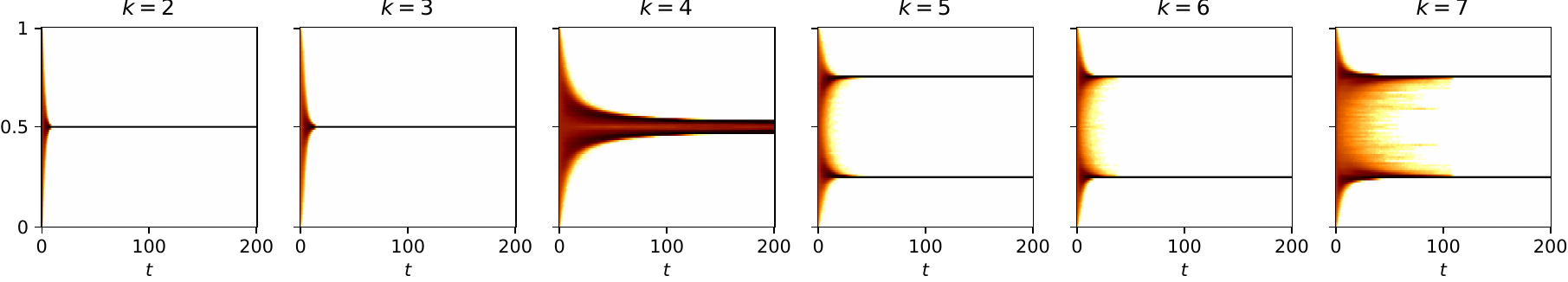}\\
  \includegraphics[width=\textwidth]{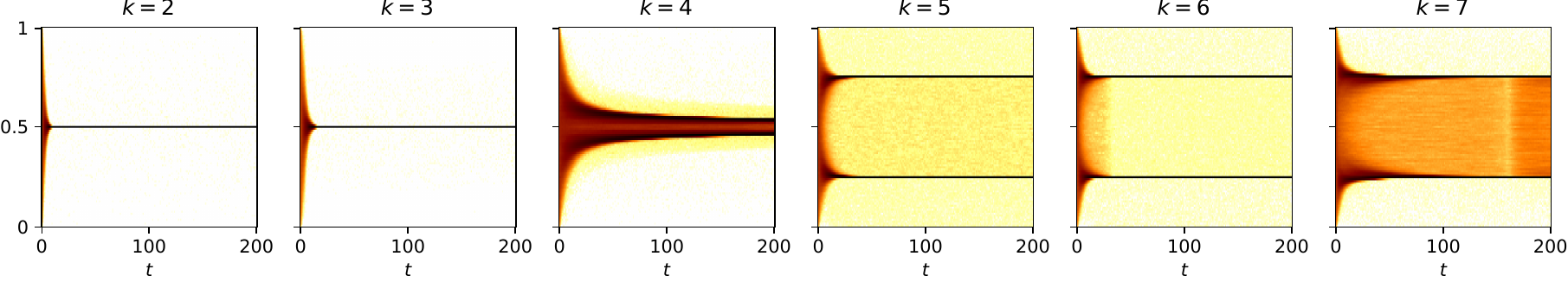}
  \caption{Replicator dynamics runs with enhanced symmetry just as in \Cref{fig:small-k-symmetry}, but showing only a single trial instead of aggregating 50 runs. With enhanced symmetry, the behavior is very consistent across runs.}
  \label{fig:small-k-symmetry-1-trial}
\end{figure}

\begin{figure}[h]
  \centering
  \includegraphics[width=\textwidth]{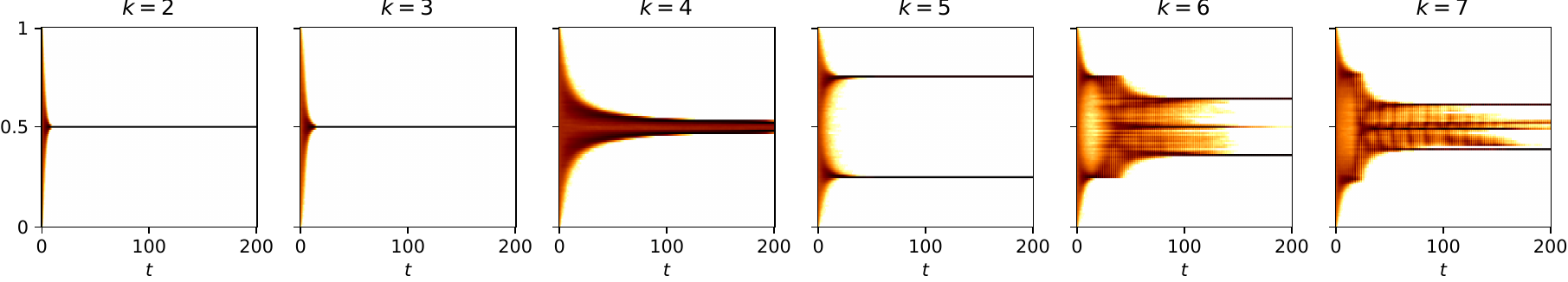}\\
  \includegraphics[width=\textwidth]{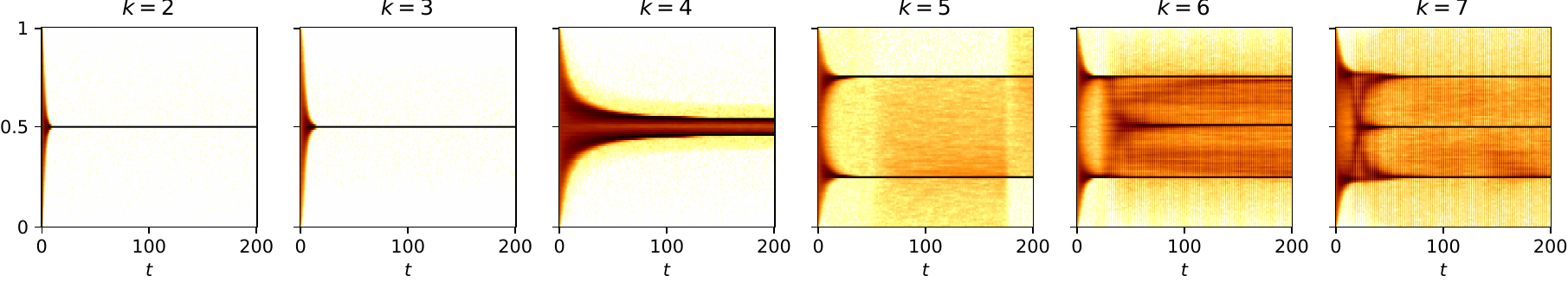}
  \caption{Replicator dynamics runs just as in \Cref{fig:small-k} (no enhanced symmetry), but showing only a single trial instead of aggregating 50 runs to highlight the inconsistent behavior for $k = 6$ and $7$ without enhanced symmetry.}
  \label{fig:small-k-1-trial}
\end{figure}

\begin{figure}[h]
  \centering
  \includegraphics[width=\textwidth]{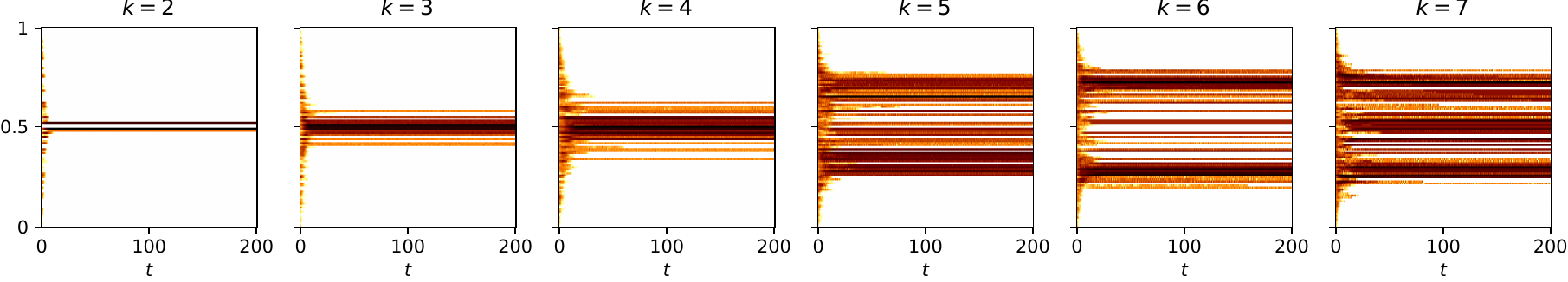}\\
  \includegraphics[width=\textwidth]{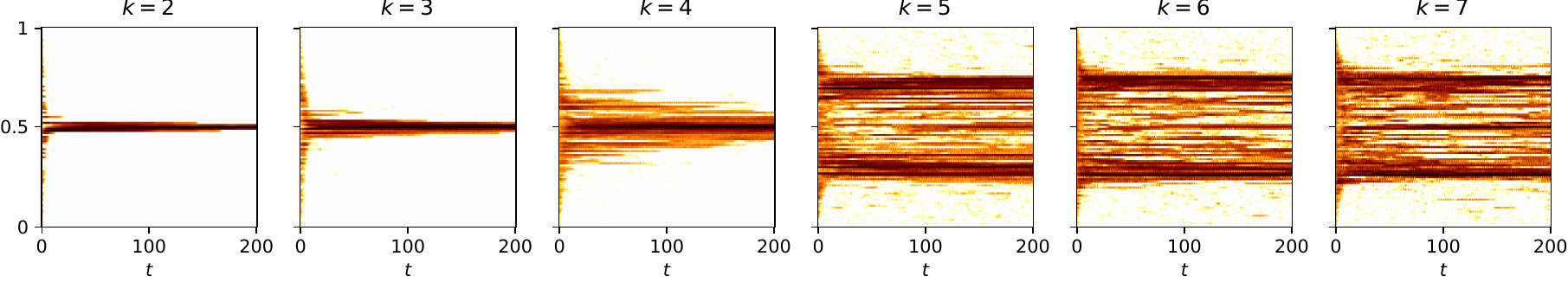}
  \caption{Replicator dynamics runs with only 50 elections per generation, without enhanced symmetry. Each plot shows 50 trials. The top row has no noise, while the bottom row uses $0.01$-uniform noise. Even with a small sample size, our main finding holds. }
  \label{fig:n-50}
\end{figure}

\begin{figure}[h]
  \centering
  \includegraphics[width=\textwidth]{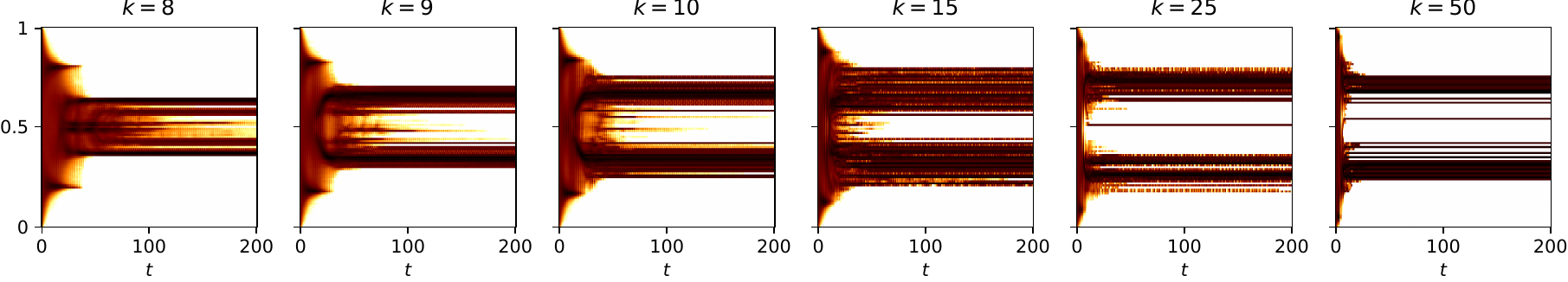}\\
  \includegraphics[width=\textwidth]{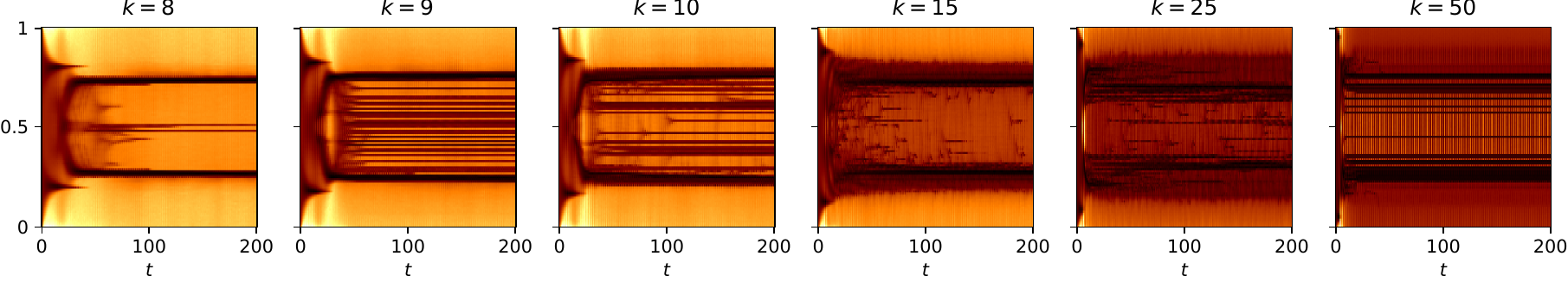}
  \caption{Replicator dynamics runs just as in \Cref{fig:large-k-symmetry}, but without enhanced symmetry. As with smaller values of $k$, the behavior becomes more chaotic without enhanced symmetry. }
  \label{fig:large-k}
\end{figure}

\begin{figure}[t]
\centering
\includegraphics[width=\textwidth]{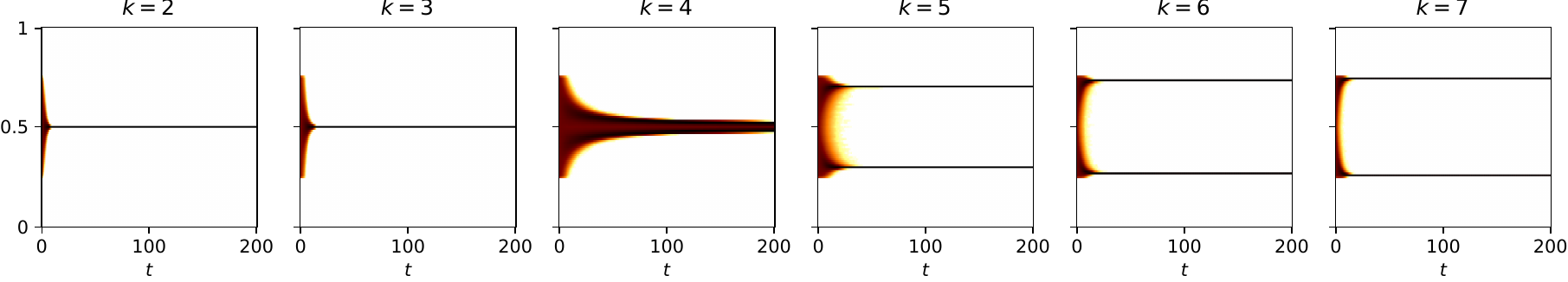}
  \caption{Replicator dynamics with initial candidate distribution Uniform($1/4, 3/4$). These plots show 50 trials with 100,000 elections per generation, no noise, and without enhanced symmetry. The dynamics are very well-behaved with $(1/4, 3/4)$ support, removing the need for enhanced symmetry; compare to \Cref{fig:small-k}.}
  \label{fig:1/4-3/4}
\end{figure}
\clearpage

\subsection{Additional variant plots}\label{app:variant-plots}

\begin{figure}[h]
\centering
\includegraphics[width=\textwidth]{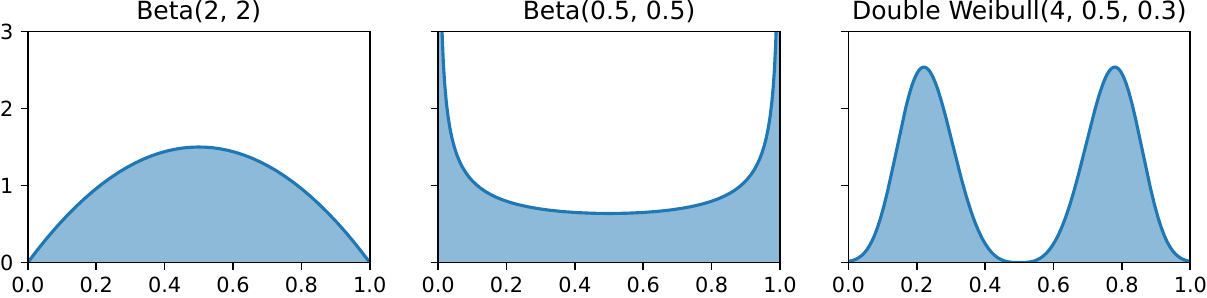}\\
\caption{PDFs of different voter distributions used in \Cref{fig:variants}.}
\label{fig:voter-pdfs}
\end{figure}

\begin{figure}[h]
\centering
\includegraphics[width=\textwidth]{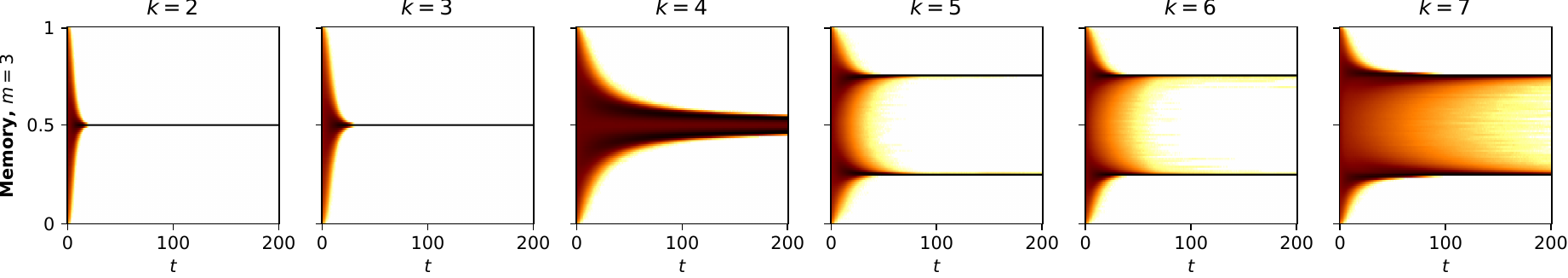}\\
\caption{Replicator dynamics with $m=3$ generations of memory, no enhanced symmetry, and $50$ trials per plot. There is no qualitative difference between $m=3$ and $m=2$ (compare to \Cref{fig:variants}).}
\label{fig:3-gen-mem}
\end{figure}

\begin{figure}[h]
  \centering
\includegraphics[width=\textwidth]{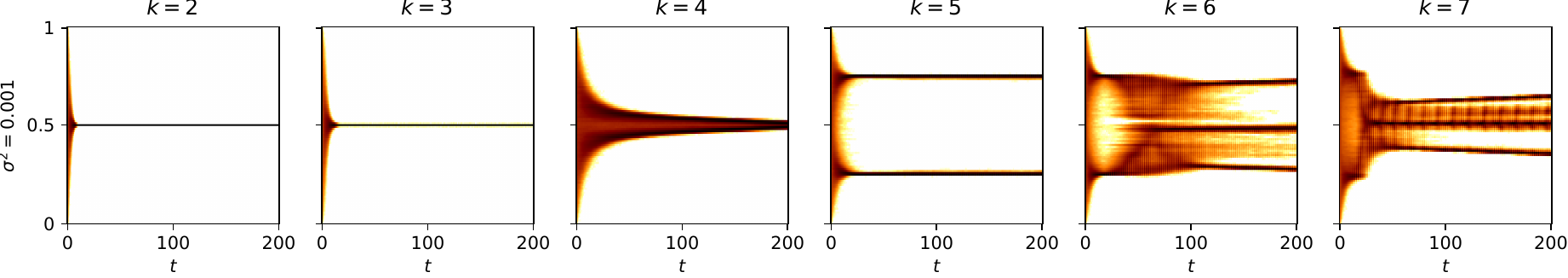}\\[.5em]
\includegraphics[width=\textwidth]{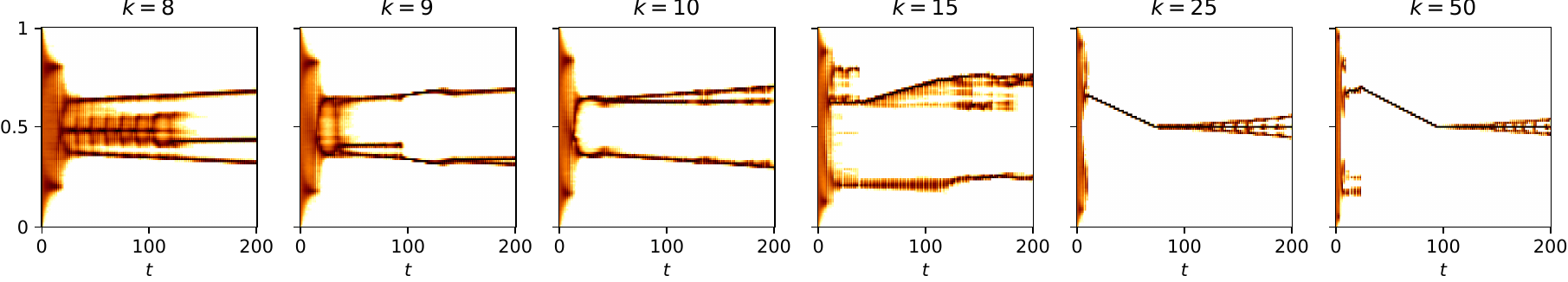}\\[.5em]
\includegraphics[width=\textwidth]{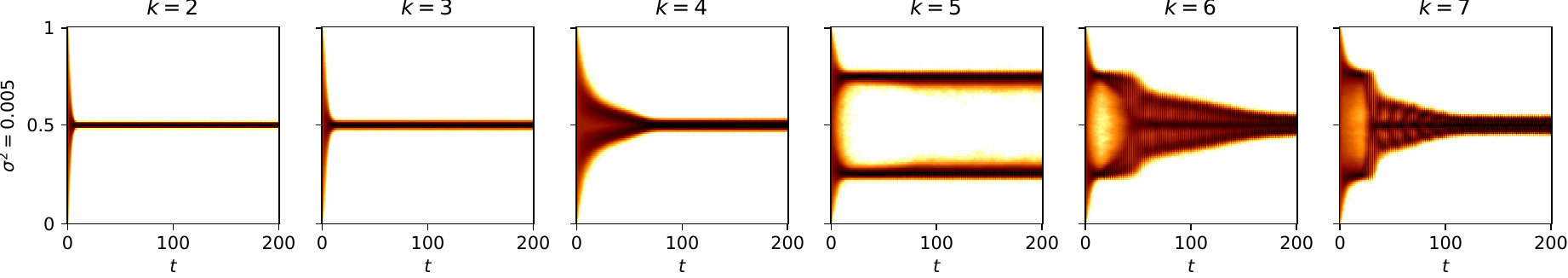}\\[.5em]
\includegraphics[width=\textwidth]{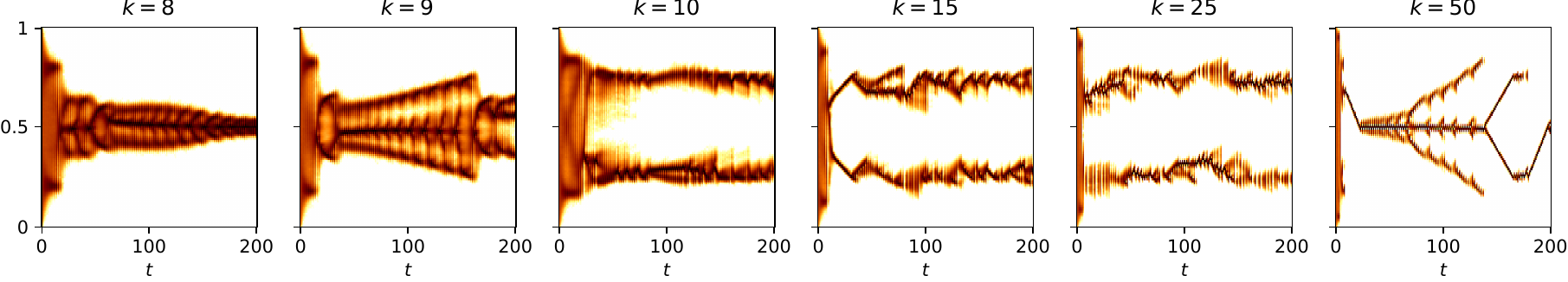}\\[.5em]
\includegraphics[width=\textwidth]{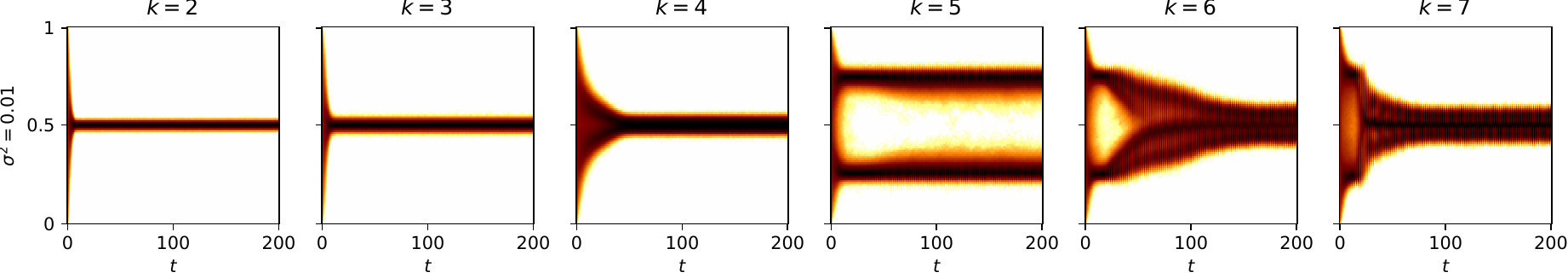}\\[.5em]
\includegraphics[width=\textwidth]{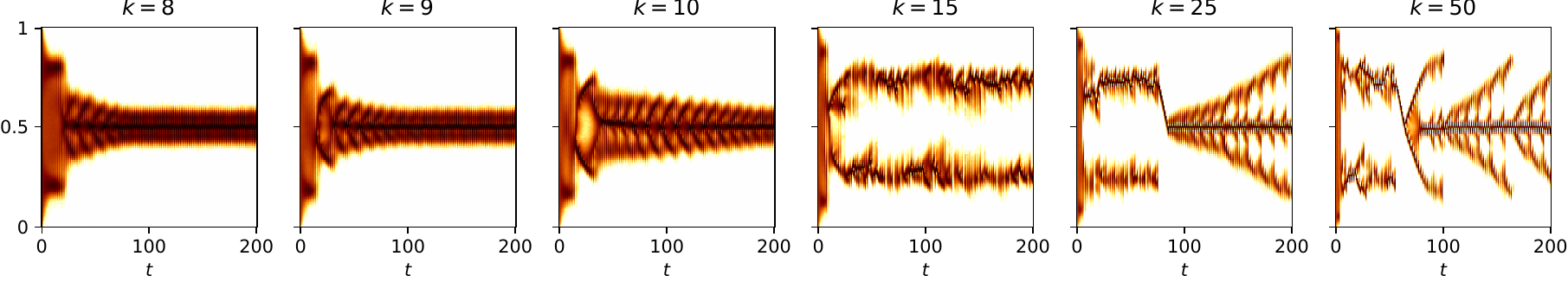}
\caption{Single trials of the replicator dynamics with perturbation noise and 100,000 elections per generation. The first two rows use $\sigma^2 = 0.001$, the middle two use $\sigma^2 = 0.005$, and the bottom two use $\sigma^2 = 0.01$. Perturbation noise combined with Monte-Carlo asymmetries can result in complex and unpredictable branching with higher $k$.}
\label{fig:perturb-1-trial}
\end{figure}

\begin{figure}[h]
  \centering
  \includegraphics[width=0.5\textwidth]{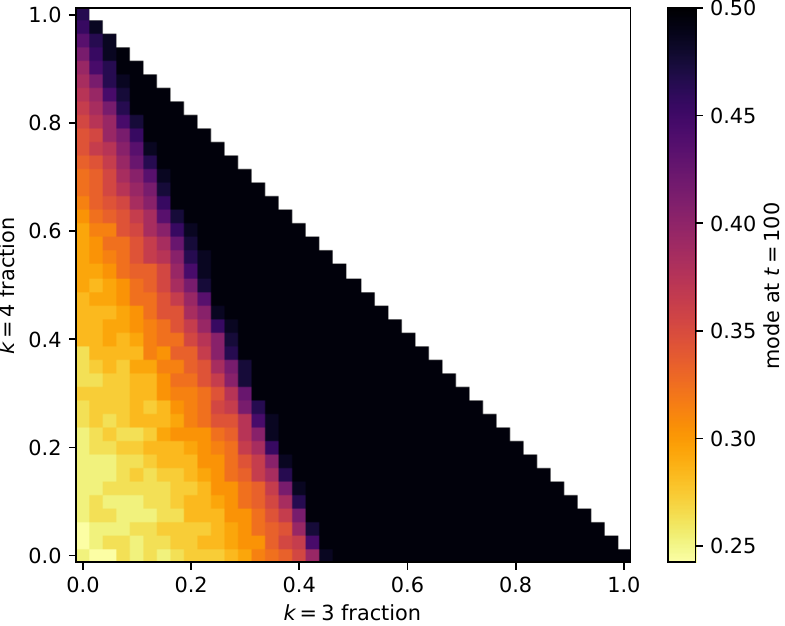}
  \caption{Heatmap showing the position of the candidate distribution mode at $t=100$  when elections have a mixture of $k =3, 4, $ and $5$ candidates each (only modes $\le 1/2$ are shown). These simulations use 100,000 elections per generation, with $k$ split between $3, 4, $ and $5$ in different proportions at each point. The fraction of elections with $3$ candidates varies along the $x$ axis, while the fraction  with $4$ candidates varies along the $y$ axis. Any remaining elections have $k=5$. For instance, the lower left corner has all 100,000 elections use $k = 5$, while the point $(1/3, 1/3)$ has an even mix of candidate counts. When either the $k = 3$ or $k = 4$ fraction is high enough (but especially $k = 3$), the distribution converges to the center, with the mode at $1/2$. However, with enough $k = 5$ elections, two clusters emerge, and more $k = 5$ elections pushes them farther apart. }
  \label{fig:k-mixture-heatmap}
\end{figure}

\begin{figure}[h]
\centering
\includegraphics[width=\textwidth]{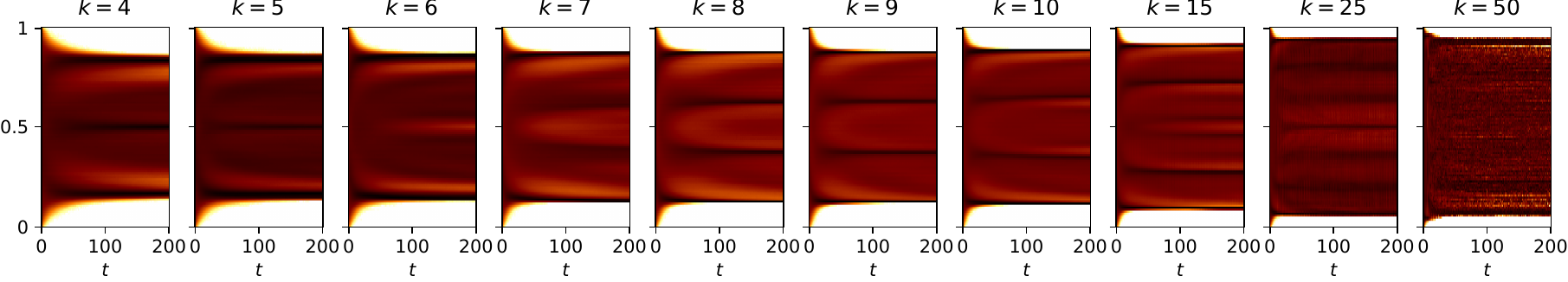}\\
\caption{Replicator dynamics with top-$h$ copying where $h=3$, no enhanced symmetry, $50$ trials per plot, and 100,000 elections per generation.}
\label{fig:top-3}
\end{figure}

\clearpage

\section{Additional Proofs}\label{app:proofs}

\subsection{Proofs from \Cref{sec:main-theory}}\label{app:main-proofs}

\kthree*

\begin{proof}
Let $x < 1/2$ and define $p = F_{3, t-1}(x)$. Consider the following cases for the positions of the three candidates $X_{1,t}, X_{2,t},$ and $X_{3,t}$. Call  candidates in $(x, 1-x)$ \emph{inner}. 
\begin{enumerate}
  \item All three candidates in $[0, 1/2)$ (and the symmetric case). First suppose all three are in $[0, 1/2)$ (the other side is symmetric).  If there is at least one inner candidate (w.p.\ $1/2 ^3 - p^3$), then the winner is inner. Accounting for symmetry, an inner candidate wins in this case w.p.\ $2  (1/2 ^3 - p^3) = 1/4 - 2p^3$.
  \item Two candidates in $(x, 1/2)$ and one in $(1/2, 1-x)$ (and the symmetric case). Since all candidates are inner, an inner candidate wins. Accounting for symmetry, an inner candidate wins in this case w.p.\ $2 \left[3  (1/2 - p)^3\right] = 6 (1/2-p)^3$.
  \item Two candidates in $[0, x)$ and one in $(1/2, 1-x)$ (and the symmetric case). The candidate in $(1/2, 1-x)$ wins with vote share at least $1/2$. Accounting for symmetry, an inner candidate wins in this case w.p.\ $2  \left[3 p^2 (1/2 - p)\right] = 6p^2(1/2 -p)$.
  \item One candidate in $[0, x)$, one in $(x, 1/2)$, and one  in $(1/2, 1-x)$. Label them 1, 2, and 3, respectively. Candidate 3 gets vote share $1- (X_3 +X_2) / 2 = \left[(1-X_3) + (1-X_2)\right]/2$, while candidate 1 gets vote share $(X_1 + X_2)/2$. Since $X_3 < 1-x$, $1-X_3 > X_1$; and since $X_2 < 1/2$, $1-X_2 > X_2$. Thus candidate 3 has higher vote share than candidate 1 and an inner candidate wins. Accounting for symmetry, an inner candidate wins in this case w.p.\ $2\left[3\cdot 2 p(1/2-p)^2\right] = 12p(1/2 - p)^2$. \label{case:k-3-1110}
\end{enumerate}
Adding up these cases yields a lower bound on the probability that an inner candidate wins:
\begin{align*}
  \Pr(x < \plurality(X_{1,t}, X_{2,t}, X_{3,t}) < 1-x) &\ge 1/4 - 2p^3 + 6(1/2 - p)^3 + 6p^2 (1/2-p) + 12p(1/2 - p)^2\\
  &= 1 - 3/2 \cdot p - 2p^3.
\end{align*}
By symmetry, this yields the claimed upper bound on the probability a candidate in $[0, x]$ wins:
\begin{align*}
  F_{3, t}(x) &=\Pr(\plurality(X_{1,t}, X_{2,t}, X_{3,t}) \le x)\\
  &=(1 - \Pr(x < \plurality(X_{1,t}, X_{2,t}, X_{3,t}) < 1-x))/2\\
  &\le \left[1-(1 - 3/2 \cdot p - 2p^3)\right]/2\\
  &= 3/4 \cdot p +p^3\\
  &= 3/4 \cdot F_{3, t-1}(x) +F_{3, t-1}(x)^3.
\end{align*}

We now show the closed form bound by induction on $t$.  We'll simultaneously show that $F_{3, t}(x) \le F_{ 0}(x)$. For the base case $t = 0$, we have $F_{3, t}(x) \le F_{0}(x)$. Now for $t > 0$, suppose the claims hold for $t-1$. Using the bound above, we know that 
\begin{align*}
  F_{3, t}(x) &\le 3/4 \cdot F_{3, t-1}(x) + F_{3, t-1}(x)^3\\
  &=  F_{3, t-1}(x) \cdot \left[3/4 + F_{3, t-1}(x)^2\right]\\
  &\le F_{3, t-1}(x) \cdot \left[3/4 + F_{0}(x)^2\right]\tag{by IH}\\
  &\le F_{0}(x) \cdot \left[3/4 + F_{0}(x)^2\right]^{t-1} \cdot \left[3/4 + F_{0}(x)^2\right]\tag{by IH}\\
  &= F_{0}(x) \cdot \left[3/4 + F_{0}(x)^2\right]^t
\end{align*}
This is the main claim we wanted to show. We can now also show the supporting fact that $F_{3, t}(x) \le F_{0}(x)$. For $x< 1/2$, $F_{0}(x) \le 1/2$ by symmetry. Thus $3/4 + F_{0}(x)^2 \le  3/4 + 1/2^2 = 1$, so by the inequality above,  $F_{3, t}(x) \le F_{0}(x) \cdot \left[3/4 + F_{0}(x)^2\right]^t \le F_{0}(x)\cdot 1^t$.
\end{proof}

\kfourlemma*
\begin{proof}

Let $x \in (1/3,  1/2)$ and  $p = F_{4, t-1}(x)$. We'll find a lower bound on the probability an inner candidate in $(x, 1-x)$ wins. Consider the following cases for candidate positions in a $k=4$ plurality election:
\begin{enumerate}
  \item All four candidates in $[0, 1/2)$ (and the symmetric case). An inner candidate wins if at least one candidate is inner. Accounting for symmetry, an inner candidate wins in this case w.p.\ $2(1/2^4 - p^4) = 1/8 - 2p^4 $.
    \item Three candidates in $[0, 1/2)$ and one in $(1/2, 1-x)$ (and the symmetric case). The candidate on the right has a higher vote share than any outer candidate on the left (as in  \Cref{thm:k-3-convergence} Case~\ref{case:k-3-1110}), so an inner candidate wins. Accounting for symmetry, an inner candidate wins in this case w.p.\ $2(4\cdot 1/2^3 \cdot (1/2-p)) = 1/2 - p$.
  \item Two candidates in $(x, 1/2)$ and two in $(1/2, 1-x)$. All candidates are inner, so an inner candidate wins. This occurs w.p.\ $\binom{4}{2}\cdot (1/2 - p)^4 = 6(1/2-p)^4$.
  \item Two candidates in $[0, x)$ and two in $(1/2, 1-x)$ (and the symmetric case). Since $x > 1/3$, the rightmost candidate gets vote share greater than $1/3$. Meanwhile, the leftmost candidate gets vote share less than $1/3$. The second-leftmost candidate gets vote share less than $(2/3)/2 = 1/3$ (the candidates flanking it are closer together than 0 and $1-x < 2/3$). Thus an inner candidate wins. Accounting for symmetry, an inner candidate wins in this case w.p.\ $2\binom{4}{2}p^2(1/2-p)^2 = 12p^2(1/2-p)^2$
  \item Two candidates in $(x, 1/2)$, one in $(1/2, 1-x)$, and one in $(1-x, 1]$ (and the symmetric case).  Label the candidates 1--4 in left--right order. By symmetry, candidate 3 is farther from 1/2 than candidate 2 with probability $2/3$: all $3!=6$ orderings of distance from $1/2$ between candidates 1--3 are equiprobable and only the 2 where candidate 3 is closest to $1/2$ fail this property. In this scenario, candidate 4 has vote share $1-(X_3 + X_4)/2 = ((1-X_3) + (1-X_4))/2$ and candidate 1 has vote share $(X_1 + X_2)/2$. Since  $1-X_3 < X_2$ (candidate 2 is closer to the center than 3) and $1-X_4 < X_2$ (since $X_2 > x$ and $X_4 > 1-x$), candidate 1 has a larger vote share than candidate 4, the only outer candidate. Thus an inner candidate wins. Accounting for symmetry, an inner candidate wins in this case w.p.\ $2\cdot 4\cdot 3 \cdot 2/3 \cdot p (1/2-p)^3   = 16p(1/2-p)^3$

  \end{enumerate}

Combining all five cases gives a lower bound on the probability that an inner candidate wins:
\begin{align*}
  &\Pr(x < \plurality(X_{1,t}, X_{2,t}, X_{3,t}, X_{4,t}) < 1-x)\\
  &\ge 1/8 - 2p^4 + 1/2 - p + 6(1/2-p)^4 + 12p^2(1/2-p)^2 + 16p(1/2-p)^3\\
  &= 1 - 2p\\
  &= 1 - 2\cdot F_{4, t-1}(x).
\end{align*}
By symmetry, this means
\begin{align*}
  F_{4, t}(x) &= \Pr(\plurality(X_{1,t}, X_{2,t}, X_{3,t}, X_{4,t}) \le x)\\
  &= \left[1-\Pr(x < \plurality(X_{1,t}, X_{2,t}, X_{3,t}, X_{4,t}) < 1-x)\right]/2\\
&\le \left[1-(1 - 2\cdot F_{4, t-1}(x))\right]/2\\
&= F_{4, t-1}(x).
\end{align*}
The claim then follows by induction on $t$.
\end{proof}

\kfour*

\begin{proof}
Let $x \in (1/3, 1/2)$ and  $p = F_{4, t-1}(x)$. By the argument in the proof of \Cref{lemma:k-4-1/3-bound}, an inner candidate wins with probability at least $1-2p$. We can strengthen this bound using \Cref{lemma:k-4-1/3-bound} and one more case omitted from that analysis (which can't easily be used there): three candidates in $(x/3 + 1/3, 1/2)$ and one in $(1-x, 1]$ (and the symmetric case). Note that $x / 3 + 1/3 = x + 2/3 \cdot (1/2 - x)$ is the point two-thirds of the way from $x$ to 1/2. The leftmost candidate gets vote share more than $x/3 + 1/3$. Meanwhile, the lone outer candidate gets vote share less than $x + (1-x - (x/3 + 1/3)) / 2 = x/3 + 1/3$, so an inner candidate wins. By \Cref{lemma:k-4-1/3-bound}, we know $F_{4, t-1}(x/3 + 1/3) \le F_{4, 0}(x/3 + 1/3)$. Thus, a candidate is in $(x/3 + 1/3, 1/2)$ with probability $1/2 -F_{4, t-1}(x/3 + 1/3) \ge 1/2 - F_{4, 0}(x/3 + 1/3)$. Therefore, accounting for symmetry, an inner candidate wins in this case w.p.\ at least $2 \cdot 4 \cdot (1/2 - F_{4, 0}(x/3 + 1/3))^3 \cdot p = 8 p(1/2 - F_{4, 0}(x/3 + 1/3))^3$.

Combining this new case with the cases from the proof of \Cref{lemma:k-4-1/3-bound}, an inner candidate wins w.p.\ at least $1-2p + 8  p(1/2 - F_{4, 0}(x/3 + 1/3))^3$. By symmetry, this means 
\begin{align*}
  F_{4, t}(x) &\le \left[1 - (1-2p +  8  p(1/2 - F_{4, 0}(x/3 + 1/3))^3 )\right] / 2\\
  &= \left[2p -  8  p(1/2 - F_{4, 0}(x/3 + 1/3))^3\right] / 2\\
    &= p\left[1 -  4  (1/2 - F_{4, 0}(x/3 + 1/3))^3 \right]\\
  &= F_{4, t-1}(x) \cdot \left[1 -  4  (1/2 - F_{4, 0}(x/3 + 1/3))^3 \right].
\end{align*}
The claim then follows by induction on $t$. 
\end{proof}

\largek*
\begin{proof}
  Suppose $F_{k, t-1}(1/4) \le \alpha$ for some small $\alpha$. Let $x \in (1/4, 1/2)$ and $F_{k, t-1}(x) = p$, so $F_{k, t-1}(x) - F_{k, t-1}(1/4) \ge p - \alpha$. We'll lower bound the probability that the winner is is an \emph{outer} candidate outside of $(x, 1-x)$, focusing mainly on cases where all candidates are in $(1/4, 3/4)$ so we can apply \Cref{lemma:1/4-3/4}.

  If all candidates are in $[0, 1/2)$, then an outer candidate only wins if all candidates are left of $x$, which occurs w.p.\ $p^k$. Accounting for the symmetric case gives an outer candidate win probability of $2p^k$ when all candidates are on the same side. Now suppose there is at least one candidate on each side. If the left- and rightmost candidates are in $(1/4, x)$ and $(1-x, 3/4)$, respectively, then an outer candidate wins by \Cref{lemma:1/4-3/4}. We can find the probability this occurs as the probability that all candidates are in $(1/4, 3/4)$ minus the probability that all candidates are in $(1/4, 1-x]$ or in $[x, 3/4)$---since this means there is at least one candidate each in $(1-x, 3/4)$ and $(1/4, x)$. Since $F_{k, t-1}(1/4) \le \alpha$, the following is a lower bound on the probability the leftmost candidate is at $X_{1,t} \in (1/4, x)$ and the rightmost is at $X_{k,t} \in (x-1, 3/4)$:
  \begin{align*}
    &\Pr(X_{1,t} \in (1/4, x), X_{k,t} \in (x-1, 3/4))\\
    &\ge\underbrace{(1-2\alpha)^k}_\text{all in $(1/4, 3/4)$} - \Big[\underbrace{(1-\alpha - p)^k}_\text{all in $(1/4, 1-x]$} + \underbrace{(1-\alpha - p)^k}_\text{all in $[x, 3/4)$} - \underbrace{(1-2p)^k}_\text{all in $[x, 1-x]$} \Big]\tag{by inclusion--exclusion}\\
    &= (1-2\alpha)^k -2(1-\alpha - p)^k + (1-2p)^k.
  \end{align*}
  Combining this with the case where all candidates are on the same side (and then dividing by 2 to account for symmetry) yields a lower bound on $F_{k, t}(x)$:
  \begin{align}
    F_{k, t}(x) &\ge \left[2p^k + (1-2\alpha)^k -2(1-\alpha - p)^k + (1-2p)^k\right] / 2\notag\\
    &= p^k + (1-2\alpha)^k/2-(1-\alpha - p)^k + (1-2p)^k/2.\label{eq:large-k-outer-lower-bound}
  \end{align}

  We can now use this bound to prove non-convergence. Suppose for a contradiction that $\lim_{t \rightarrow \infty} F_{k, t}(x) = 0$ for all $x < 1/2$. Then there exists some $t^*$ such that $F_{k, t}(1/4) \le \alpha = \left[1 - (499/512)^{1/k}\right]/2$ for all $t > t^*$. But now consider $x^* = F^{-1}_{k, t^*}(1/4) < 1/2$. Since $F_{k, t}(1/4) \le \alpha$ for all $t > t^*$, we can use the fact above to show inductively that $F_{k, t}(x^*) \ge  1/4$ for all $t \ge t^*$. For the base case $t = t^*$, the claim is vacuously true: $F_{k, t^*}(x^*) = 1/4 \ge  1/4$. Now suppose for $t > t^*$ that $F_{k, t-1}(x^*) \ge 1/4$. Then $z = F_{k, t-1}^{-1}(1/4) \le x^*$. From~\eqref{eq:large-k-outer-lower-bound}, we then have:
  \begin{align*}
    F_{k, t}(z)&\ge 1/4^k + (1-2\alpha)^k/2-(1-\alpha - 1/4)^k + (1-2 /4)^k/2\\
    &> (1-2\alpha)^k/2-(3/4-\alpha)^k \tag{throw away terms}\\
    &\ge (1-2\alpha)^k/2-(3/4)^5\tag{since $k \ge 5$, $\alpha > 0$}\\
    &= \left(1-2\left[1 - (499/512)^{1/k}\right]/2\right)^k / 2 -(3/4)^5\tag{plug in $\alpha$}\\
    &= 499/1024 -243/1024\\
    &= 1/4.
  \end{align*}
  By the monotonicity of the CDF, $F_{k, t}(x^*) > 1/4$, since $z \le x^*$. By induction, $F_{k, t}(x^*) \ge 1/4$ for all $t \ge t^*$ This contradicts that $\lim_{t \rightarrow \infty} F_{k, t}(x) = 0$ for all $x < 1/2$.
\end{proof}

\subsection{Proofs from \Cref{sec:noisy}}\label{app:noisy-proofs}

Our proofs with $\epsilon$-uniform noise make extensive use of the following lemma, which allows us to translate the convergence of an iterated map bounding a sequence into an eventual bound on the sequence. 

\begin{lemma}\label{lemma:bounded-convergence}
  Consider an iterated map $x_t = f(x_{t-1})$ where $f:[0, 1/2) \rightarrow [0, 1/2)$ is non-decreasing. Suppose $\lim_{t \rightarrow \infty} x_t = c$ for all $x_0 \in I \subseteq [0, 1/2)$.
  
  \begin{enumerate}
  \item If $y_t \le f(y_{t-1})$ for all $t > 0$, then $\limsup_{t \rightarrow \infty} y_t \le c$ for all $y_0 \in I$. 
  \item  If $y_t \ge f(y_{t-1})$ for all $t>0$, then $\liminf_{t \rightarrow \infty} y_t \ge c$ for all $y_0 \in I$. 
  \end{enumerate}

\end{lemma}
\begin{proof}
  Let $y_0 \in I$ and define $x_0 = y_0$. Suppose $y_t \le f(y_{t-1})$ for all $t > 0$. We'll show $y_t \le x_t$ by induction. The base case $t = 0$ holds by the definition of $x_0$. Suppose for $t > 0$ that $y_{t-1} \le x_{t-1}$. Then $y_t \le f(y_{t-1}) \le f(x_{t-1}) = x_t$ (since $f$ is non-decreasing), so $y_t \le x_t$ for all $t$ by induction. Thus, $\limsup_{t \rightarrow \infty}  y_t \le \limsup_{t \rightarrow \infty} x_t = \lim_{t \rightarrow \infty} x_t = c$. The second claim with $y_t \ge f(y_{t-1})$ follows from the exact same argument with each $\le$ replaced by $\ge$ and $\limsup$ replaced by $\liminf$. 
\end{proof}

\ktwonoisylemma*
\begin{proof}
  We begin by looking for the fixed points of the map:
  \begin{align*}
    &2p^2 (1-\epsilon)^2 + 4px\epsilon(1-\epsilon) + 2x^2 \epsilon^2 = p\\
    &\Leftrightarrow  2 (1-\epsilon)^2p^2 + (4x\epsilon(1-\epsilon)-1)p + 2x^2 \epsilon^2 = 0.
  \end{align*}
  Applying the quadratic formula and simplifying yields the two fixed points:
  \begin{align*}
    p_1^*&= \frac{1 - 4x \epsilon(1-\epsilon) -  \sqrt{1 - 8 \epsilon x (1 - \epsilon )}}{4 (1-\epsilon)^2}\\
    p_2^*&=\frac{1 - 4x \epsilon(1-\epsilon) +  \sqrt{1 - 8 \epsilon x (1 - \epsilon )}}{4 (1-\epsilon)^2}.
  \end{align*} 
We'll show that $p^*_1$ is stable  and that $p_1^* \le \epsilon$ while  $p_2^*$ is unstable and $p_2^*> 1/2$. To see that $p_1^* \le \epsilon$, consider $\epsilon - p_1^*:$
  \begin{align}
    \epsilon - \frac{1 - 4x \epsilon(1-\epsilon) -  \sqrt{1 - 8 \epsilon x (1 - \epsilon )}}{4 (1-\epsilon)^2} &= \frac{4 (1-\epsilon)^2\epsilon -1 + 4x \epsilon(1-\epsilon) +  \sqrt{1 - 8 \epsilon x (1 - \epsilon )}}{4 (1-\epsilon)^2}.\label{eq:epsilon-minus-p1}
  \end{align}
  It suffices to show the numerator is non-negative. Taking its derivative with respect to $x$ shows the numerator is decreasing in $x$: 
\begin{align*}
   \frac{\partial}{\partial x} \left[ 4(1-\epsilon)^2\epsilon -1 + 4x \epsilon(1-\epsilon) +  (1 - 8 \epsilon x (1 - \epsilon ))^{1/2} \right] &= 4\epsilon(1-\epsilon) -  \frac{4\epsilon(1-\epsilon)}{(1 - 8 \epsilon x (1 - \epsilon ))^{1/2}}\\
   &< 4\epsilon(1-\epsilon) -  4\epsilon(1-\epsilon)\\
   &= 0.
\end{align*} 
Thus, it suffices to show the function is non-negative when $x = 1/2$. For $x = 1/2$, 
\begin{align*}
   4 (1-\epsilon)^2\epsilon -1 + 4x \epsilon(1-\epsilon) +  \sqrt{1 - 8 \epsilon x (1 - \epsilon )} &= 4 (1-\epsilon)^2\epsilon -1 + 2 \epsilon(1-\epsilon) +  \sqrt{1 - 4 \epsilon (1 - \epsilon )}\\
   &= 4 \epsilon ^3-10 \epsilon ^2+6 \epsilon -1 +  \sqrt{(1-2 \epsilon)^2}\\
   &= 4 \epsilon ^3-10 \epsilon ^2+6 \epsilon -1 +  |1-2 \epsilon|.
\end{align*}
Consider the cases $\epsilon \le 1/2$ and $\epsilon > 1/2$. If $\epsilon \le 1/2$,
\begin{align*}
  4 \epsilon ^3-10 \epsilon ^2+6 \epsilon -1 +  |1-2 \epsilon| &= 4 \epsilon ^3-10 \epsilon ^2+4 \epsilon\\
  &= 2 \epsilon (2-\epsilon)   (1-2 \epsilon)\\
  &\ge 0 \tag{since $\epsilon \le 1/2$}.
\end{align*}
If $\epsilon > 1/2$,
\begin{align*}
   4 \epsilon ^3-10 \epsilon ^2+6 \epsilon -1 +  |1-2 \epsilon| &=  4 \epsilon ^3-10 \epsilon ^2+8 \epsilon -2\\
   &= 2 (1-\epsilon )^2 (2 \epsilon -1)\\
   &\ge 0 \tag{since $\epsilon > 1/2$}
\end{align*}
Therefore the numerator in \eqref{eq:epsilon-minus-p1} is non-negative, so $p_1^* \le \epsilon$. To show $p_1^*$ is stable, consider the derivative of the iterated map in \eqref{eq:k-2-noisy-map}:
\begin{align}
  \frac{\partial}{\partial p}\left[ 2p^2 (1-\epsilon)^2 + 4px\epsilon(1-\epsilon) + 2x^2 \epsilon^2\right] &= 4(1-\epsilon)^2p + 4x \epsilon (1-\epsilon).\label{eq:k-2-noisy-map-derivative}
\end{align}
Plugging in $p_1^*$:
\begin{align*}
  4(1-\epsilon)^2p_1^* + 4x \epsilon (1-\epsilon)&=  4(1-\epsilon)^2\frac{1 - 4x \epsilon(1-\epsilon) -  \sqrt{1 - 8 \epsilon x (1 - \epsilon )}}{4 (1-\epsilon)^2} + 4x \epsilon (1-\epsilon)\\
  &= 1 - \sqrt{1 - 8 \epsilon x (1 - \epsilon )}\\
  &< 1 - \sqrt{1 - 4 \epsilon (1 - \epsilon )} \tag{$x < 1/2$}\\
  &\le 1 \tag{$\epsilon(1-\epsilon) \le 1/4$}.
\end{align*}
Thus the derivative of the iterated map at $p_1^*$ has magnitude strictly less than 1, so $p_1^*$ is a stable fixed point.

Now, consider the other fixed point $p_2^*$:
\begin{align*}
  p_2^*&=\frac{1 - 4x \epsilon(1-\epsilon) +  \sqrt{1 - 8 \epsilon x (1 - \epsilon )}}{4 (1-\epsilon)^2}\\
  &>\frac{1 - 2 \epsilon(1-\epsilon) +  \sqrt{1 - 4 \epsilon  (1 - \epsilon )}}{4 (1-\epsilon)^2}\tag{$x < 1/2$}\\
  &= \frac{1 - 2 \epsilon(1-\epsilon) +  \sqrt{(1-2\epsilon)^2}}{4 (1-\epsilon)^2}\\
  &= \frac{1 - 2 \epsilon(1-\epsilon) +  |1-2\epsilon|}{4 (1-\epsilon)^2}.
\end{align*}
If $\epsilon \le 1/2$,
\begin{align*}
  p_2^* &>\frac{1 - 2 \epsilon(1-\epsilon) +  1-2\epsilon}{4 (1-\epsilon)^2}\\
  &= \frac{2(1-\epsilon)^2}{4 (1-\epsilon)^2}\\
  &= 1/2.
\end{align*}
If $\epsilon > 1/2$,
\begin{align*}
  p_2^* &>\frac{1 - 2 \epsilon(1-\epsilon) -  1+2\epsilon}{4 (1-\epsilon)^2}\\
  &= \frac{2\epsilon^2}{4(1-\epsilon)^2}\\
  &>\frac{2(1/2)^2}{4(1-1/2)^2}\\
  &= 1/2.
\end{align*}
In either case, $p_2^* > 1/2$. Additionally, plugging $p_2^*$ into the derivative \eqref{eq:k-2-noisy-map-derivative} yields $1+  \sqrt{1 - 8 x\epsilon (1 - \epsilon )} > 1$ (for $x < 1/2$), showing $p_2^*$ is unstable. Thus, for $p \in [0, 1/2]$, the quadratic map  converges to the stable fixed point $p_1^*\le \epsilon$.
\end{proof}

\begin{lemma}\label{lemma:k-3-noisy-map}
For any $\epsilon \in (0, 1/3)$, the cubic iterated map given by
  \begin{align*}
    p' &= 3/4 \cdot [\epsilon /2 + (1-\epsilon)p] +[\epsilon /2 + (1-\epsilon)p]^3
  \end{align*}
  converges to $p^* \le 1.5 \epsilon$ for all initial $p \in [0, 1/2)$. Moreover, the map is non-decreasing in $p$ on $[0, 1/2)$.  
\end{lemma}
\begin{proof}
  
The fixed points of this map can be found using the cubic formula (equivalently, we used Mathematica):
\begin{align*}
  p_1^* &= \frac{1}{4} \sqrt{\frac{1+15 \epsilon}{(1-\epsilon)^3}} - \frac{1+2 \epsilon}{4 (1-\epsilon)}\\
  p_2^* &= -\frac{1}{4} \sqrt{\frac{1+15 \epsilon}{(1-\epsilon)^3}} - \frac{1+2 \epsilon}{4 (1-\epsilon)}\\
  p_3^* &= \frac{1}{2}.
\end{align*}
We can ignore the negative fixed point $p_2^*$, since $p$ can never be negative. We'll show that for $\epsilon < 1/3$, $p_1^* \in [0, 1.5 \epsilon]$, $p_1^*$ is stable, and the cubic map converges to $p_1^*$ for $p \in [0, 1/2)$. To begin with, we'll show $p_1^* \ge 0$:

\begin{align*}
   p_1^* &= \frac{1}{4} \sqrt{\frac{1 + 15 \epsilon}{(1 - \epsilon)^3}} - \frac{1 + 2 \epsilon}{4 (1 - \epsilon)}\\
 &= \frac{(1 + 15 \epsilon)^{1/2}}{4(1 - \epsilon)^{3/2}} - \frac{
 (1 + 2 \epsilon)(1 - \epsilon)^{1/2}}{4 (1 - \epsilon)^{3/2}}\\
 &= \frac{(1 + 15 \epsilon)^{1/2} - ((1 + 2 \epsilon)^2)^{1/2}(1 - \epsilon)^{1/2}}{4(1 - \epsilon)^{3/2}}\\
 &= \frac{(1 + 15 \epsilon)^{1/2} - ((1 + 2 \epsilon)^2(1 - \epsilon))^{1/2}}{4(1 - \epsilon)^{3/2}}\\
  &= \frac{(1 + 15 \epsilon)^{1/2} - ( 1 + 3 \epsilon - 4 \epsilon^3)^{1/2}}{4(1 - \epsilon)^{3/2}}\\
  &\ge 0.\tag{since $1 + 15 \epsilon > 1 + 3 \epsilon - 4 \epsilon^3 $}
\end{align*}
Now we'll show that $  p^*_1 \le 1.5 \epsilon$. To do this, we'll show $1.5\epsilon  -   p^*_1 \ge 0$:
\begin{align*}
   1.5 \epsilon - p_1^* &= 1.5 \epsilon  - \frac{(1 + 15 \epsilon)^{1/2} - ( 1 + 3 \epsilon - 4 \epsilon^3)^{1/2}}{4(1 - \epsilon)^{3/2}}\\
    &= \frac{6\epsilon(1 - \epsilon)^{3/2}}{4(1 - \epsilon)^{3/2}}  - \frac{(1 + 15 \epsilon)^{1/2} - ( 1 + 3 \epsilon - 4 \epsilon^3)^{1/2}}{4(1 - \epsilon)^{3/2}}\\
    &= \frac{6\epsilon(1 - \epsilon)^{3/2} - (1 + 15 \epsilon)^{1/2} + ( 1 + 3 \epsilon - 4 \epsilon^3)^{1/2}}{4(1 - \epsilon)^{3/2}}.
\end{align*}
It suffices to show the numerator is non-negative on $[0, 1/3)$:
\begin{align*}
  &6\epsilon(1 - \epsilon)^{3/2} - (1 + 15 \epsilon)^{1/2} + ( 1 + 3 \epsilon - 4 \epsilon^3)^{1/2}\\
  &=   6\epsilon(1-\epsilon)(1 - \epsilon)^{1/2} - (1 + 15 \epsilon)^{1/2} + (1 + 2 \epsilon)(1-\epsilon)^{1/2}\\
  &= (1-\epsilon)^{1/2}\left[6\epsilon(1-\epsilon) + (1 + 2 \epsilon)\right] - (1 + 15 \epsilon)^{1/2}\\
  &= (1-\epsilon)^{1/2}(1 + 8\epsilon-6\epsilon^2) - (1 + 15 \epsilon)^{1/2}\\
  &= \left[(1-\epsilon)(1 + 8\epsilon-6\epsilon^2)^2\right]^{1/2} - (1 + 15 \epsilon)^{1/2}\\
  &= (1 + 15 \epsilon + 36 \epsilon^2 - 148 \epsilon^3 + 
 132 \epsilon^4 - 36 \epsilon^5)^{1/2} - (1 + 15 \epsilon)^{1/2}.
\end{align*}
To show this is non-negative, it suffices to show $36 \epsilon^2 - 148 \epsilon^3 + 
 132 \epsilon^4 - 36 \epsilon^5 $ is non-negative. Factoring yields
\begin{align*}
  36 \epsilon^2 - 148 \epsilon^3 + 
 132 \epsilon^4 - 36 \epsilon^5 &= 4 \epsilon^2 (1 - 3 \epsilon) (9 - 10 \epsilon + 
   3 \epsilon^2).
\end{align*}
Finally, we can see this is non-negative for $\epsilon \in (0, 1/3)$, so $p^*_1 \le 1.5 \epsilon$ for $\epsilon \in (0, 1/3)$.

Now, to show $  p^*_1$ is a stable fixed point, we can take the derivative of the cubic map at $p^*_2$:
\begin{align}
&\frac{\partial}{\partial p} \left(3/4 \cdot [\epsilon /2 + (1-\epsilon)p] +[\epsilon /2 + (1-\epsilon)p]^3 \right)\notag\\
  &=\frac{\partial}{\partial p} \left[p^3 (1-\epsilon)^3+\frac{3}{2} p^2 \epsilon  (1-\epsilon)^2+\frac{3}{4} p (1-\epsilon) \left( \epsilon ^2+1\right)+ \frac{1}{8}\epsilon^3 +\frac{3 }{8}\epsilon \right]\notag\\
  &= 3(1 - \epsilon)^3 p^2  + 3  \epsilon(1 - \epsilon)^2  p  + \frac{3}{4}  (1 - \epsilon) (1 +  \epsilon^2)\label{eq:k-3-map-derivative}
\end{align}
Plugging in $p^*_1$  and simplifying yields
\begin{align*}
&3(1 - \epsilon)^3 \left(\frac{1}{4} \sqrt{\frac{1 + 15 \epsilon}{(1 - \epsilon)^3}} - \frac{1 + 2 \epsilon}{4 (1 - \epsilon)}\right)^2  + 3  \epsilon(1 - \epsilon)^2  \left(\frac{1}{4} \sqrt{\frac{1 + 15 \epsilon}{(1 - \epsilon)^3}} - \frac{1 + 2 \epsilon}{4 (1 - \epsilon)}\right)  + \frac{3}{4}  (1 - \epsilon) (1 +  \epsilon^2)\\
 &= 3(1 - \epsilon)^3 \left(  \frac{1 + 15 \epsilon}{16(1 - \epsilon)^3} + \frac{(1 + 2 \epsilon)^2}{16 (1 - \epsilon)^2}  -  \frac{1 + 2 \epsilon}{8(1 - \epsilon)}\sqrt{\frac{1 + 15 \epsilon}{(1 - \epsilon)^3}}  \right) + \frac{3}{4}   \epsilon(1 - \epsilon)^2  \sqrt{\frac{1 + 15 \epsilon}{(1 - \epsilon)^3}}   \\
 &\qquad -  \frac{3}{4}  \epsilon(1 - \epsilon) (1 + 2 \epsilon)+ 3/4  (1 - \epsilon) (1 +  \epsilon^2)\\
  &=   \frac{3 (1 + 15 \epsilon)}{16} + \frac{3(1 - \epsilon) (1 + 2 \epsilon)^2}{16}  -  \frac{3(1 - \epsilon)^2 (1 + 2 \epsilon)}{8}\sqrt{\frac{1 + 15 \epsilon}{(1 - \epsilon)^3}}  + \frac{3}{4}   \epsilon(1 - \epsilon)^2  \sqrt{\frac{1 + 15 \epsilon}{(1 - \epsilon)^3}}  \\
  &\qquad -  \frac{3}{4}  \epsilon(1 - \epsilon) (1 + 2 \epsilon) + 3/4  (1 - \epsilon) (1 +  \epsilon^2)\\
   &=   \frac{3}{16} (1 + 15 \epsilon) + \frac{3}{16}(1 - \epsilon) (1 + 2 \epsilon)^2  +(1 - \epsilon)^2\left(\frac{3}{4}   \epsilon-  \frac{3}{8} (1 + 2 \epsilon) \right)\sqrt{\frac{1 + 15 \epsilon}{(1 - \epsilon)^3}}  -  \frac{3}{4}  \epsilon(1 - \epsilon) (1 + 2 \epsilon) \\
   &\qquad  + 3/4  (1 - \epsilon) (1 +  \epsilon^2)\\
 &= -\frac{3}{8}(1 - \epsilon)^2\sqrt{\frac{1 + 15 \epsilon}{(1 - \epsilon)^3}} + \frac{9}{8} + \frac{15}{8}\epsilon\\
  &=  -\frac{3}{8}\sqrt{(1 + 15 \epsilon)(1-\epsilon)} + \frac{9}{8} + \frac{15}{8}\epsilon.
\end{align*}
To see this is positive for $\epsilon < 1/3$, note that $\frac{3}{8}\sqrt{(1 + 15 \epsilon)(1-\epsilon)} < \frac{3}{8}\sqrt{(1 + 15 /3)} \approx 0.92 < 9/8$. We can also show the derivative of the cubic map at $p^*_2$ is less than 1. To do this, we'll show that 1 minus the derivative at $p_1^*$ is positive:
\begin{align*}
  1 - \left(-\frac{3}{8}\sqrt{(1 + 15 \epsilon)(1-\epsilon)} + \frac{9}{8} + \frac{15}{8}\epsilon\right) &= \frac{3}{8}\sqrt{(1 + 15 \epsilon)(1-\epsilon)} - \frac{1}{8} - \frac{15}{8}\epsilon\\
  &=\sqrt{\frac{9}{64}(1 + 15 \epsilon)(1-\epsilon)} - \sqrt{\left(\frac{1}{8} + \frac{15}{8}\epsilon\right)^2}.\\
\end{align*}
By the monotonicity of square roots, it suffices to show that the following quadratic is positive:
\begin{align*}
  \frac{9}{64}(1 + 15 \epsilon)(1-\epsilon) - \left(\frac{1}{8} + \frac{15}{8}\epsilon\right)^2 &= -\frac{45 \epsilon ^2}{8}+\frac{3 \epsilon }{2}+\frac{1}{8}\\
  &= \frac{1}{8} (1- 3 \epsilon) (15 \epsilon +1).
\end{align*}
which we can see is positive for $\epsilon \in (0, 1/3)$. Thus, the derivative of the cubic map at $p_1^*$ is positive but less than 1, so $p_1^*$  is a stable fixed point. The fixed point at $1/2$ is unstable, in contrast: plugging $p^*_3 = 1/2$ into the derivative \eqref{eq:k-3-map-derivative} and simplifying yields $3/2 (1-\epsilon)$, which is larger than 1 for $\epsilon < 1/3$.  Thus the cubic map converges to $p_1^*$ for initial values in $[0, 1/2)$. 

Finally, to show the map is non-decreasing in $p$, notice that derivative \Cref{eq:k-3-map-derivative} is non-negative for $p \ge 0$ and $\epsilon \in (0, 1].$ 
\end{proof}

\kthreenoisy*
\begin{proof}
  Let $x < 1/2$ and define $p = F_{3, t-1}^\epsilon(x)$.  With $\epsilon$-uniform noise, $\Pr(X_{i,t} \le x) = \epsilon x + (1-\epsilon)p$. The case analysis from \Cref{thm:k-3-convergence} then proceeds exactly the same way, so we can replace $p$ by $\epsilon x + (1-\epsilon)p$ in the bound from \Cref{thm:k-3-convergence} to get the equivalent bound with $\epsilon$-uniform noise:
  \begin{align}
    F_{3, t}^\epsilon(x) &\le 3/4 \cdot [\epsilon x + (1-\epsilon)p] +[\epsilon x + (1-\epsilon)p]^3.\label{eq:k-3-noisy-recurrence}
  \end{align}
   While it would be possible to work directly with the cubic map \eqref{eq:k-3-noisy-recurrence}, its fixed points are extremely messy. As such, we instead analyze the upper bound given by $x < 1/2$ and then use \Cref{lemma:bounded-convergence}:
\begin{align}
  F_{3, t}^\epsilon(x) &< 3/4 \cdot [\epsilon /2 + (1-\epsilon)p] +[\epsilon /2 + (1-\epsilon)p]^3.\label{eq:k-3-noisy-upper-bound}
\end{align}
If $F_0(x) < 1/2$, then \Cref{lemma:k-3-noisy-map} states that the map \eqref{eq:k-3-noisy-upper-bound} upper bounding $F_{3, t}^\epsilon(x)$ converges to $p^*\le 1.5 \epsilon$. Thus, applying \Cref{lemma:bounded-convergence} gives $\limsup_{t \rightarrow \infty}F_{3, t}^\epsilon(x)\le p^* \le 1.5 \epsilon$ as claimed. If $F_0(x) = 1/2$ (which is possible since we don't require that $F_0$ is positive near $1/2$), then applying \eqref{eq:k-3-noisy-upper-bound},
\begin{align*}
  F_{3, 1}^\epsilon(x) &< 3/4 \cdot [\epsilon /2 + (1-\epsilon)/2] +[\epsilon /2 + (1-\epsilon)/2]^3\\
  &= 3/4 \cdot 1/2 +[1/2]^3\\
  &= 1/2.
\end{align*}
Thus $F_{3, 1}^\epsilon(x) < 1/2$, so we can apply \Cref{lemma:bounded-convergence,lemma:k-3-noisy-map} with initial $p = F_{3, 1}^\epsilon(x)$ rather than $F_{0}(x)$. 
\end{proof}

\kfourlemmanoisy*
\begin{proof}
  Let $x \in (1/3, 1/2)$. Just as in \Cref{thm:k-3-noisy-convergence}, we can take the bound from \Cref{lemma:k-4-1/3-bound} and replace $F_{4, t-1}^\epsilon(x)$ with $\epsilon x + (1-\epsilon)F_{4, t-1}^\epsilon(x)$ to get the claimed upper bound with $\epsilon$-uniform noise. The second part of the claim follows by induction after noting $F_{4, 1}^\epsilon(x) \le \epsilon x +  (1 - \epsilon) F_{4, 0}^\epsilon(x) \le \epsilon \max\{x, F_{4, 0}^\epsilon(x)\} +  (1 - \epsilon) \max\{x, F_{4, 0}^\epsilon(x)\} =  \max\{x, F_{4, 0}^\epsilon(x)\}$.
  \end{proof}

\kfournoisy*

\begin{proof} 
Let $p = F_{4, t-1}^\epsilon(x)$. By \Cref{lemma:k-4-1/3-bound-noisy} and symmetry, an inner candidate in $(x, 1-x)$ wins with probability at least $1-2p(1-\epsilon) - 2x\epsilon$. We'll strengthen this bound in the same way as in \Cref{thm:k-4-convergence}, using the case with three candidates in $(x/3 + 1/3, 1/2)$ and one in $(1-x, 1]$ (and the symmetric case). With $\epsilon$-uniform noise and accounting for symmetry, the probability this case occurs is
\begin{align*}
&8 [\epsilon x + (1-\epsilon)p] [1/2 - \epsilon(x/3 + 1/3) - (1-\epsilon)F_{4, t-1}^\epsilon(x/3 + 1/3)]^3\\
&\ge 8 [\epsilon x + (1-\epsilon)p] [1/2 - \epsilon(x/3 + 1/3) - (1-\epsilon)\max\{x/3 + 1/3, F_{0}(x/3 + 1/3) \}]^3\tag{by \Cref{lemma:k-4-1/3-bound-noisy}}\\
&= 8 [\epsilon x + (1-\epsilon)p] \beta^3.
\end{align*}
Then, adding this case to the cases implicitly used in \Cref{lemma:k-4-1/3-bound-noisy} (see \Cref{lemma:k-4-1/3-bound} for the list of cases), we find 
\begin{align*}
  &\Pr(x < \plurality(X_{1,t}^\epsilon, X_{2,t}^\epsilon, X_{3,t}^\epsilon, X_{4,t}^\epsilon) < 1-x)\\
  &\ge 1-2p(1-\epsilon) - 2x\epsilon + 8 [\epsilon x + (1-\epsilon)p] \beta^3.
\end{align*}
By symmetry,
\begin{align}
  F_{4, t}^\epsilon(x) &= \left[1 - \Pr(x < \plurality(X_{1,t}^\epsilon, X_{2,t}^\epsilon, X_{3,t}^\epsilon, X_{4,t}^\epsilon) < 1-x)\right] / 2\notag \\
  &\le p(1-\epsilon) +x\epsilon -4 [\epsilon x + (1-\epsilon)p] \beta^3\notag \\
  &= p (1-\epsilon) (1 -4\beta^3  ) + \epsilon x (1 - 4 \beta^3).\label{eq:k-4-noisy-map}
\end{align}
We'll show that the iterated map \eqref{eq:k-4-noisy-map} upper bounding $F_{4, t}^\epsilon(x)$ converges to a fixed point upper bounded by $\frac{1}{8\beta} \epsilon$ and then apply \Cref{lemma:bounded-convergence}. First, we'll find the fixed point (unique, since this is a linear map):
\begin{align}
   &p^* (1-\epsilon) (1 -4\beta^3  ) + \epsilon x (1 - 4 \beta^3) = p^* \notag \\
  \Leftrightarrow \quad &p^*  [(1-\epsilon) (1 -4\beta^3  )-1] + \epsilon x (1 - 4 \beta^3) = 0\notag \\
    \Leftrightarrow \quad &p^*   = \frac{\epsilon x (1 - 4 \beta^3)}{1-(1-\epsilon) (1 -4\beta^3  )}.
\end{align}
To show convergence to $p^*$, it suffices to show that the slope of the map is in $(-1, 1)$ (any such linear map converges to its unique fixed point, e.g., by the Banach fixed-point theorem). First, we can show $\beta \in (0, 1/2]$: 
\begin{align*}
  \beta &= 1/2 - \epsilon(x/3 + 1/3) - (1-\epsilon)\max\{x/3 + 1/3, F_{0}(x/3 + 1/3) \}\\
  &> 1/2 - \epsilon(1/2) - (1-\epsilon)\max\{1/2, F_{0}(x/3 + 1/3) \} \tag{since $x < 1/2$}\\
  &= 1/2 - \epsilon(1/2) - (1-\epsilon)(1/2) \tag{since $F_{0}(x/3 + 1/3) \le 1/2$}\\
  &= 0.
\end{align*}
Thus, the slope $(1-\epsilon)(1-4\beta^3) \in [0, 1)$, so the map \eqref{eq:k-4-noisy-map} converges to $p^*$ for all initial values $p$ and is non-decreasing in $p$. Now we can upper bound $p^*$:

\begin{align*}
  p^* &= \frac{\epsilon x (1 - 4 \beta^3)}{1-(1-\epsilon) (1 -4\beta^3  )}\\
  &<  \frac{\epsilon /2}{1- (1 -4\beta^3  )}\tag{$\epsilon \ge 0$, $x < 1/2$}\\
  &= \frac{\epsilon }{8\beta^3  }.
\end{align*}
Thus, by \Cref{lemma:bounded-convergence}, $\limsup_{t\rightarrow \infty}F_{4, t}^\epsilon(x) \le p^* <  \frac{1}{8\beta^3} \epsilon$.
\end{proof}

\largeknoisy*

\begin{proof}
 Suppose for a contradiction that the candidate distribution does approximately converge to the center. That is, suppose that for all $c > 0$ and $x <1/2$, there exists some $\epsilon_\text{max} > 0$ such that with $\epsilon$-uniform noise, for any $\epsilon \in (0, \epsilon_{\text{max}}]$,  $\limsup_{t \rightarrow \infty} F_{k, t}^\epsilon(x) < c$. If $\limsup_{t \rightarrow \infty} F_{k, t}^\epsilon(x) < c$, then there is some $t^*$ such that for all $t \ge t^*$, $F_{k, t}^\epsilon(x) \le c$. In particular, let $\epsilon_\text{max}^*$ and $t^*$ be the corresponding values for $x = 1/4$. Then for any $\epsilon$-uniform noise with $\epsilon \le \epsilon_\text{max}^*$, $F_{k, t^*}^\epsilon(1/4) \le c$. Additionally, consider the point $z = (F_{k, t^*}^\epsilon)^{-1}(1/4)$. By our assumption, there is some $\epsilon_\text{max}'$ and $t'$ such that if $\epsilon < \epsilon_\text{max}'$, then $F_{k, t}^\epsilon(z) \le c$ for all $t \ge t'$. We can make $\epsilon$ and $c$ as small as needed, so we'll pick:
 \begin{align}
   c &< \big[1-\left(125 / 128\right)^{1/k}\big]/3\\
   \epsilon &< \min\big\{\epsilon_\text{max}^*, \epsilon_\text{max}',  \big[1-\left(125 / 128\right)^{1/k}\big]/3\big\}.
 \end{align}
  Note that $\big[1-\left(125 / 128\right)^{1/k}\big]/3$ is largest at $k = 5$, when its value is approximately $0.0016$. Also note that we must have $t' > t^*$, since $F_{k, t^*}^\epsilon(z) = 1/4 > c$.
  
  Now, we can apply the same argument as in \Cref{thm:large-k-no-convergence}, finding a lower bound on $F_{4, t^*+1}^\epsilon(x)$ (for $x \in (1/4, 1/2)$) given that only a $c$-fraction of the winners in generation $t^*$ are left of $1/4$ (note that this parameter was called $\alpha$ in the proof of \Cref{thm:large-k-no-convergence}). Let $p = F_{k, t^*}^\epsilon(x)$ with $\epsilon$-uniform noise. For brevity, we will avoid repeating the argument from \Cref{thm:large-k-no-convergence} and instead substitute directly into the resulting bound~\eqref{eq:large-k-outer-lower-bound}. Replacing $p$ with $\Pr(X_{i, t^*}^\epsilon \le x) = \epsilon x + (1-\epsilon)p$ and $\alpha$ with $\Pr(X_{i, t^*}^\epsilon \le 1/4) \le \epsilon /4 + (1-\epsilon)c $ in~\eqref{eq:large-k-outer-lower-bound} then yields
  \begin{align}
    F_{4, t^*+1}^\epsilon(x) &\ge [\epsilon x + (1-\epsilon)p]^k + (1-2[\epsilon /4 + (1-\epsilon)c])^k/2\notag \\
    &\qquad-(1-[\epsilon /4 + (1-\epsilon)c] - [\epsilon x + (1-\epsilon)p])^k  + (1-2[\epsilon x + (1-\epsilon)p])^k/2.\label{eq:large-k-epsilon-uniform-lower-bound}
  \end{align}

We will now derive a contradiction: that $F_{k, t}^\epsilon(z)$ never goes below $1/4$ as $t $ increases from $t^*$ to $t'$, when it should go below $c$. We know $z > 1/4$ by the monotonicity of the CDF, since $F_{k, t^*}^\epsilon(1/4) \le c$. So, we can apply the lower bound~\eqref{eq:large-k-epsilon-uniform-lower-bound} to $z$ (where $p =1/4$):
 \begin{align}
F_{k,t^*+1}^\epsilon(z)&\ge [\epsilon z + (1-\epsilon)/4]^k + (1-2[\epsilon /4 + (1-\epsilon)c])^k/2\notag \\
    &\qquad-(1-[\epsilon /4 + (1-\epsilon)c] - [\epsilon z + (1-\epsilon)/4])^k  + (1-2[\epsilon z + (1-\epsilon)/4])^k/2 \notag\\
    &> (1-2[\epsilon /4 + (1-\epsilon)c])^k/2 -(1-[\epsilon /4 + (1-\epsilon)c] - [\epsilon z + (1-\epsilon)/4])^k \notag\\
    &= (1-\epsilon /2 -2(1-\epsilon)c)^k/2 -(1-[\epsilon /4 + (1-\epsilon)c] - [\epsilon z + (1-\epsilon)/4])^k \notag\\
    &> (1-\epsilon  -2c)^k/2 -(1 -  (1-\epsilon)/4)^k\notag \\
    &\ge (1-\epsilon  -2c)^k/2 -(3/4+\epsilon/4)^5.\label{eq:lower-bound-on-z}
  \end{align}
  Now consider each term in~\eqref{eq:lower-bound-on-z}. By our upper bounds on $c$ and $\epsilon$,  \begin{align*}
    (1-\epsilon-2c)^k/2 &> \left(1 - \big[1-\left(125 / 128\right)^{1/k}\big]/3 - 2 \big[1-\left(125 / 128\right)^{1/k}\big]/3\right)^k/2\\
    &= \left(1 - \big[1-\left(125 / 128\right)^{1/k}\big]\right)^k/2\\
    &= 125 / 256.
  \end{align*}

 Meanwhile, $\epsilon$ is also small enough that $(3/4 + \epsilon/4)^5 < \frac{244}{1024}$ (solving for $\epsilon$ reveals we need $\epsilon < 0.0024$, which we have ensured). This then means that $(1-\epsilon - 2c)^k/2 - (3/4 + \epsilon/4)^5 > \frac{125}{256} - \frac{244}{1024} = 1/4$. Following the above chain of inequalities, this shows that $F_{k,t^*+1}^\epsilon(z)>1/4$. 
  
  We will now show by induction that for all $t > t^*$, $F_{k, t}^\epsilon(z) > 1/4$, a contradiction (since by our assumption, $F_{k, t'}^\epsilon(z) \le c$ with $t' > t^*$). We have just shown the base case $t = t^*+1$ above. Then, suppose as an inductive hypothesis that $F_{k, t}^\epsilon(z) \ge 1/4$ for $t > t^*$. Define $w = (F_{k, t}^\epsilon)^{-1}(1/4)$; by the inductive hypothesis, we know $w \le z$.  By the argument above, $F_{k, t+1}^\epsilon(w) > 1/4$ (the argument only requires that $F_{k, t}^\epsilon(1/4) \le c$ and $w \in (1/4, 1/2)$, which are satisfied here). Thus, by the monotonicity of the CDF, $F_{k, t+1}^\epsilon(z) > 1/4$. By induction, we then have $F_{k, t'}^\epsilon(z) > 1/4$, a contradiction.
\end{proof}

\subsection{Proofs from \Cref{sec:1/4-3/4}}\label{app:1/4-3/4-proofs}

Before proving \Cref{thm:1/4-3/4}, we need a few supporting lemmas.
We begin with a result from the proof of \Cref{thm:large-k-no-convergence}, specialized to the case when all candidates are in $(1/4, 3/4)$. 

\begin{restatable}{lemma}{largekboundedsuppmap}
\label{lemma:large-k-1/4-3/4-map}
  Suppose $F_0\in \mathcal F$ is supported on $(1/4, 3/4)$. For $k \ge 5$ and $x \in (1/4, 1/2)$,
  \begin{equation*}
    F_{k, t}(x) \ge 1/2 + F_{k, t-1}(x)^k - (1-F_{k, t-1}(x))^k + (1-2F_{k, t-1}(x))^k / 2.
  \end{equation*}
\end{restatable}

\begin{proof}
We can use the same argument as in \Cref{thm:large-k-no-convergence} to find a lower bound on $F_{k, t}(x)$, but now $F_{k, t-1}(x) \le \alpha = 0$ since $F_0$ (and therefore all subsequent $F_{k, t}$) is supported only on $(1/4, 3/4)$. Plugging $\alpha = 0$ into~\eqref{eq:large-k-outer-lower-bound} yields the claim, noting $p = F_{k, t-1}(x)$.  
\end{proof}

This gives us an iterated map which bounds $F_{k, t}(x)$ from below. We can show that this map converges to $1/2$ in a large interval around $1/2$, meaning that the candidate distribution converges to one with no mass in this interval. We cannot give an explicit form for the basin of attraction of this map since it depends on a root of a polynomial of order $k$, but we can show the interval grows in $k$ and characterize it for $k = 5$. 

\begin{restatable}{lemma}{limitedsupport}
\label{lemma:large-k-1/4-3/4-map-convergence}
  For all $k \ge 5$, the iterated map given by $p' = 1/2 + p^k - (1-p)^k + (1-2p)^k/2$ converges to $1/2$ for all initial $p \in ([1-\sqrt{3/7}]/2= 0.172\dots, 1/2]$. Moreover, this map in non-decreasing in $p$ on $[0, 1/2)$.
\end{restatable}

\begin{proof}
First, we'll show $1/2$ is a stable fixed point of the map. Indeed, $1/2 + (1/2)^k - (1-1/2)^k + (1- 2 (1/2))^k / 2 = 1/2 + 1/2^k - 1/2^k = 1/2$. The stability of this fixed point is determined by the derivative
  \begin{align}
  \frac{\partial}{\partial p}\left( 1/2 + p^k - (1-p)^k + (1-2p)^k / 2\right) &= k p^{k-1} + k(1-p)^{k-1} - k(1-2p)^{k-1}.\label{eq:large-k-map-derivative}
\end{align}
At $1/2$, the derivative is $k (1/2)^{k-1} + k(1-1/2)^{k-1} - k(1-1)^{k-1} = k (1/2)^{k-2}$. For $k \ge 5$, $k (1/2)^{k-2} < 1$, showing the fixed point is stable. 

For $k=5$, we can find the other fixed points of the map by factoring:
\begin{align*}
 & 1/2 + p^5 - (1-p)^5 + (1- 2 p)^5 / 2 = p\\
  &\Leftrightarrow 1/2 + p^5 - (1-p)^5 + (1- 2 p)^5 / 2 - p = 0\\
  &\Leftrightarrow p(1-p)  (1-2 p) \left(-7 p^2+7 p-1\right) = 0.\\
  &\Leftrightarrow p \in \left\{0, (1 - \sqrt{3/7}) / 2, 1/2, (1 + \sqrt{3/7}) / 2, 1 \right\} = \left\{ 0, 0.172\dots, 0.5, 0.827\dots, 1\right\} .
\end{align*}

Plugging in the $k = 5$ fixed point $ (1 - \sqrt{3/7}) / 2 = 0.172\dots$ to the derivative~\eqref{eq:large-k-map-derivative} yields $ \approx 1.43$, so this fixed point in unstable. Next, note that the map monotonically increases in $p$ for $p \in (0, 1/2)$, since the derivative~\eqref{eq:large-k-map-derivative} is positive (as $1 - p > 1-2p$; similarly, the map is non-increasing on $[0, 1/2)$, as claimed). Thus, for $k= 5$, the map is larger than $p$ but smaller than $1/2$ for $p$ in $([1 - \sqrt{3/7}] / 2, 1/2)$ and initial values in this range converge to the stable fixed point $1/2$. 

The final step is to show the map is increasing in $k$ for $p \in (0, 1/2)$, which means that the basin of attraction only grows in $k$. To do this, consider the derivative of the map with respect to $k$:
\begin{align*}
   \frac{\partial}{\partial k}\left( 1/2 + p^k - (1-p)^k + (1-2p)^k / 2\right) &= \frac{1}{2} (1-2 p)^k \log (1-2 p) -(1-p)^k \log (1-p) + p^k \log p.
\end{align*} 
We establish this is positive in the following lemma. 
\begin{lemma}
  For all $k \ge 3$ and $0 < p < 1/2$,
    \begin{equation*}
    1/2 (1-2p)^k \log (1-2p) - (1-p)^k \log(1-p) + p^k \log p > 0.
  \end{equation*}
\end{lemma}
\begin{proof}
We thank River Li\footnote{\url{https://math.stackexchange.com/q/4789978}} for a key idea behind this analysis, based on the following integral trick:
  \begin{align*}
   \int_{0}^1\frac{x-1}{1+t(x-1)}\, dt &= \log(1 + t(x-1)) \big|_{t=0}^1\\
   &= \log(1 + 1(x-1)) - \log(1 + 0(x-1))\\
   &= \log x.
  \end{align*}
  Now, apply this identity to the function in question for $k=3$:
  \begin{align*}
    &1/2 (1-2p)^3 \log(1-2p) - (1-p)^3 \log (1-p) + p^3 \log p\\
    &= 1/2 (1-2p)^3\int_{0}^1 \frac{-2p}{1 + t(-2p)} \,dt \,-(1-p)^3  \int_{0}^1 \frac{-p}{1 + t(-p)} \,dt \, + p^3\int_{0}^1 \frac{p-1}{1 + t(p-1)} \,dt \,\\
    &= \int_{0}^1 \left( -\frac{p(1-2p)^3}{1 -2pt} + \frac{p(1-p)^3}{1 -pt} -   \frac{p^3(1-p)}{1 + pt-t}\right) \, dt.
  \end{align*}
  We'll show that the integrand is positive for all $t \in [0, 1]$, which implies the integral is also positive. Converting to a common denominator,
  \begin{align*}
    &-\frac{p(1-2p)^3}{1 -2pt} + \frac{p(1-p)^3}{1 -pt} -   \frac{p^3(1-p)}{1 + pt-t}\\
    &= \frac{-p(1-2p)^3 (1 -p t)(1 + p t-t) + p(1-p)^3(1 -2p t)(1 + p t-t) - p^3(1-p)(1 -2p t)(1 -p t)}{(1 -2p t)(1 -p t)(1 + p t-t)}.
  \end{align*}
  Since $0 < p < 1/2$ and $0 \le t \le 1$, the denominator is positive, so we just need to show the numerator is positive. We can factor:
  \begin{align*}
    &-p(1-2p)^3 (1 -p t)(1 + p t-t) + p(1-p)^3(1 -2p t)(1 + p t-t) - p^3(1-p)(1 -2p t)(1 -p t)\\
    &= p^2 (1 - 2 p) (3 - 4 p - 4 t + 4 p t + p^2 t + t^2 + p t^2 - 
   4 p^2 t^2 + 2 p^3 t^2)\\
   &= p^2 (1 - 2 p) 
   \left[(1 + p -4p^2 + 2p^3)t^2 - (4 - 4p - p^2) t + 3 - 4p\right].
  \end{align*}
  Again, since $p^2$ and $(1-2p)$ are positive, we just need to show the right factor is positive. Now, notice that $4 - 4p - p^2 > 0$ (since $p < 1/2$), and $t \le \frac{t^2+1}{2}$, so
  
  \begin{align*}
    (1 + p -4p^2 + 2p^3)t^2 - (4 - 4p - p^2) t + 3 - 4p &\ge (1 + p -4p^2 + 2p^3)t^2 - (4 - 4p - p^2) \frac{t^2 + 1}{2} + 3 - 4p\\
    &=\frac{1}{2} \left[2 - 4 p + p^2 - ( 2 - 6 p + 7 p^2 - 4 p^3) t^2\right]
  \end{align*}
  Now, $2 - 6 p + 7 p^2 - 4 p^3$ is positive for $p \in (0, 1/2)$. We can see this since its derivative, $-6 + 14p - 12p^2$ is negative (achieving a maximum of $-23/12$ at $p=7/12$) and the polynomial has a zero at $p = 1/2$. Thus, we can shrink the function by replacing $t^2$ by $1$:
  \begin{align*}
    \frac{1}{2} \left[2 - 4 p + p^2 - ( 2 - 6 p + 7 p^2 - 4 p^3) t^2\right] &\ge \frac{1}{2} \left[2 - 4 p + p^2 - ( 2 - 6 p + 7 p^2 - 4 p^3)\right] \\
    &= p(1 - p) (1 - 2 p).
  \end{align*}
  Finally, we see that this is positive for all $p \in (0, 1)$, which implies that
  \begin{align*}
    1/2 (1-2p)^3 \log(1-2p) - (1-p)^3 \log (1-p) + p^3 \log p > 0.
  \end{align*}
  
  We can now use this as a base case $k=3$ in an inductive argument. For the inductive case ($k \ge 3$), suppose 
  \begin{align*}
    1/2 (1-2p)^k \log (1-2p) - (1-p)^k \log(1-p) + p^k \log p > 0.
  \end{align*}
  Note that the first and third terms are negative, while the middle term is positive (because of the logs). So, let $x = \min\{1/2 (1-2p)^k \log (1-2p), p^k \log p\}$. We then have 
  \begin{align*}
  &1/2 (1-2p)^{k + 1} \log (1-2p) - (1-p)^{k+1} \log(1-p) + p^{k+1} \log p  \\
  & \ge (1-2p)x - (1-p)^{k+1} \log(1-p)   + px \tag{replace both terms by their minimum}\\
  &= (1-p)x - (1-p)(1-p)^{k} \log(1-p)\\
  &\ge (1-p)1/2 (1-2p)^k \log (1-2p) - (1-p)(1-p)^{k} \log(1-p) + (1-p)p^k \log p\\
  &= (1-p)\left[1/2 (1-2p)^k \log (1-2p) - (1-p)^k \log(1-p) + p^k \log p \right]\\
  &> 0 \tag{by IH}.
  \end{align*}
The claim then holds for all $k \ge 3$ by induction.
\end{proof}
Therefore, since the map only increases in $k$, the basin of attraction for the stable fixed point at $1/2$ can only grow as $k$ increases from $5$.
\end{proof}

\boundedsupp*

\begin{proof}
    Applying \Cref{lemma:bounded-convergence} to the bound from \Cref{lemma:large-k-1/4-3/4-map} and the convergence and monotonicity from \Cref{lemma:large-k-1/4-3/4-map-convergence} gives $\liminf_{t \rightarrow \infty} F_{k, t}(x) \ge 1/2$. Meanwhile,  $ F_{k, t}(x) \le 1/2$ for all $t$ by symmetry, so $\limsup_{t \rightarrow \infty} F_{k, t}(x) \le 1/2$. Therefore $\lim_{t \rightarrow \infty} F_{k, t}(x) = 1/2$.
\end{proof}

\boundedsuppdensity*

\begin{proof}
By \Cref{lemma:1/4-3/4}, only the left- or rightmost candidate can win. Thus, if a candidate at $1/2$ is the winner, all other candidates must either be left of $1/2$ or right of $1/2$. Moreover, if all other candidates are on one side, then a candidate at $1/2$ wins. Thus, a candidate at $1/2$ wins if and only if all other candidates fall on the left or the right.  Note that multiple candidates are at $1/2$ with probability 0, since the candidate distribution is atomless. By symmetry, this occurs with probability $2\cdot (1/2)^{k-1} = (1/2)^{k-2}$. Therefore $\Pr(\plurality(1/2, X_{2,t}, \dots, X_{k,t})) = (1/2)^{k-2}$. By \Cref{eq:pdf-iteration}, we then have
  \begin{align*}
    f_{k, t}(1/2) &= k  \cdot f_{k, t-1}(1/2) \cdot \Pr(\plurality(1/2, X_{2,t}, \dots, X_{k,t}))\\
    &= k  \cdot f_{k, t-1}(1/2) \cdot(1/2)^{k-2}.
  \end{align*}
We can now prove the claim by induction on $t$. For $t=0$, indeed $f_{k, 0}(1/2) = f_{k, 0}(1/2) \cdot \left[k (1/2)^{k-2}\right] ^0$. For $t \ge 1$, applying the inductive hypothesis to the above inequality yields
\begin{align*}
  f_{k, t}(1/2) &=  k  \cdot f_{k, t-1}(1/2) \cdot(1/2)^{k-2}\\
  & =  k  \cdot f_{k, 0}(1/2) \cdot \left[k (1/2)^{k-2}\right]^{t-1}\cdot(1/2)^{k-2}\\
  & =  f_{k, 0}(1/2) \cdot \left[k (1/2)^{k-2}\right]^{t}.\qedhere
\end{align*}
\end{proof}

\subsection{Proofs from \Cref{sec:nash}}\label{app:nash-proofs}

\centersmsne*

\begin{proof}
     Let $p'$ denote the mass at $1/2$ in generation $t+1$. If all $k$ candidates are at $1/2$, then so is the winner. Similarly, if all but one candidate are at $1/2$, then the lone deviant loses with vote share less than $1/2$ (with left--right tie-breaking). Thus, $p' \ge p^k + k p^{k - 1} (1-p)$. For any $k$, this lower bound is larger than $p$ for $p$ sufficiently close to 1. To see this, take the derivative at $p=1$: $\frac{d}{dp}\left[ p^k + k p^{k - 1} (1-p)\right] = kp^{k-1}+ k(k-1)p^{k-2} - k^2p^{k-1}$, which is $0$ at $p=1.$ Thus, for any small enough $\epsilon$, $p^k + k p^{k - 1} (1-p)$ is larger than $1 - \epsilon$ when evaluated at $1- \epsilon$. Thus, $p$ will converge to 1 by the monotone convergence theorem. 
\end{proof}

\twospikesmsne*
\begin{proof}
  By symmetry, every candidate has a $1/k$ win probability if they all follow this strategy. Suppose a deviant chooses a distribution that is supported on a point besides $x$ and $1-x$. If they choose a point between $x$ and $1-x$, they lose unless all other candidates pick the same side, which occurs w.p.\ $2(1/2)^{k-1} = 1/2^{k - 2}$. For $k \ge 4$, this is at most $1/k$ (and strictly less for $k > 4$), so sampling points in $(x, 1-x)$ does not increase with probability. Alternatively, if the deviant samples a point in $[0, x)$ (or symmetrically, $(1-x, 1)$), they certainly lose unless no candidates pick $x$, which occurs with probability $1/2^{k-1}$---smaller than $1/k$ for $k \ge 4$. Thus deviating to a point left of $x$ only hurts. Combining the above findings, a deviant does not benefit by sampling any point other than $x$ or $1-x$. Finally, a deviant does not benefit by changing the probability with which they sample either point by symmetry of the other candidates' choices. Since no deviation is beneficial, the strategy is a Nash equilibrium.
\end{proof}

\twospikeconvergence*
\begin{proof}
  Let $p'$ denote the mass at $x$ in generation $t+1$. If all candidates are at $x$ or $1-x$ (w.p.\ $(2p)^k$), then a candidate at $x$ wins with probability $1/2$ by symmetry. Alternatively, suppose all but one candidate are at $x$ or $1-x$. The probability that there is at least one candidate at both $x$ and $1-x$ and a wildcard is $k(1-2p)((2p)^{k-1} - 2p^{k-1})$. In such a case, the wildcard loses if they are in the middle (since they get vote share less than $1/4$) and they lose if they are on the outside (to the opposite outside candidate). Thus, $p' \ge (2p)^k/2 + k(1-2p)((2p)^{k-1} - 2p^{k-1}) / 2$. For $k \ge 5$, this is larger than $p$ for $p$ sufficiently close to $1/2$. To see this, take the derivative at $p = 1/2$: $f'(p) = \frac{d}{dp} \left[(2p)^k/2 + k(1-2p)((2p)^{k-1} - 2p^{k-1}) / 2 \right]$. We then find $f'(1/2) = 2^{2-k}k$, which is smaller than $1$ for $k \ge 5$. Thus, $p$ will converge to $1/2$ by the monotone convergence theorem.
\end{proof}

\psnes*
\begin{proof} We show in each case than no deviation is beneficial.
\begin{enumerate}
  \item If all $k$ candidates are at $1/2$, then the winner is chosen uniformly from the leftmost and rightmost candidate at $1/2$, who each get vote share $1/2$. If any one candidate moves to some point away from $1/2$, they get vote share strictly less than $1/2$, while the middle candidate opposite them gets vote share $1/2$ and wins. Thus, no candidate can benefit by deviating.
    \item Since $k \ge 4$, both points $x$ and $1-x$ have at least two candidates. The candidates who end up being the outermost at $x$ and $1-x$ each get vote share $x$, while the innermost candidates get vote share $1/2 - x$, which is strictly smaller since $x > 1/4$. Any candidate who moves towards the edge gets vote share strictly less than $x$ and loses to the other side. Any candidate who moves into $(x, 1-x)$ gets vote share $1/2 - x$ and loses. Finally, no candidate benefits by moving from $x$ to $1-x$ (or vice-versa), since they will always be in a lottery to be the outermost which has at least as many candidates as their original point (even if $k$ is odd and a deviant moves from the more populated point). \label{item:two-spike-eq}
    \item Since $k \ge 5$, both points $1/4$ and $3/4$ have at least two candidates. There is a three-way tie with vote share $1/4$ between the leftmost candidate, the rightmost candidate, and the one at 1/2---the inner candidates at $1/4$ and $3/4$ get vote share strictly less than $1/4$. As in the previous case, every candidate certainly loses if they move left of $1/4$ or right of $3/4$. The side candidates also lose if they move into $(1/4, 3/4)$. As before, there is no benefit to switching from $1/4$ to $3/4$ or vice-versa. Finally, the candidate at $1/2$ only shrinks their win probability by moving to $1/4$ or $3/4$ (and worsens their margin against a competitor by moving to any other point in $(1/4, 3/4)$). Thus, no candidate benefits by deviating.  \label{item:two-spike-eq-with-center}

    \item Every candidate gets vote share $1/k$ and has a chance to win. If any candidate moves, their partner will get vote share more than $1/k$ and the deviant will still have vote share at most $1/k$, so no one can deviate beneficially. 
  \end{enumerate}
\end{proof}

\begin{lemma}
\label{lemma:psne}
  Any PSNE with uniform voters, complete plurality mixmizing candidates, and left--right tie-breaking must satisfy the following properties:
    \begin{enumerate}[(a)]
    \item Any point occupied by one candidate cannot be between another occupied point and a boundary.
    \item Any point with at least three candidates must be adjacent to a boundary.
    \item Any point with two candidates not adjacent to a boundary must have the same vote share on both sides.
    \item In any two-point equilibrium, the points must be equidistant from $1/2$. 
  \end{enumerate} 
\end{lemma}
\begin{proof}
  \begin{enumerate}[(a)]
    \item Otherwise, the candidate can move away from the boundary to increase their vote share and decrease an opponent's vote share.
    \item Otherwise, one of the candidates could move distance $\epsilon$ either to the right or left of the point to guarantee the maximum possible vote share (instead of having  probability $<1/3$ of being on that side). This only decreases other vote shares---except the new left- or rightmost candidate created, which only has vote share $\epsilon / 2$ (note this requires at least three candidates; with only two, moving increases the vote share of the partner). When adjacent to a boundary, this doesn't work---moving $\epsilon$ towards a boundary would decrease the vote share achieved (even if it's guaranteed), which could create a plurality loss, as in Equilibria~\ref{item:two-spike-eq} and~\ref{item:two-spike-eq-with-center} from \Cref{thm:psnes}. 
    \item If not, then one of the candidates can move $\epsilon$ towards the side with higher vote share to guarantee it. For small enough $\epsilon$, the deviant will have higher vote share than their former partner. This also decreases the vote share of the bordering candidates the deviant moved towards. Thus, this either increases the plurality win probability of the deviant or at least decreases the expected margin against the winner.
    \item Suppose not, and call the points $x$ and $y$. Assume without loss of generality that $x < y$ and $x < 1-y$. Let $z = (y - x) / 2$ be the vote share inner candidates get. If $z \ge y$, then a candidate at $x$ can move right by $\epsilon$ to improve their winning chances, getting certain vote share $z$ rather than a chance at it. If $z < 1-y$, then a candidate at $y$ can move right by $\epsilon$ to guarantee a win. Thus the points must be equidistant from $1/2$.
  \end{enumerate}
\end{proof}

\begin{theorem}
\label{thm:small-k-psnes}
 The following is a complete list of the PSNEs with uniform voters,complete plurality maximizing candidates, and left--right tie-breaking for small $k$:
  \begin{enumerate}[leftmargin=4\parindent,labelsep=4pt]
    \item[$k=2$:] $(1/2, 1/2)$
    
    \item[$k=3$:] $(1/2, 1/2, 1/2)$
    
    \item[$k=4$:]
    \begin{enumerate}[leftmargin=15pt,labelsep=4pt]
    \item $(1/2, 1/2, 1/2, 1/2)$
    \item $(1/4, 1/4, 3/4, 3/4)$
    \item $(x, x, 1-x, 1-x)$, for any $x \in (1/4, 1/2)$
    \end{enumerate} 
    
    \item[$k=5$:]
    \begin{enumerate}[leftmargin=15pt,labelsep=4pt]
    \item $(1/2, 1/2, 1/2, 1/2, 1/2)$
    \item $(1/4, 1/4, 1/2, 3/4, 3/4)$
    \item $(x, x, 1-x, 1-x, 1-x)$, for any $x \in (1/4, 1/2)$
    \item $(x, x, x, 1-x, 1-x)$, for any $x \in (1/4, 1/2)$.
    \end{enumerate} 
  \end{enumerate}
  
\end{theorem}

\begin{proof}
  We know by \Cref{thm:psnes} that these are all Nash equilibria, so we only need to show no other equilibria exist.
   \begin{enumerate}[leftmargin=3\parindent]
    \item[$k=2$:] If a candidate is at a point other than $1/2$, then they can move to $1/2$ and do strictly better (regardless of their opponent's position), so no other equilibrium is possible.
    \item[$k=3$:] We know no point with one candidate can be adjacent to a boundary in equilibrium by \Cref{lemma:psne}. So all candidates must be at the same point. If that point is anything other than $1/2$, it would not be an equilibrium, so $(1/2, 1/2, 1/2)$ must be the unique equilibrium.
    \item[$k=4$:] There is no way to have a single-candidate point not adjacent to a boundary, since no partition of 4 that includes a 1 has two numbers larger than 1 to flank the single-candidate point. Thus, any equilibrium either has two points with two candidates each or one point with all four candidates. The latter type of equilibrium must be at $1/2$, so we only need to characterize the two-point equilibria. 
    
    We know by \Cref{lemma:psne} that in two-point equilibria, the points must be  equidistant from $1/2$ and so can be written as $x$ and $1-x$. Now, we can show that we must have $x \in [1/4, 1/2)$. If $x < 1/4$, then a candidate at $x$ can move right by $\epsilon$ to guarantee the winning inner vote share rather than a $1/2$ chance at it. Thus, the only two point equilibria are those claimed.
    
    \item[$k=5$:] A single-point equilibrium must be at $1/2$. 
    
    A two-point equilibrium cannot be a 1--4 split since the lone candidate would be adjacent to a boundary, so any two-point equilibrium must be a 2--3 split. By \Cref{lemma:psne}, the points must be equidistant from $1/2$, so call them $x$ and $1-x$. We cannot have $x < 1/4$, or else a candidate at $x$ would move right by $\epsilon$ to guarantee a winning vote share. Unlike for $k=4$, we also cannot have $x = 1/4$. If we did, consider the point with 3 candidates. One of them could move to $1/2$ to guarantee vote share $1/4$, which would be tied for the winning share, whereas they only had a $2/3$ chance of getting that vote share before. Thus the only two-point equilibria are those claimed.
    
    We cannot have four- or five-point equilibria, since we would then be forced to place single-candidate points adjacent to the boundary. However, we can have a three-point equilibrium with a 2-1-2 split (a 3-1-1 is impossible for the same boundary reason). So we only need to show that the claimed 2-1-2 equilibrium is the only one. First, the lone candidate must be at the midpoint of the two outer points to optimize its most competitive margin. Next, we'll show the outer points must be equidistant from the boundaries. Suppose not: say the outer points are $x$ and $y$ with $x < 1-y$. If the inner vote share at $x$ ($(y - x)/4$) is smaller than $x$, then a candidate at $x$ has no chance of winning. But by moving to $x- \epsilon$ for some small $\epsilon$, they can guarantee the larger vote share and reduce their expected losing margin. If the inner vote share at $x$ is larger than the outer vote share, then a candidate at $x$ loses to the lone inner candidate; but again, they can move to $x + \epsilon$ improve their expected losing margin. The only remaining option is that the inner and outer shares at $x$ are equal (so the inner share is $x$ and the middle candidate gets vote share $2x$). In that case, consider subcases based on $1-y$. If $1-y < 2x$, then the middle candidate always wins. Since $1-y > x$, a candidate at $y$ can move to $y + \epsilon$ to reduce their expected losing margin against the middle candidate. If $1-y > 2x$, then a candidate at $y$ can move to $y + \epsilon$ to guarantee a win rather than a $1/2$ chance. If $1-y = 2x$, then a candidate at $x$ can move to $y$, giving it a chance to enter the winning lottery for vote share $2x$ (note that the candidate they leave behind at $x$ now also gets vote share $2x$).  
    
    Now that we know the outer points are equidistant from the boundaries, the middle candidate must then be at $1/2$. We can now show that the only possible outer points $x$ and $1-x$ are given by $x = 1/4$. If $x > 1/4$, then the middle candidate cannot win; but they could move to $x$ to join a lottery for the winning vote share. If $x <1/4$, then a candidate at $x$ cannot win. If the inner vote share at $x$ is larger than $x$, a candidate at $x$ can move to $x+\epsilon$ to reduce their expected losing margin against the middle candidate. Symmetrically, if $x$ is larger than the inner vote share, then a candidate at $x$ can move to $x - \epsilon$ to reduce their expected losing margin. Finally, consider the case where the inner and outer vote shares are equal ($x = 1/6$). A candidate at $x$ can move into $(1/6, 1/2)$, keeping the same vote share $1/6$ while reducing the vote share of the winning candidate at $1/2$, thus improving their losing margin. Therefore, the only three-point equilibrium is the one claimed with $x = 1/4$. 
    \end{enumerate}
\end{proof}


\section{Formal definitions of variants}\label{app:variants}
To handle non-uniform voter distributions, we define $\plurality_{V}(x_1, \dots, x_k)$ to be the position of the plurality winner among $x_1, \dots, x_k$ if the voter distribution is $V$.
\begin{definition}
Given an initial candidate distribution $F_0$ and a candidate count $k$, and a distribution of voters $V$, the \emph{replicator dynamics for candidate positioning with voter distribution $V$} are, for all $t > 0$,
  \begin{align*}
  &F_{k, t}(x) = \Pr( \plurality_{V}(X_{1,t},\dots, X_{k,t}) \le x), \\
  & X_{i,t}\sim F_{k, t-1},\ \forall i = 1, \dots, k.
  \end{align*} 
\end{definition} 

\begin{definition}
Given an initial candidate distribution $F_0$, a candidate count $k$, and $m$ generations of memory, the \emph{replicator dynamics for candidate positioning with $m$ generations of memory} are, for all $t > 0$,
  \begin{align*}
  &F_{k, t}(x) = \Pr( \plurality(X_{1,t},\dots, X_{k,t}) \le x),\\
  & X_{i,t}\sim \begin{cases}
     F_{k, t-1} &\text{w.p.\ $\frac{1}{m}$}\\
     ...\\
     F_{k, t-m} &\text{w.p.\ $\frac{1}{m}$}.\\
  \end{cases}
  \end{align*} 
\end{definition} 

\begin{definition}
Given an initial candidate distribution $F_0$, a candidate count $k$, and a variance $\sigma^2 \in [0, 1]$, the \emph{replicator dynamics for candidate positioning with $\sigma^2$-perturbation noise} are, for all $t > 0$,
  \begin{align*}
  &F_{k, t}(x) = \Pr( \plurality(X_{1,t},\dots, X_{k,t}) \le x),\\
  & X_{i,t}\sim \min(1, \max(0, F_{k, t-1} + \mathcal{N}(0,\sigma^2))),\ \forall i = 1, \dots, k.
  \end{align*} 
\end{definition}

\begin{definition}
Given an initial candidate distribution $F_0$, and candidate count proportions $p_2, p_3, \dots, p_{k_\text{max}}$, the \emph{replicator dynamics for candidate positioning with variable candidate counts} are, for all $t > 0$,
  \begin{align*}
  &F_{t}(x) = \sum_{k = 2}^{k_\text{max}} p_k \cdot \Pr( \plurality(X_{1,t},\dots, X_{k,t}) \le x),\\
  & X_{i,t}\sim F_{t-1}.
  \end{align*} 
\end{definition}

Let $F^{(i)}_{k, t}$ denote the distribution of the $i$-th place candidate generation $t$ with $k$ candidates per election, where $i \leq k$. We define $F^{(i)}_{k, 0} = F_0$ for all $k$ and all $i$, although we typically write $F_0$ since the initial distribution does not depend on $k$. Under this notation $F^{(1)}_{k, t} = F_{k, t}$ where $F_{k, t}$ is the CDF of the winner distribution.  Then, 
\begin{definition}
Given an initial candidate distribution $F_0$, a candidate count $k$, and $h \leq k$, the \emph{replicator dynamics for candidate positioning with top-$h$ copying} are, for all $t > 0$,
  \begin{align*}
  &F_{k, t}(x) = \Pr( \plurality(X_{1,t},\dots, X_{k,t}) \le x),\\
  & X_{i,t}\sim \begin{cases}
     F^{(1)}_{k, t} &\text{w.p.\ $\frac{1}{h}$}\\
     ...\\
     F^{(h)}_{k, t} &\text{w.p.\ $\frac{1}{h}$}.\\ 
  \end{cases}
  \end{align*} 
\end{definition}

\end{document}